\documentclass[english,draftcls,onecolumn]{IEEEtran}
\usepackage[T1]{fontenc}
\usepackage{amsmath,amsfonts,amssymb}
\usepackage{graphicx}
\usepackage{subfigure}
\usepackage{babel}

\newtheorem{df}{Definition}
\newtheorem{thm}{Theorem}[section]
\newtheorem{prop}{Proposition}[section]
\newtheorem{lemma}{Lemma}[section]

\newtheorem{cor}{Corollary}[section]
\newtheorem{conjecture}{Conjecture}

\newcommand{\newfigure}[4] 
{
\begin{figure}[h]
\begin{center}
\includegraphics*[scale=#4]{#1}
\caption{#2}\label{#3}
\end{center}
\end{figure}
} 

\newcommand{\seq}[3]{{#1_{#2}, \ldots, #1_{#3}}}
\newcommand{\set}[1]{{\mathcal #1}}

\newcommand{\nge}[1]{\stackrel{#1}{\ge}}
\newcommand{\nequal}[1]{\stackrel{#1}{=}}
\def\define{{\triangleq}}

\def\routing{{\varrho}}
\newcommand{\Real}{{\mathbb{R}}}

\def\graph{{\mathsf  G}}
\def\nodes{{\mathcal V}}
\def\edges{{\mathcal E}}
\def\field{{GF}}
\def\sessions{{\mathcal S}} 
\def\destinationLocation{{D}}
\def\sourceLocation{{L}}

\def\multicastRequirement{{\mathsf M}}
\newcommand{\tail}[1]{{\sf tail}(#1)}
\newcommand{\head}[1]{{\sf head}(#1)}
\def\incoming{{\sf in }}

\def\networkCoding{{\Phi}}
\def\networkcoding{{\phi}}
\def\sourceLocation{{O}}
\def\destinationLocation{{D}}
\def\multicastRequirement{{\mathsf M}}

\def\inputRate{{\lambda}}
\def\edgeRate{{\omega}}
\def\edgeRV{{Y}}   

\def\cone{{\overline{\rm con}}}
\def\la{{\langle}}
\def\ra{{\rangle}}
\def\N{{\cal N}}

\def\A{\mathcal A}

\def\reals{\mathbb R}

\def\support{{\mathsf {SP}}}
\def\major{\mathsf{CL}}

\def\proj{{\text{proj}}}
\def\problem{{\mathsf P}}
\def\indep{{\sf I}}
\def\net{{\sf T}}
\def\de{{\sf D}}
\def\sec{{\sf S}}
\def\defined{{\triangleq}}
\newcommand{\zach}{$0$-achievable}

\newcommand{\vach}{$\epsilon$-achievable}

\begin{document}
\title{Network Coding Capacity Regions via Entropy Functions}
\author{Terence H. Chan and Alex Grant}
\maketitle\thispagestyle{empty}

\begin{abstract}
  In this paper, we use entropy functions to characterise the set of
  rate-capacity tuples achievable with either zero decoding error, or
  vanishing decoding error, for general network coding problems. We
  show that when sources are colocated, the outer bound obtained by
  Yeung, \emph{A First Course in Information Theory}, Section 15.5
  (2002) is tight and the sets of zero-error achievable and
  vanishing-error achievable rate-capacity tuples are the same. We
  also characterise the set of zero-error and vanishing-error
  achievable rate capacity tuples for network coding problems subject
  to linear encoding constraints, routing constraints (where some or
  all nodes can only perform routing) and secrecy
  constraints. Finally, we show that even for apparently simple
  networks, design of optimal codes may be difficult. In particular,
  we prove that for the incremental multicast problem and for the
  single-source secure network coding problem, characterisation of the
  achievable set is very hard and linear network codes may not be
  optimal.
\end{abstract}

\section{Introduction}

Determining network coding capacity regions (the set of link
capacities and source rates admitting a network coding solution for a
given multicast) is a fundamental problem in information theory and
communications.  Recently, the network coding capacity region for
general network coding problems was implicitly determined using
entropy functions~\cite{Yan2007The-Capacity}\footnote{We assume the
  reader is familiar with polymatroids, entropy functions and
  representable functions. We review these and other related concepts
  in Section~\ref{sec:background} and as needed throughout the
  paper.}. This characterisation has a similar (but slightly more
complicated) form to the outer bound in~\cite[Section
15.5]{Yeung02first}, which is expressed in terms of almost entropic
functions. Explicit characterisation of the set of entropy functions
(or its closure) is however a very difficult open problem (for
instance, it is known that this set is not
polyhedral~\cite{Matus07infinitely}).

It is therefore natural to wonder
whether there might be a simpler, explicit characterisation of network
coding capacity regions which somehow avoid the use of entropy
functions. However, these two problems are inextricably
linked, and in general, determining network
coding capacity regions is as hard as finding the set of all entropy
functions, or equivalently, determining all information
inequalities~\cite{Chan2008Dualities}.

One approach to avoid these intrinsic challenges of the general case is
to seek special cases, or specific classes of networks for which an
explicit, computable solution is possible. To date, only a few such
special cases have been found.  One notable example is where a single
source data stream is unicast to multiple destinations.  In this case,
the capacity region is characterised by graph-theoretic maximal
flow/minimum cut bounds, and linear codes are
optimal~\cite{Li.Yeung.ea03linear}. As second example is a secure
network coding problem when all links have equal capacities and the
eavesdropper's capability is only limited by the total number of links
it can wiretap. In this case, the minimum cut bound is
also tight~\cite{cai2002snc}.

Another approach is to develop computable bounds on
the capacity region. Relaxation from entropy functions to polymatroids
yields the so-called linear programming bound~\cite[Section
15.6]{Yeung02first}. Although this bound is explicit, both the number
of variables and the number of constraints increase exponentially
with the number of links in the network, making the bound
computationally infeasible even for modest networks.
Other works such as
\cite{Harvey2006-On-the-capacity,Kramer2006Edge-cut,Thakor2009Network}
aim to obtain useful outer bounds with computationally efficient
algorithms for their evaluation. In fact, it can be shown that all of
the bounds obtained in those works are relaxations of the linear
programming bound from \cite[Section 15.6]{Yeung02first}.

This paper extends \cite{Yan2007The-Capacity} in several aspects.  In
both~\cite{Yan2007The-Capacity} and \cite{Yeung02first}, vanishing
decoder error probabilities are allowed.  In Section \ref{sec:part1},
we will extend these results to the case where the decoding error
probability must be exactly zero. We also prove that when all the
sources are colocated, the outer bound~\cite{Yeung02first} is in fact
zero-error achievable and tight. For the general non-colocated source
case, we show that tightness of the outer bound reduces to a question
of whether or not the addition of a zero-rate link can change the
capacity region of a particular network that we derive from the
original network. This leads us to conjecture tightness of the bound
in general.

The existing capacity result~\cite{Yan2007The-Capacity} does not place
any constraints on the operation of intermediate nodes, allowing
arbitrary network coding operations. We further extend this
entropy-function based approach to three different
practically-motivated cases where additional constraints are placed on
the network: (a) all nodes use linear encoding, (b) all, or some nodes
can only perform routing, and (c) we desire secrecy in the presence of
an eavesdropper.

In Section \ref{sec:part2}, we consider the case where only linear
network codes are allowed. We prove equivalence of zero-error and
vanishing-error achievability, and characterise the linear network
coding capacity region using representable functions.

When comparing the performance of network coding to routing-only
networks (where nodes can only store and forward received packets), it
may be useful to have a capacity characterisation for routing in terms
of entropy functions. In Section \ref{sec:part3} we introduce almost
atomic functions which provide just such a characterisation. We go on
to consider heterogeneous networks, containing both network coding
nodes and routing nodes, and show how to obtain an entropy function
characterisation of the capacity region.

In Section \ref{sec:part4}, we impose secrecy constraints, where we
assume the presence of eavesdropper who has access to certain links
and desires to decode particular sets of sources. The objective is to
design a transmission scheme such that the eavesdropper remains
ignorant of the source messages. We will once again characterise the
resulting general secure network coding capacity region via
representable functions.

Finally, in Section \ref{sec:challenges}, we will consider two very
simple network coding problems where despite the apparent simplicity
of the setup, characterisation of the capacity region turns out to be
extremely difficult, and linear codes may be suboptimal.  The first is
incremental multicast, where the sources and sinks are ordered such
that sink $i$ demands sources $1,2,\dots, i$.  The second example is
secure unicast of a single source. This demonstrates that the
seemingly innocuous addition of a security constraint loosens the
minimum cut bound~\cite{Li.Yeung.ea03linear}. Similarly we see that
the min-cut result from~\cite{cai2002snc} does not hold even for this
simple case.
 
\emph{Notation:} $\reals$ is the set of all real numbers and
$\reals_{0} $ is the set of all nonnegative real numbers. Random
variables will be denoted by uppercase roman letters $X$ and sets will
be denoted using uppercase script $\set{X}$. The power set $2^\set{X}$
is the set of all subsets of $\set{X}$. For a discrete random variable
$X$ taking values in the set (or alphabet) $\set{X}$, its
\emph{support} $\support(X)$ is
\begin{align*}
  \support(X) & \defined \{x\in\set{X}: \: \Pr(X=x) > 0 \} .
\end{align*} 
Realisations of a random variable will typically denoted via lowercase
$x$. 

For sets $\{X_1,X_2,\dots,X_n\}$ and
$\set{S}\subseteq\{1,2,\dots,n\}$, the subscript notation
$X_{\set{S}}$ will mean $\{X_i, i\in\set{S}\}$. Where it will cause no
confusion, set notation braces will be omitted from singletons and
union will be denoted by juxtaposition. Thus
$\set{A}\cup\set{B}\cup\{i\}$ can be written $\set{A}\set{B}i$ and so
on. Ordered tuples will be denoted 
\[
(x(i), i=1,2,\dots,n) = (x_1,x_2,\dots,x_n).
\]

\def\bfh{h}
\def\calH{H}
\def\myindex{myindex}

\section{Background}\label{sec:background}
In this section we provide the formal problem definition for
transmission of information in networks consisting of error-free
broadcast links. This includes representation of such networks as
hypergraphs, the notions of a multicast connection requirement, network
codes and zero-error or vanishing-error achievability. We then
review existing results on characterisation of the network coding
capacity via the use of entropy functions. In this section, we do not
impose any additional constraints (such as linearity or security)
beyond zero- or vanishing-decoding-error probability.

\subsection{Unconstrained Network Coding for Broadcast Networks}
We represent a communication network by a directed hypergraph
$\graph=(\nodes, \edges)$.  The set of nodes
\begin{equation*}
  \nodes = \left\{\seq V1{|\nodes|}\right\}
\end{equation*}
and the set of hyperedges 
\begin{equation*}
  \edges = \left\{\seq E1{|\edges|}\right\}
\end{equation*}
respectively model the set of communication nodes and error-free
broadcast links.  In particular, each hyperedge $e\in\edges$ is
defined by a pair $(\tail e, \head e)$, where $\tail{e} \in \nodes$ is
the transmit node and $\head{e} \subseteq \nodes$ is the set of nodes
which receive identical error-free transmissions from $\tail{e}$.
When $\head{e}$ is a singleton, $e$ models an ordinary point-to-point
link.

We assume that the network is free of directed cycles (a nonempty
sequence of links $\{f_{1}, \ldots , f_{k}\}$ such that $\tail{f_{i}}
\in \head{f_{i-1}}$ for $i=2, \ldots, k$ and $\tail{f_{1}} \in
\head{f_{k}}$).

\begin{df}[Connection Constraint]\label{df:cc}
  For a given communication network $\graph$, a \emph{connection
    constraint} $\multicastRequirement$ is a tuple $(\sessions,
  \sourceLocation, \destinationLocation)$, where $\sessions$ indexes
  the sources, $\sourceLocation:\sessions\mapsto 2^\nodes$ specifies
  the source locations and
  $\destinationLocation:\sessions\mapsto 2^\nodes$ specifies the sink
  nodes. Unless specified otherwise, we let
  \[
  \sessions=\left\{S_{1},\ldots, S_{|\sessions|}\right\}
  \]
  be the index set of $|\sessions|$ independent sources. Source
  $s\in\sessions$ is available at every node in
  $\sourceLocation(s)\subseteq\nodes$.  Note that in general each
  source can be available to more than one network node. The sink
  nodes $\destinationLocation(s) \subseteq \nodes$ are nodes where source
  $s$  should be reconstructed according to some desired error criteria.
\end{df}

It is conceptually useful to imagine each source $s$ as a message sent
along an imaginary source edge, which for simplicity will also be
labelled $s$. In this case, we can use the notation $\head{s}$ to
denote $\sourceLocation(s)$.  For any  $e \in\edges$ and $u\in
\nodes$, we define
\begin{align}
\incoming(e) & \triangleq 
  \{ f\in\sessions \cup \edges:     \tail{e}\in \head{f}\} \\
  \incoming(u) & \triangleq 
   \{ f\in \sessions\cup \edges:  u \in \head{f}\} .
\end{align}
In other words, $\incoming (\cdot)$ is the set of incoming edges
(including the imaginary source edges).  

Solution of the network coding problem
$\problem=(\graph,\multicastRequirement)$ requires a transmission
scheme allowing source $s$ to be reliably reconstructed at the sink
nodes $\destinationLocation(s)$. We have not yet specified the
transmission capacity of each hyperedge, or the rate of each
source. Characterisation of the capacity region for network coding
means determination of the combinations of source rates and hyperedge
capacities which admit a network coding solution. Before we can
proceed, we need to formalise what we mean by \emph{network code},
\emph{reliable} and \emph{rate-capacity tuple}.
\begin{df}[Network Code]\label{df:nc}
  A network code 
  \begin{equation}\label{eq:df2netcode}
  \networkCoding \triangleq \{\networkcoding_e :\: e \in \edges \}
  \end{equation}
  for the problem
  $\problem=(\graph,\multicastRequirement)$ is a set of
  local encoding functions
  \[
  \networkcoding_e : \prod_{f \in \incoming(e)} \set{Y}_{f} \mapsto \set{Y}_{e}.
  \]
  where $\set{Y}_{s}$ is the alphabet of source $s\in\sessions$ and
  $\set{Y}_{e}$ is the alphabet for messages transmitted on
  hyperedge $e\in\edges$.
\end{df}
Each network code induces a set of  random variables
 \begin{align}\label{eq:3}
   \{Y_{f}, f\in\sessions\cup\edges \}.
 \end{align}
 as follows:  
 \begin{enumerate}
 \item 
  $\{Y_s, s\in\sessions\}$ is a set of mutually independent random variables, each of which is uniformly
 distributed over its support and denotes  a message generated by a source.  
 
 \item For each $e\in\edges$,  
\begin{align}\label{eq:df:nc}
  \edgeRV_e = \networkcoding_e  (\edgeRV_f : f \in \incoming(e) ) 
\end{align}
and denotes the message transmitted on hyperedge $e\in\edges$.  
\end{enumerate}
If the set of
random variables induced by a network code is given, then the local
encoding functions \eqref{eq:df2netcode} are determined with
probability one. In other words, if
\[
\Pr\left(Y_{f} = y_{f}, f\in \incoming(e)\right) > 0,
\]
then for all $y_{e} \neq  \networkcoding_e  \left(y_f : f \in \incoming(e) \right)$,  
\[
\Pr\left(Y_{e}=y_{e} \mid Y_{f} = y_{f}, f\in \incoming(e)\right) =0.
\] 
For this reason, we will often specify a network code by its set of
induced random variables.

The following lemma follows directly from the above definitions and gives a
necessary and sufficient condition under which a set of random
variables is induced by a network code.
\begin{lemma} \label{lemma:networkcodedef} A set of random variables
  $\{Y_{f}, f\in\sessions\cup\edges \}$ defines a network code, with
  respect to a network coding problem $\problem$, if and only if
\begin{enumerate}
\item $\edgeRV_{s}$ is uniformly distributed over its support for all $s\in\sessions$.
\item $H\left(Y_s, s\in\sessions\right)  = \sum_{s\in\sessions} H\left(Y_s\right)$.
\item $H\left(\edgeRV_e \mid \edgeRV_f : f \in \incoming(e)    \right)
  = 0$ for all $ e\in\edges$.  
\end{enumerate}
 \end{lemma}
Conditions 2) and 3) are due to the mutual independence of the sources
and the deterministic encoding constraints.

\begin{df}[Rate-Capacity Tuples]
  For a network coding problem $\problem$ let
  \[
  \chi(\problem) \: \defined \: \Real_{0}^{|\sessions|} \times
  \Real_{0}^{|\edges|}.
  \]
  be the set of all \emph{rate-capacity tuples}
  \[
  (\inputRate,\edgeRate)= (\inputRate(s) : \: s \in \sessions, \:
  \edgeRate(e) :\: e \in \edges )
  \]
  for $\problem$.
\end{df}

\begin{df}[Fitness]  \label{df:ratesupport}
A rate-capacity tuple 
$(\inputRate , \edgeRate ) \in\chi(\problem)$
is \emph{fit} for a network code $\{Y_f,f\in\sessions\cup\edges\}$ on
$\problem$ if there exists $c>0$ such that  for all $e\in \edges$
    and $s\in\sessions$,
\begin{align}
  \lambda({s}) & \le c \log |\support(Y_{s})|, \label{eq:6}\\
  \omega({e}) & \ge c \log |\support(Y_{e})|.
\end{align} 
The tuple is \emph{asymptotically fit} for a sequence of network codes
$\{Y_f^n,f\in\sessions\cup\edges\}$ for $n=0,1,\dots$ if there exists a sequence
 $c_n>0$ such that  for all $e\in \edges$
    and $s\in\sessions$,
    \begin{align}
      \lim_{n\to\infty} c_{n }\log |\support(\edgeRV_s^{n})|  & \ge    \inputRate(s), \\
      \lim_{n\to\infty} c_{n}  \log |\support(\edgeRV_e^{n})|   & \le
      \edgeRate(e).    
    \end{align}
\end{df}
Note that fitness does not imply achievability of a rate-capacity
tuple, rather that the tuple is not impossible.  Fitness indicates
that (up no normalisation) each individual source rate is not too
large to be achieved by the corresponding source variable with the
given alphabet size, and that each hyperedge capacity is large enough
to carry the corresponding edge variable regardless of particular
distribution.
 
\begin{df}[Zero-error Achievable Rate-Capacity Tuples]\label{df:admissible} 
  A  rate capacity tuple
  \[
  (\inputRate , \edgeRate ) = (\inputRate(s) : \: s \in \sessions, \:
  \edgeRate(e) :\: e \in \edges )\] is called \emph{zero-error
    achievable}, or \emph{ \zach } if there exists a sequence of network codes
  $\networkCoding^{n}, n=1,2,\dots$ and corresponding induced random
  variables $\{\edgeRV_f^{n}: f \: \in \edges \cup \sessions\}$ such
  that
  \begin{enumerate}
  \item $(\inputRate , \edgeRate )$ is asymptotically
    fit for $\networkCoding^n$.
  \item for any source $s\in\sessions$ and receiver node $u \in
    \destinationLocation(s)$, the source message $Y_s^{n}$ can be
    uniquely determined from the received messages $\left(\edgeRV_f^{n} :
    f \in \incoming(u) \right)$. In other words, 
    
  \begin{equation}
    H\left( Y_s^{n} \mid \edgeRV_f^{n}, f \in \incoming(u) \right) = 0.\label{eq:decode}
  \end{equation}

  \end{enumerate}
\end{df}

In Definition \ref{df:admissible}, each network code in the sequence
has zero probability of decoding error. Relaxing this criteria to
allow decoding error probability that vanishes in the limit, we have
the following definition.
\begin{df}[Vanishing Error Achievable]\label{df:achievable}
  A rate capacity tuple $(\inputRate , \edgeRate )$ is called
  \emph{vanishing error achievable}, or \vach\ if the tuple is
  asymptotically fit, and 
  \begin{enumerate}
  \item[$2^{\prime}$)] for all $s\in\sessions$ and $u \in
    \destinationLocation(s)$, there exists decoding functions
    $g_{s,u}^{n}$ such that
    \begin{align*}
      \lim_{n\to\infty} \Pr(Y_{s}^{n} \neq g_{s,u}^{n}(\edgeRV_f^{n}:
      f \in \incoming(u) )) = 0 .
    \end{align*}
    In other words, decoding error probabilities vanish
    asymptotically.
  \end{enumerate}
\end{df}

For any subset $\set{R} \subseteq {\chi} (\problem)$, define
$\major(\set{R}) $ as the subset of ${\chi} (\problem)$ containing all
tuples $(\lambda,\omega)$ such that there exists a sequence of
$(\lambda^{n} , \omega^{n}) \in \set{R}$ and positive numbers $c_n$
satisfying
\begin{align}
\lim_{n\to\infty} c_{n}{ \omega^{n}(e)}   \le \omega(e), \\
\lim_{n\to\infty} c_{n}{\lambda^{n} (s)}   \ge \lambda(s).
\end{align}
Clearly, if every tuple in $\set{R}$ is \zach/\vach, then
$\major(\set{R}) $ is also \zach/\vach.
  
The central theme of this paper is the characterisation of \zach\ and
\vach\ regions for network coding via the use of \emph{entropy
  functions}.
\begin{df}[Entropy Function]  
  A set of random variables $\{Y_i, i\in\N\}$ (where $\N$ is some
  index set) induces a real \emph{entropy function} $h:2^\N\mapsto\reals$ such that
  for any $\alpha\subseteq\N$,
  \begin{equation*}
    h\left(\alpha\right) = H\left(Y_i, i\in\alpha\right)
\end{equation*}
is the joint Shannon entropy\footnote{We define $h\left(\alpha\right)=0$ whenever
  $\alpha$ is an empty set.}  of $(Y_i : i\in\alpha)$, which according
to our notational conventions we will also write $H\left(Y_\alpha\right)$. 
\end{df}
 
Let  
\[
\set{H}[\set{N}] \triangleq  \Real^{2^{|\set{N}|}}
\]
be the $2^{|\set{N}|}$-dimensional Euclidean space whose coordinates
are indexed by subsets of $\set{N}$. Thus, any element
$g\in\set{H}[\set{N}]$ has coordinates
$(g(\alpha),\alpha\subseteq\set{N})$. Elements of $\set{H}[\set{N}]$
are called \emph{rank functions}\footnote{This terminology comes from matroid
theory and does not imply that such functions must be defined via
ranks of linear operators -- although such functions are rank functions
by this definition.}.  Clearly, entropy functions are rank functions.
\begin{df}[Entropic Functions]
  A rank function $h\in\set{H}[\N]$ is
  \begin{itemize}
  \item \emph{Entropic} if $h$ is the entropy function of a set of
    $|\N|$ random variables. The set of entropic functions is denoted
    $\Gamma^*(\N)\subset\set{H}[\set{N}]$~\cite{Yeung97framework}. When
    the index set $\N$ for the set of random variables is understood,
    we simply denote the set of entropic functions as $\Gamma^*$.
  \item \emph{Weakly entropic} if  there exists
    $c>0$ such that $c\cdot h$ is entropic.
  \item \emph{Almost entropic} if there exists a sequence of weakly
    entropic functions $h^i$ such that
    $$\lim_{i\rightarrow\infty} h^i = h.$$
    The set of almost entropic functions is $\bar\Gamma^*$.
  \end{itemize}
\end{df}

For any rank function $g\in\set{H}[\set{N}]$, define the notations
\begin{align}
g(\alpha\mid\beta) & \: \defined \: g(\alpha \cup \beta) - g(\beta), \\
g(\alpha \wedge \beta)& \: \defined \: g(\alpha) + g(\beta) - g(\alpha
\cup \beta)  .  
\end{align}
If $g$ is in fact an entropy function induced by random variables
$\{Y_{i}, i\in\N\}$, then $g(\alpha\mid\beta)$ is the usual
conditional entropy $H\left(Y_\alpha \mid Y_\beta\right)$ and
$g(\alpha \wedge \beta)$ is the usual  mutual information $
I(Y_{\alpha} ; Y_{\beta} )$. We avoid the
standard notation $I(\cdot;\cdot)$ since it hides the underlying
entropy function, which will be critical in most of what we do.

The set $\Gamma^*$ plays an important role in information
theory. Characterisation of this set amounts to characterising every
possible information inequality. Thus $\Gamma^*$ essentially fixes the
``laws'' of information theory. However it turns out that $\Gamma^*$
has a very complex structure and an explicit characterisation is still
missing~\cite{Chan2011Recent}.  It has been proved that the closure
$\bar{\Gamma}^*$ is a closed convex cone~\cite{Yeung97framework} and
hence is more analytically manageable than $\Gamma^*$. For many
applications, it is in fact sufficient to consider $\bar{\Gamma}^*$.
However, it was proved in \cite{Zhang.Yeung98characterization} that
when $|\N| \ge 3$,
\[
\Gamma^* \: \neq \:  \bar{\Gamma}^*.
\]

\subsection{Existing Results}\label{sec:existingresulte}
For a given network coding problem
$\problem=(\graph,\multicastRequirement)$, let $\Gamma^{*}(\problem)$
and $\set{H}[\problem]$
respectively denote $\Gamma^{*}(\sessions\cup\edges)$ and
$\set{H}[\sessions\cup\edges]$.  Define the coordinate projection
\[
\proj_{\problem} : \set{H}[\problem] \mapsto {\chi}(\problem)
\]
such that for any $h\in
\set{H}[\problem]$,
\begin{align}
  \proj_{\problem}[h](s) & = h(s) , \: \forall s\in\sessions \\
  \proj_{\problem}[h](e) & = h(e), \: \forall e\in\edges
\end{align}
Similarly, for any subset $\set{R}\subseteq\set{H}[\problem]$, 
\begin{align}
{\proj}_{\problem}[\set{R}] \define \{ \proj_{\problem}[h] : h\in\set{R} \}.
\end{align}
Again, if the underlying network coding problem $\problem$ is understood implicitly,
we will simply use the notations $\proj[h]$ and $\proj[\set{R}]$.


Consider any network coding problem $\problem=(\graph,
\multicastRequirement)$. Define the following subsets of
$\set{H}[\problem]\defined \set{H}[\sessions \cup \edges ]$:
\begin{align}
  \set{C}_{\indep} (\problem)& \define \left\{ h\in\set{H} [\problem]:
    {h}({\set{S}})  = \sum_{s\in\sessions} {h}( s) \right\}, \label{eq:constraint1} \\
\set{C}_{\net} (\problem)& \define 
 \left\{ 
	\begin{array}{l}
\hspace{-0.2cm} h \in\set{H}[\problem]: h\left( s\mid   \incoming(e) \right)  = 0,   \forall  e\in\edges 
 	\end{array}
\right\},\\
\set{C}_{\de} (\problem)& \define \left\{ 
	\begin{array}{l}
\hspace{-0.2cm} h \in\set{H}[\problem]: h\left( s\mid  \incoming(u) \right)  = 0, \\  
  	\hspace{2.5cm} \forall  s\in\sessions, u\in\destinationLocation(s) 
 	\end{array}
\right\}. \label{eq:constraint3}
\end{align}
The above subsets will be denoted by $\set{C}_{\indep} ,
\set{C}_{\net} $ and $\set{C}_{\de} $ respectively if the network
coding problem $\problem$ is understood implicitly. Consider a network
code $\{Y_i, i\in {\sessions}\cup{\edges}\}$ with induced entropy
function $h\in\set{H}[\sessions\cup\edges]$. By
Lemma~\ref{lemma:networkcodedef} we see that $h \in \set{C}_{\indep} $
since the sources are mutually independent, and $h\in \set{C}_{\net}$
due to deterministic transmission through the network. If the network
code is zero-error, $h\in \set{C}_{\de}$ follows from the decodability
constraint~\eqref{eq:decode}.
 
The set of \vach\ rate-capacity tuples can be characterised exactly as
follows~\cite{Yan2007The-Capacity}.
\begin{thm}[Yan, Yeung and Zhang -- \vach\ Region~\cite{Yan2007The-Capacity}]\label{thm:yanbd}
  For a given network coding problem $\problem=(\graph,
  \multicastRequirement)$, a rate-capacity tuple $(\lambda,
  \omega)\in{\chi} (\problem) $ is \vach\ if and only if
  \begin{align}\label{eq:yanbd}
    (\lambda, \omega)  \in \major(  \proj_{\problem} \left[ \cone(
      \Gamma^{*} \cap \set{C}_{\indep} \cap \set{C}_{\net})  \cap
      \set{C}_{\de}  \right]). 
  \end{align}
\end{thm}
 
Inner and outer bounds for the \zach\ region were also investigated
in~\cite{Song.Yeung.ea03zero-error} using a similar framework as
in~\cite{Yan2007The-Capacity,Chan2008Dualities}. However,
\cite{Song.Yeung.ea03zero-error} allowed the use of variable length
coding, where the amount of data traffic on a particular link is
measured as the average number of transmitted bits. 

In contrast, this paper studies the \emph{worst case} scenario where
the amount of traffic is measured by the maximum number of bits
transmitted on  a link (hence, it is sufficient to consider fixed-length
codes).  It is worth pointing out that when decoding error is not
allowed, there is a significant difference between using the average
or the maximum number of transmitted bits. For example, consider a
source $X$ compressed/encoded by an optimal uniquely-decodable code.
The average length of resulting codeword is roughly equal to
$H(X)$. However, for all uniquely-decodable codes, the maximum length
of the encoded codeword must be at least $\log|\support(X)|$, which
can be much greater than $H(X)$ if $X$ is heavily biased.

The following outer bound follows directly from Theorem
\ref{thm:yanbd}, but was proved earlier in~\cite{Yeung02first}
\begin{cor}[Yeung -- Outer Bound~\cite{Yeung02first}]\label{cor:ourbd}
  If a rate-capacity tuple $(\lambda, \omega) $ is \vach, then
\begin{align}\label{eq:ourbd}
(\lambda, \omega)  \in \major(  \proj_{\problem} [ \bar \Gamma^{*}  \cap \set{C}_{\indep} \cap \set{C}_{\net}  \cap \set{C}_{\de} ]).
\end{align}
\end{cor}
This bound \eqref{eq:ourbd} is not necessarily tight, since
$\Gamma^{*}$ is not closed and convex in general. Therefore,
\[
 \cone( \Gamma^{*} \cap \set{C}_{\indep} \cap \set{C}_{\net} )
\]
theoretically may be a proper subset of  
\[
 \cone( \bar \Gamma^{*} \cap \set{C}_{\indep} \cap \set{C}_{\net} )
 \nequal{(a)} \bar \Gamma^{*}  \cap \set{C}_{\indep} \cap
 \set{C}_{\net} 
\]
where $(a)$ follows from that $\bar \Gamma^{*}  $ is a closed and convex cone. 
%

It is clear that if $(\lambda, \omega) \in {\chi}(\problem)$ is \zach,
then it is also \vach\ and hence must satisfy the outer bound
\eqref{eq:ourbd} in Corollary \ref{cor:ourbd}.  In fact, it can be
seen directly that~\eqref{eq:ourbd} must be an outer bound for the set
of \zach\ rate-capacity tuples. Suppose $(\lambda, \omega) \in
\chi(\problem) $ is \zach. Then there exists a sequence of network
codes $\{Y_f^{n}, f\in {\sessions\cup\edges}\}$ with induced entropy
functions $h^n$ and positive constants $c_n$ such that for all
$e\in\edges$ and $s \in\sessions$
\begin{align*}
\lim_{n\to\infty}  c_{n}{H (Y_{e}^{n}) }  & \le  \lim_{n\to\infty} c_{n} {\log | \support(Y_{e}^{n})|} \le \omega({e}),   \\
\lim_{n\to\infty} {c_n}{H (Y_{s}^{n}) }  & = \lim_{n\to\infty} {c_n}{\log |\support( Y_{s}^{n})|}  \ge \lambda({s}), 
\end{align*}
and that $h^{n} \in \set{C}_{\indep}\cap \set{C}_{\net} \cap
\set{C}_{\de} $. Consequently,
\[
(\lambda,\omega) \in \major( \proj_{\problem} [   \Gamma^{*} \cap \set{C}_{\indep} \cap \set{C}_{\net} \cap \set{C}_{\de} ]).
\]
The proof that \eqref{eq:ourbd} is an outer bound for \vach\ tuples is
similar. However, as vanishing error is allowed, Fano's inequality is
invoked to ensure $\lim_{n\to\infty} {c_n} h^{n} \in \set{C}_{\de}$.

In the next section, we deliver our first main result, namely
that~\eqref{eq:ourbd} is tight when the sources are colocated.

\section{Tightness of Yeung's Outer Bound}\label{sec:part1}
The analytical challenges in characterising $\Gamma^*$ (let alone its
intersection with $\set{C}_{\indep}$ and $\set{C}_{\net}$) may render
Theorem~\ref{thm:yanbd} unattractive as a characterisation of the
network coding capacity region. In this section we show that the more
manageable bound
 of Corollary~\ref{cor:ourbd}, which involves the closure of
$\Gamma^*$ is in fact tight when the sources are colocated, a notion
that we make precise below in Definition~\ref{df:colocated}. Our proof
will use quasi-uniform random variables, discussed in~\ref{sec:tools},
which are a valuable tool in proving zero-error results. The proof of
the main result, Theorem~\ref{thm:arbitrarymain} is given
in~\ref{sec:Proof}.

\begin{df}[Colocated sources]\label{df:colocated}
  Consider a network coding problem $\problem = (\graph,
  \multicastRequirement)$.  Its sources are called \emph{colocated}
  if
  \[
  \sourceLocation(s)=\sourceLocation(s'), \quad \forall
  s,s'\in\sessions.
  \]
\end{df}
In other words, if a node has an access to any source $s$, it also has access to all the other sources.
 
\begin{thm}[Colocated sources]\label{thm:arbitrarymain}
  Consider a network coding problem
  $\problem=(\graph,\multicastRequirement)$ with colocated sources
  according to Definition~\ref{df:colocated}.  Then
  \begin{enumerate}
  \item A rate-capacity tuple $(\lambda, \omega)$ is \zach\ if and
    only if it is \vach.
  \item The outer bound in Corollary \ref{cor:ourbd} is tight.
\end{enumerate}
\end{thm} 

We fail to prove the tightness of the outer bound in Corollary
\ref{cor:ourbd} when sources are not colocated. However, we will give
evidence in \ref{sec:general} to support our conjecture that the outer
bound should be tight in general.
     
\subsection{Tools: Quasi-Uniform Random Variables}\label{sec:tools}
Before we prove Theorem \ref{thm:arbitrarymain} in
Section~\ref{sec:Proof}, we introduce key tools and intermediate
results.  In particular, the proof relies on the concept of
quasi-uniform random variables, which are crucial for proving
zero-error results.  

\begin{df}[Quasi-Uniform Random
Variables~\cite{Chan2011Recent}]\label{df:qu} A set of random variables $\{X_{i},
i\in\N \}$ is called \emph{quasi-uniform} if for any subset
$\alpha\subseteq \N$, the random variable $X_{\alpha} \define
(X_{i},i\in\alpha)$ is uniformly distributed over its support, or
equivalently,
\[
H\left(X_{\alpha}\right) = \log |\support(X_{\alpha})|.
\]
\end{df}
 
\begin{lemma}\label{lemma:quw}
  Suppose $\{A,B\}$ is quasi-uniform. Then one can construct a random
  variable $W$ such that
\begin{align*}
H(W) &= H\left(A\mid B\right), \\
H\left(A \mid B, W\right)&=0.
\end{align*}
\end{lemma}
\begin{proof}[Sketch of proof]
  As $\{A,B\}$ is quasi-uniform, it can be proved from
  Definition~\ref{df:qu} that for any $b\in\support(B)$,
  \[
  \Pr(A=a \mid  B=b)=
  \begin{cases}
    2^{-H\left(A\mid B\right)} & \text{ if } \Pr(A=a , B=b) >0 \\
    0 & \text{ otherwise.}
  \end{cases}
  \]
  Assume without loss of generality that 
  \[
  \{q(1,b), \ldots, q(2^{H\left(A\mid B\right)},b) \}
  \]
  is the set of all elements in $\support(A)$ such that
  \[
  \Pr(A=a \mid  B=b) > 0.
  \]
  Let $W$ be a random variable such that for any $(a,b)$ in $
  \support(A,B)$,
  \[
  \Pr(W=w \mid  A=a, B=b) =
  \begin{cases}
    1 & \text{ if } a = q(w,b) \\
    0 &\text{ otherwise.}
  \end{cases}
  \] 
  The lemma can then be verified directly. 
\end{proof}

\df[Quasi-Uniform Rank Functions]\label{df:classification}
A rank function $h \in \set{H}[\N]$  is called  
\begin{itemize}
\item \emph{Quasi-uniform} if $h$ is the entropy function of a set of
  $|\set{N}|$ quasi-uniform random variables.

\item \emph{Weakly quasi-uniform} if there exists $c > 0$ such that $
  c\cdot h$ is quasi-uniform;

\item \emph{Almost quasi-uniform} if there exists a sequence of weakly
  quasi-uniform rank functions $ h^{i} $ such that
  \[
  \lim_{i\to\infty} h^{i} = h.
  \]
\end{itemize}
\enddf      

\begin{lemma}\label{lem:qusum}
  If $h_{1},h_{2}\in\set{H}[\N]$ are quasi-uniform, then their sum,
  defined for all $\set{A}\subseteq\set{N}$ as
  $h_1(\set{A})+h_2(\set{A})$, is also quasi-uniform.
\end{lemma}
\begin{proof}
  Suppose $A_\N$ and $B_\N$ are two independent sets of quasi-uniform
  random variables whose entropy functions are $h_{1}$ and $h_{2}$
  respectively. It is straightforward to construct a new set of
  variables $X_\N$ with entropy function $h_1+h_2$, via 
  \[
  X_{i} = (A_{i}, B_{i}), \quad \forall i \in \N.
  \]
  The lemma follows, since $ X_\N$, and hence $h_1+h_2$, is quasi-uniform.
\end{proof}

For any weakly entropic function $h$, \cite{Chan.Yeung02relation}
explicitly constructed a sequence of weakly quasi-uniform functions
with limit $h$.  It can be verified directly that this sequence of
weakly quasi-uniform functions satisfies the same functional
dependency constraints as $h$. Hence, we have the following
proposition.
\begin{prop}\label{prop:3.supporta}
  For any weakly entropic rank function $ \bfh$, there exists a
  sequence of  quasi-uniform random variables $\{U^{\ell}_{i},
  i\in\N \}$ and positive numbers $c_{\ell}$ such that
  \begin{enumerate}
  \item For any $\alpha \subseteq \N$,
    \begin{align}\label{eq:13}
      \lim_{\ell \to\infty} c_\ell H\left( U^{\ell}_{\alpha}\right) = h\left(\alpha\right).
    \end{align}
  \item If $h\left(k\mid \alpha\right) = 0$, then 
    \begin{align}\label{eq:13b}
      H\left( U^{\ell}_{k} | U^{\ell}_{\alpha}\right)  = 0, \quad \text{for all } \ell.
    \end{align}
    In other words, $h$ is the limit of a sequence of weakly quasi-uniform functions $f^{\ell}$ where 
    \[
    h\left(k\mid\alpha\right) = 0 \implies f^{\ell}(k\mid\alpha) = 0.
    \]
  \end{enumerate}
\end{prop}

In fact, Proposition \ref{prop:3.supporta} remains valid even if $h$ is almost entropic.

\begin{prop}\label{prop:3.support}
  For any almost entropic rank function $\bfh \in \bar\Gamma^*(\N)$,
  there exists a sequence of quasi-uniform random
  variables $\{U^{\ell}_{i}, i\in\N \}$ and positive numbers
  $c_{\ell}$ such that \eqref{eq:13} and \eqref{eq:13b} hold.
\end{prop}
\begin{proof}
  By \cite{Chan.Yeung02relation}, there exists a sequence of
  quasi-uniform random variables $\{U^{\ell}_{i}, i\in\N \}$ and
  positive numbers $c_{\ell}$ such for all $\alpha \subseteq \N$,
  \begin{align} 
    \lim_{\ell \to\infty} c_\ell H\left( U^{\ell}_{\alpha}\right) = h\left(\alpha\right). \label{eq:24}
  \end{align}
  The challenge however is that \eqref{eq:13b} may not hold if $h$ is
  not weakly entropic (we only know that $h$ is the limit of a sequence
  of weakly entropic functions).  In the following, we will show how to modify
  the  $\{U^{\ell}_{i}, i\in\N \}$ such that \eqref{eq:13b} indeed holds.   
  First, notice that  $\{U^{\ell}_{i}, i\in\N  \} $ is quasi-uniform. 
  Hence, for any  $k\in\N$ and $\alpha\subseteq\N$, 
  $\{U^{\ell}_{k}, U^{\ell}_{\alpha} \}$ is quasi-uniform. 
  By Lemma \ref{lemma:quw}, one can construct a random variable $W_{k,\alpha}^{\ell}$ such that  
  \begin{align}
    H\left(W_{k,\alpha}^{\ell}\right) & = H (U_{k}^{\ell} \mid U_{\alpha}^{\ell}), \label{g.eq.25}\\
    H\left(U_{k}^{\ell} \mid U_{\alpha}^{\ell}, W_{k,\alpha}^{\ell}\right) & = 0.\label{g.eq.26}
  \end{align}
  Let 
  \begin{align*}
W_{\Delta}^{\ell} &\defined \{W_{k,\alpha}^{{\ell}}, (k,\alpha) \in
  \Delta\} \\\intertext{where}
    \Delta &\define \{ (k,\alpha) : {\bfh}( {k} \mid {\alpha})=0 ,
    k\in \N , \alpha \subseteq \N \} 
  \end{align*}
 Then
  \begin{align*}
    0 & \le \lim_{\ell \to \infty} c_{\ell}H\left( W_{\Delta}^{\ell} \right) \\
    & \le \lim_{\ell \to \infty} c_{\ell}  \sum_{(k,\alpha) \in \Delta} H\left(W^{\ell}_{k,\alpha}\right) \\
    & =  \sum_{(k,\alpha) \in \Delta}  \lim_{\ell \to \infty} c_{\ell} H (U_{k}^{\ell} \mid  U_{\alpha}^{\ell})  \\
    & = \sum_{(k,\alpha) \in \Delta}  \lim_{\ell \to \infty} c_{\ell}
    ( H (U_{k}^{\ell}, U_{\alpha}^{\ell}) - H ( U_{\alpha}^{\ell})  ) \\ 
    & \nequal{(a)} \sum_{(k,\alpha) \in \Delta} (h\left(k,\alpha\right) - h\left(\alpha\right)) \\
    & =0
  \end{align*}
  where ($a$) follows from \eqref{eq:24}. Consequently,
  \begin{equation}
  \lim_{\ell \to \infty} c_{\ell}H\left( W_{\Delta}^{\ell} \right) =0.\label{eq:WDeltalimit}
\end{equation}
We now construct our new set of random variables, $\{{V}^{\ell}_{i},
i\in\N \}$ by defining
\[
{V}^{\ell}_{i} \: \defined \: (U^{\ell}_{i}, W_{\Delta}^{\ell} ), \quad \forall  i\in \N.
\] 
It is obvious that $H ( {V}_{k}^{\ell} \mid {V}_{\alpha}^{\ell})  =  0$ for all $(k,\alpha) \in \Delta$. 
Let $f^{\ell}$ be the entropy function of $V^\ell_\N$.
 Then by~\eqref{eq:24}, for any $\beta \subseteq \N$,   
\begin{align}
  h\left(\beta\right) & = \lim_{\ell\to\infty} c_{\ell} H\left(U_{\beta}^{\ell}\right)  \label{eq:29}\\
  & \le \lim_{\ell\to\infty} c_{\ell} H\left(U_{\beta}^{\ell},
  W_{\Delta}^{\ell}\right)  \\ 
  & \le \lim_{\ell\to\infty} c_{\ell} H\left(U_{\beta}^{\ell} \right) +
  \lim_{\ell\to\infty} c_{\ell} H\left(W_{\Delta}^{\ell} \right)  \\  
  & \stackrel{(b)}{=} \lim_{\ell\to\infty} c_{\ell} H\left(U_{\beta}^{\ell}\right) \\
  & = h\left(\beta\right). \label{eq:34}
\end{align}
where (b) is by~\eqref{eq:WDeltalimit}.
Consequently, 
 \[
 \lim_{\ell\to\infty} c_{\ell} f^{\ell }(\beta) = \lim_{\ell\to\infty} c_{\ell} H\left(U_{\beta}^{\ell},
  W_{\Delta}^{\ell}\right) = h(\beta), \: \forall \beta\subseteq \N.
 \]
   Since $f^{\ell}$ is weakly
entropic and $f^{\ell}(k\mid\alpha) = 0$ for all $(k,\alpha) \in
\Delta$, we can once again use Proposition \ref{prop:3.supporta} to
construct a sequence of weakly quasi-uniform functions $g^{j}$ such
that $\lim_{j\to\infty } g^{j} = h$ and $g^{j}(k\mid\alpha) = 0$ for all
$(k,\alpha) \in \Delta$. 
\end{proof}

\subsection{Proof for Theorem \ref{thm:arbitrarymain}}\label{sec:Proof}
\def\binparticleEmptyFn{\kappa} 

The first claim of Theorem \ref{thm:arbitrarymain} is that the \vach\
and \zach\ regions are equivalent when sources are colocated.  For any
network coding problem $\problem = (\graph,\multicastRequirement)$, it
is clear that if $(\lambda, \omega) \in {\chi}(\problem)$ is
\zach, then it is also \vach\ and hence must satisfy the
outer bound \eqref{eq:ourbd} in Corollary \ref{cor:ourbd}.  Thus, to
prove Theorem \ref{thm:arbitrarymain}, it suffices to show that for
colocated sources, the rate-capacity tuple $\proj_{\problem}[h] $ is
\zach\ for all $h \in \bar\Gamma^{*} \cap \set{C}_{\indep} \cap
\set{C}_{\net} \cap \set{C}_{\de}$.

Our proof technique is similar to that used
in~\cite{Yan2007The-Capacity}. However, instead of constructing
network codes from strongly typical sequences, we use quasi-uniform
random variables.  Codes constructed from typical sequences admit a
small (but vanishing) error. However, as we shall see, codes
constructed from quasi-uniform random variables can be carefully
designed to ensure zero decoding error probability.

Consider a network coding problem
$\problem=(\graph,\multicastRequirement)$ where all sources are
colocated.  Suppose
\begin{equation*}
h \in \bar\Gamma^{*} \cap \set{C}_{\indep} \cap   \set{C}_{\de} \cap \set{C}_{\net}.
\end{equation*}
Since $h$ is almost entropic, Proposition \ref{prop:3.support} implies
the existence of a sequence of quasi-uniform random variables
\begin{align}\label{eq:37}
  \{ U^{n}_{f}, f\in\sessions \cup \edges \}
\end{align}
 and positive numbers $c_n$ such that
\begin{align}
  \lim_{n\to\infty} c_{n}H\left(U^{n}_{\alpha} \right) & = h\left(\alpha\right), \quad \forall  \alpha \subseteq \sessions \cup \edges  \label{eq:3.35}  \\
  H \left(U^{n}_{e}\mid U^{n}_{\incoming(e)}  \right)   & \nequal{(a)} 0 , \quad \forall e\in\edges \label{eq:15} \\
  H \left(U^{n}_{s}\mid U^{n}_{\incoming(u)} \right)  & \nequal{(b)} 0, \quad \forall  s\in\sessions, u\in\destinationLocation(s) \label{eq:16}
\end{align} 
where $(a)$ is due to $h\in\set{C}_\net$ and $(b)$ is due to $h\in
\set{C}_{\de}$.  Furthermore, by Lemma~\ref{lem:qusum} the sum of
two quasi-uniform rank functions is quasi-uniform. Hence,  we can assume
without loss of generality that
\begin{align}\label{eq:40}
  \lim_{n\to\infty} c_n = 0
\end{align}
and $H(U_{1}^{n})$ grows unbounded. This assumption \eqref{eq:40} will be used in the latter part when we construct a zero-error network code.
 
In the following, for each $n$, we will construct a zero-error network
code $\{Y^{n}_f, f\in {\sessions\cup\edges}\}$ from each set of
quasi-uniform random variables $\{ U^{n}_{f}, f\in\sessions \cup
\edges \}$ such that
\begin{align}
\lim_{n\to\infty} c_{n } H\left(Y^{n}_{s}\right)   & = h(s) \label{eq:17}\\
\lim_{n\to\infty} c_{n}H\left(Y^{n}_{e}\right)& \le h(e) \label{eq:18}
\end{align}
and consequently, $\proj_{\problem}[h]$  is \zach\ and the outer bound is tight.

\subsubsection*{Code construction}

For simplicity of notation, we will drop the superscript $n$ in
\eqref{eq:37} and directly denote the set of quasi-uniform random
variables by
\[
\{U_f, f\in {\sessions\cup\edges}\}.
\]

Suppose first that the $U_{s}, s\in\sessions$ are mutually
independent. The $U_{s}$ are quasi-uniform and hence uniformly
distributed over their support. Thus \eqref{eq:15} holds and Lemma
\ref{lemma:networkcodedef} implies that $\{U_f, f\in
{\sessions\cup\edges}\}$ in fact defines a network code. Furthermore,
by \eqref{eq:16}, the decoding error probability is zero, and Theorem
\ref{thm:arbitrarymain} is proved for this special case of independent
sources.

Unfortunately, $\{ U_{s}, s\in\sessions \}$ need not be mutually
independent in general. To address this problem, we will modify these
variables to satisfy the independency constraint. This is when we require all sources to be colocated.

Since the network $\graph$ is acyclic, repeated application of
\eqref{eq:15} can be used to prove  that
\[
H\left(U_{f}\mid U_{\sessions}\right) = 0, \:\forall  f\in \sessions\cup\edges.
\]
 Hence, for any $e\in\edges$, there exists functions $G_{e}$ such that
 \[
U_{e} = G_{e}(U_\sessions).  
 \]
Similarly, for any $e\in\edges$, $s\in\sessions$ and $v\in \destinationLocation(s)$,
there exists functions $g_{e}$ and $g_{s,v}$ such that for all
$u_\sessions \in \support(U_\sessions)$,  
\begin{align}
G_{e}(u_\sessions)  & \nequal{(a)} g_{e}(G_{f}(u_\sessions), f \in \incoming(e)) , \quad \forall e\in\edges  \label{eq:3.40} \\
U_{s} & \nequal{(b)} g_{s,v} (G_{f}(u_\sessions), f \in \incoming(v) ), \nonumber \\ 
& \hspace{2.8cm} \quad \forall s\in\sessions, \: v \in \destinationLocation(s). \label{eq:3.41}
\end{align}
where $(a)$ follows from \eqref{eq:15} and $(b)$ from \eqref{eq:16}.

\begin{df}[Partition]
  Let $\set{A}=\{1,2,\dots,|\set{A}|\}$ be an index set.  A
  \emph{partition} of a set $\set{X}$ into $|\set{A}|$ partitions is a
  mapping $$\Xi:\set{A}\mapsto 2^{\set{X}}$$ where
  $\Xi(i)\subseteq \set{X}$ is the set of elements in partition
  $i\in\set{A}$.
If the $\Xi(i)$ are disjoint then the partition $\Xi$ is called
  \emph{disjoint}.
\end{df}
\begin{df}[Regular Partition Set]\label{df:regularpartition}
  Let $U_s,s\in\sessions$ be  random variables with supports $\support(U_s)$. For
  $s\in\sessions$, define the index sets $ \set{A}_{s} \define \{1,
  \ldots, |\set{A}_{s}| \}$, where $|\set{A}_{s}| \le
  |\support(U_{s})|$, and let
\begin{equation*}
\Xi_s: \set{A}_s \mapsto 2^{\support(U_s)}
\end{equation*}
be a disjoint partition of $\support(U_s)$ into $|\set{A}_{s}|$
subsets. We call the set of partitions $\{ \Xi_{s}, s\in\sessions\}$ a
\emph{regular partition set} for $\{U_{s},s\in\sessions\}$ if and only if for all
$b_\sessions\in\set{A}_{\sessions}$
  \begin{align}\label{eq:47}
  \support(U_{\sessions}) \cap \prod_{s\in \set{S}} \Xi_{s}(b_{s})  \neq \emptyset.
  \end{align}

\end{df}
Note that~\eqref{eq:47} is a non-trivial condition, since in general
$\support(U_\sessions) \neq \prod_s \support(U_s)$ when the $U_s$ are
not necessarily independent.

We will now construct a zero-error network code $\{Y_f, f\in\sessions\cup\edges\}$ from a regular
partition set.  Let $\{ \Xi_{s}, s\in\sessions\}$ be a regular
partition set according to Definition~\ref{df:regularpartition}.  By
\eqref{eq:47}, for each $s\in\sessions$, there exists mappings

\begin{align}
  T_{s}: \prod_{i\in\set{S}} \set{A}_{i} \mapsto  \support(U_{s} ),
\end{align}
such that for $s\in\sessions$ and $b_\sessions\in\prod_{i\in\set{S}}
\set{A}_{i}$,
\begin{align*}
  \left(T_1(b_\sessions), T_2(b_\sessions), \dots, T_{|\sessions|}(b_\sessions)\right) &\in \support(U_\sessions)\quad\text{and} \\
  \left(T_1(b_\sessions), T_2(b_\sessions), \dots,
    T_{|\sessions|}(b_\sessions)\right) &\in
\Xi(b_1)\times\Xi(b_2)\times\dots\times \Xi(b_{|\sessions|}).
\end{align*}
We can write this more concisely as
\begin{align}
  (T_{s}(b_{\sessions}), s\in\sessions ) \in \support(U_{\sessions} )
  \cap \prod_{i\in \set{S}} \Xi_{i}(b_{i}).\label{eq:49}  
\end{align}
By \eqref{eq:49}, 
\begin{align}\label{eq:3.regular}
  T_{s}(b_{\sessions})  \in  \Xi_{s}(b_{s}).
\end{align}
And as $\Xi_{s}$ is a disjoint partition,  $b_{s}$ can be uniquely determined from $T_{s}(b_{\sessions}) $.

Now let $\{Y_{s},s\in\sessions\}$ be a set of mutually independent
random variables such that for each $s\in\sessions$, $Y_{s}$ is
uniformly distributed over $\set{A}_{s}$. Also, for each
$s\in\sessions$, define auxiliary random variables $Z_{s}$ such that
\begin{align}\label{eq:g51}
  Z_{s} = T_{s}(Y_\sessions).
\end{align}
By \eqref{eq:49} and \eqref{eq:3.regular}, it is easy to see that 
\begin{align}
\support(Z_\sessions) &\subseteq \support(U_\sessions), \\
H\left(Y_{s} \mid Z_{s}\right) &= 0, \quad s\in\sessions \label{eq:a} .
\end{align}
Further define
\begin{align}\label{eq:g51a}
  Z_{e} \defined G_{e}(Z_\sessions), e\in\edges.
\end{align}
It is now easy to see that\footnote{Recalling our convention to denote
set union by juxtaposition.} 
\begin{align}\label{eq:g:47}
  \support(Z_{\sessions\edges}) \subseteq  \support(U_{\sessions\edges}).
\end{align}
Following from \eqref{eq:g:47}, we have that
\begin{align}\label{eq:52g}
  H\left(Z_{k} \mid Z_\alpha \right) = 0
\end{align}
whenever $H\left(U_{k} \mid U_\alpha \right) = 0$ for some $k \in
\sessions\cup\edges$ and $\alpha \subseteq \sessions\cup\edges$.

Let $Y_{e} = Z_{e}$ for all $e\in\edges$. If $\incoming(e) \cap \sessions = \emptyset$,  then 
\begin{align*}
  0 & \nequal{(i)} H\left( U_{e} \mid U_{\incoming(e)}\right) \\
  & \nequal{(ii)}  H\left( Z_{e} \mid Z_{\incoming(e)}\right) \\
  & \nequal{(iii)} H\left( Y_{e} \mid Y_{\incoming(e)}\right)
\end{align*}
where ($i$)  follows from \eqref{eq:15}, ($ii$) from \eqref{eq:52g} and ($iii$) from
the fact that $\incoming(e) \subseteq \edges$.
Now, suppose $\incoming(e) \cap \sessions \not\neq \emptyset$. As all
sources are colocated in $\problem$, then $\incoming(e) \cap
\sessions = \sessions$. In this case, application of \eqref{eq:g51}
and \eqref{eq:g51a} yields
\begin{equation*}
0 \le H(Y_{e} \mid Y_\incoming(e) ) \le H\left(Y_{e} \mid Y_\sessions\right) = H\left(Z_{e} \mid Y_\sessions\right) = 0.
\end{equation*}
Finally, for any $s\in\sessions$ and $u\in \destinationLocation(s)$, 
\begin{align*}
  0 &= H\left( U_{s} \mid U_{\incoming(u)} \right) \\
    &= H\left( Z_{s} \mid Z_{\incoming(u)} \right) \\
    &= H\left( Z_{s} \mid Y_{\incoming(u)} \right) \\
    &\nequal{(a)} H\left( Y_{s} \mid Y_{\incoming(u)} \right) 
\end{align*}
 where $(a)$ follows from \eqref{eq:a}. Hence $\{ Y_{\sessions\edges} \}$
 defines a zero-error network code for $\problem$. Furthermore, it is easy to see that 
\begin{align}
  \log \left|\support(Y_{s}) \right| & = \log \left|\set{A}_{s}\right| \\
  \log \left|\support(Y_{e}) \right| & \le  \log \left|\support( {U}_{e})\right| = H\left(U_{e}\right).
\end{align}
Hence, $(\eta,\zeta)$ is \zach\ where 
\begin{align}
  \eta(s) & = \log \left|\set{A}_{s}\right|, \quad  \forall s\in\sessions,  \\
  \zeta(e) & = \log \left|U_{e}\right|, \quad \forall  e\in\edges. 
\end{align}

The final ingredient of the proof is the following proposition which
will be proved in Appendix \ref{appendix:A}.
\begin{prop}\label{prop:regular}
  Let $\{\seq U1S \}$ be a set of quasi-uniform random variables.  If
  $H\left(U_{1}\right)$ is sufficiently large, then there exists at least one
  regular partition set $\{\Xi_{s}, s\in\sessions\}$ for
  $\{\support(\set{U}_{s}), s\in\sessions\}$ where
  \[
  \set{A}_{s} \define \left\{1, \ldots,  \frac{2^{H\left(U_{s}\mid  \seq U1{s-1}\right)}}{H\left(\seq U1s\right)^{2}}
  \right\}.
  \]
\end{prop}

By Proposition \ref{prop:regular} (and \eqref{eq:40}), we can
construct a sequence of zero-error network codes such that
\begin{align*}
  H\left(Y^{n}_{s}\right)  & =  H\left( U_{s}^{n} \mid  \seq {U^{n}}1{s-1}\right) - 2 \log H\left(\seq {U^{n}}1s\right)    \\
  H\left(Y^{n}_{e}\right) & \le H\left(U^{n}_{e}\right). 
\end{align*}
Finally, from \eqref{eq:3.35} and that $h\in\set{C}_{\indep}$, we can prove that  
\begin{align*}
  \lim_{n\to \infty}  c_{n}H\left(Y^{n}_{s}\right) &= h(s), \\
  \lim_{n\to\infty} c_{n}{H\left(U^{n}_{e}\right)} & \le h(e)
\end{align*}
 for all $s\in\sessions$ and $e\in\edges$. Thus, $\proj_{\problem}[h]$ is \zach\ and 
Theorem \ref{thm:arbitrarymain}  follows.


\subsection{Generalisation -- non-colocated sources}\label{sec:general}
Assuming that all sources are colocated, we proved in the previous
subsection that a rate-capacity tuple is \vach\ if and only if it is
indeed \zach, and that the outer bound in Corollary \ref{cor:ourbd} is
indeed tight. 
We conjecture that the same outer bound remains tight even for sources
are not colocated. In this subsection, we will give some arguments to
justify our conjecture, which hinges on whether or not removing a
zero-rate link can change the capacity of a certain modification of
the original network.

Let $\problem = (\graph, \multicastRequirement)$ where $\graph =
(\set{V}, \set{E})$ and $\multicastRequirement = (\sessions,
\sourceLocation, \destinationLocation)$. We make no assumption
that the sources are colocated.  Now consider two variations on
the network coding problem $\problem$. 

\emph{Variation 1 -- addition of a ``super node'':} Let
$\problem^{1}\define (\graph^{\dagger},\multicastRequirement^{1})$,
where the underlying network $\graph^{\dagger}= (\set{V}^{\dagger},
\set{E}^{\dagger} )$ is obtained from $\graph=(\set{V}, \set{E})$ via
inclusion of a ``super node'' $v^{*}$ and links $f_{s},
s\in\sessions$,
\begin{align}
  \set{V}^{\dagger} & = \set{V} \cup \{ v^{*}\}, \\
  \set{E}^{\dagger} & = \set{E} \cup \{ f_{s}, s\in\sessions \},
\end{align}
such that $\tail{f_{s}} = v^{*}$ and $\head{f_{s}} = \sourceLocation
(s)$. Here, the links $f_{s}$ are just like an imaginary source edge.

The connection constraint $\multicastRequirement^{1}$ is $(\sessions,
\sourceLocation^{1}, \destinationLocation)$ where for all
$s\in\sessions$,
\[
O^{1}(s) = \sourceLocation(s) \cup \{  v^{*} \}.
\] 
In other words, source $s$ is available not only at source node
$\sourceLocation(s)$ but also at the super node $v^{*}$.


\emph{Variation 2:}
Let $\problem^{2}\define
(\graph^{\dagger},\multicastRequirement^{2})$, where the underlying
network $\graph^{\dagger}= (\set{V}^{\dagger}, \set{E}^{\dagger} )$ is the same as in $\problem^1$,  but the
connection constraint $\multicastRequirement^{2}=(\sessions,
\sourceLocation^{2}, \destinationLocation^{2})$ is modified as follows:
\begin{align}
  O^{2}(s) & = \{v^{*}\}, \quad \forall s\in\sessions \\
  D^{2}(s) & = D(s) \cup \sourceLocation(s).
\end{align}
Hence, all sources are available only at the super node $v^{*}$.  In
addition, source $s$ is required to be reconstructed not only at the
nodes in $D(s)$ (the sinks in the original multicast problem
$\problem$) but also at the original source nodes $\sourceLocation(s)$
in $\problem$.

Figure \ref{fig:variations} illustrates the differences between the
original problem $\problem$ and its two variations $\problem^1$ and
$\problem^2$.  In this figure, a sink node is denoted by an open
square. Labels beside each sink node indicate the set of sources
required for reconstruction. Sources are indicated by a double circle
with an imaginary link (labeled with the source index) directed from
it to the nodes where that source is available according to the
connection constraint. 

\def\scalefactor{0.8}
\begin{figure}[htbp]
\begin{center}
\subfigure[The original problem $\problem$.]
{\includegraphics*[scale=\scalefactor]{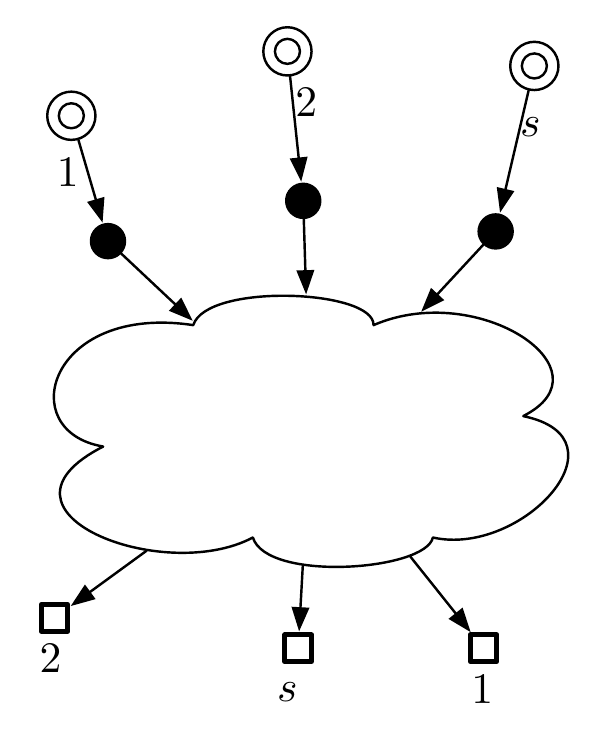}}
\subfigure[The first variation $\problem^{1}$.]
{\includegraphics*[scale=\scalefactor]{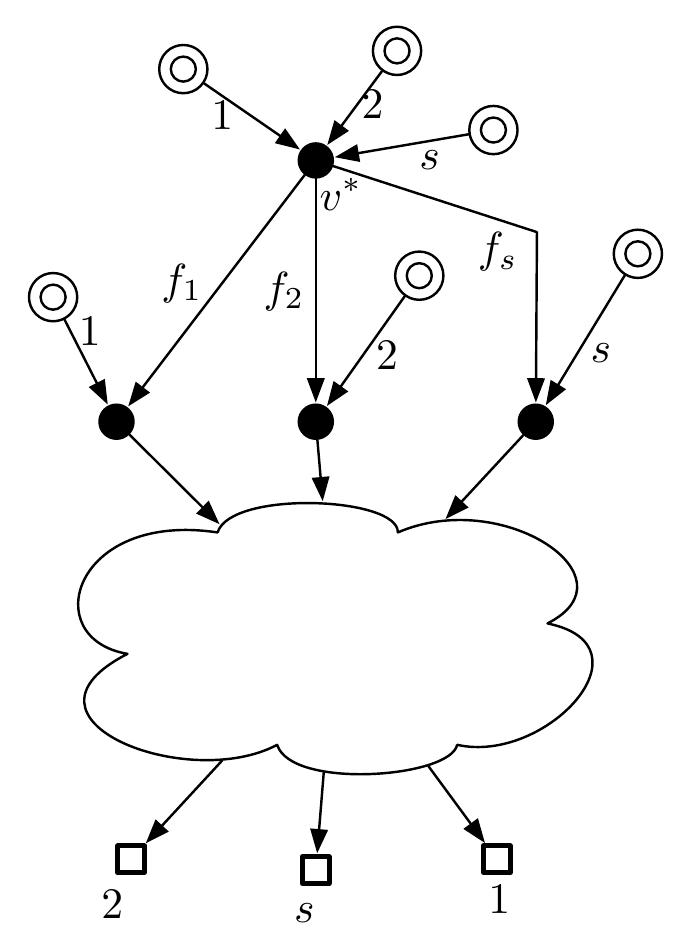}}
\subfigure[The second variation $\problem^{2}$.]
{\includegraphics*[scale=\scalefactor]{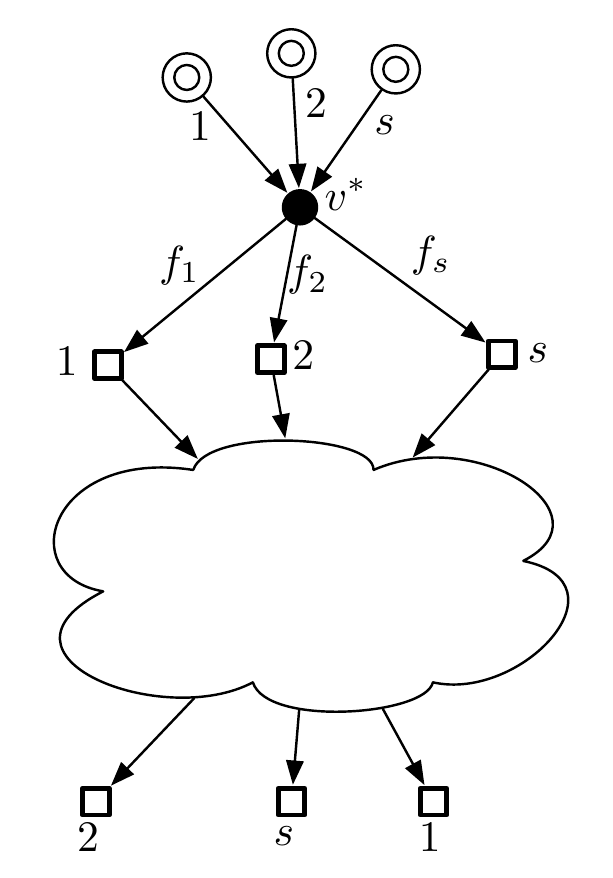}}
\caption{Variations of a multicast problem $\problem$}
\label{fig:variations}
\end{center}
\end{figure}

Given a rate-capacity tuple $(\lambda, \omega) \in {\chi}(\problem) $,
let $T^{1}[\lambda,\omega]$ be a rate capacity tuple in
${\chi}(\problem^{1}) $ (for the network coding problem
$\problem^{1}$) such that
\begin{align}
T^{1}[\lambda,\omega](s) & = \lambda(s), \quad \forall s\in\sessions, 
\\  
T^{1}[\lambda,\omega](e) & = \omega(e), \quad \forall e\in\edges,
\\
T^{1}[\lambda,\omega](f_{s}) & = 0 , \quad \forall s\in\sessions.
\end{align}
In other words,  the source rates and the hyperedge capacities remain the same for
those elements existing in the original network, and the hyperedge
capacities of the new links from the super-node $v^*$ to each of the
original source nodes are all zero.

Similarly, let $T^{2}[\lambda,\omega] $ be respectively the rate
capacity tuple in ${\chi}(\problem^{2}) $ (for the network coding
problem $\problem^{2}$) such that
\begin{align}
 T^{2}[\lambda,\omega](s) & = \lambda(s), \quad \forall s\in\sessions,
\\  
 T^{2}[\lambda,\omega](e)  & = \omega(e), \quad \forall e\in\edges,
\\
T^{2}[\lambda,\omega](f_{s}) & =  \lambda(s), \quad \forall s\in\sessions.
\end{align}

Then, it is straightforward to prove the following:
\begin{enumerate}
\item $T^{1}[\lambda,\omega]$ is \vach\ with respect to
  $\problem^{1}$ if and only if $T^{2}[\lambda,\omega]$ is \vach\
  with respect to $\problem^{2}$.
\item If $(\lambda,\omega)$ is \vach\ with respect to $\problem$,
  then $T^{1}[\lambda,\omega]$ and $T^{2}[\lambda,\omega]$ are
  \vach\ with respect to $\problem^{1}$ and $\problem^{2}$
  respectively.
\end{enumerate} 

\begin{conjecture}\label{conj:two}
  A rate-capacity tuple $(\lambda,\omega)$ is \vach\ with respect to a
  network coding $\problem$ if $T^{1}[\lambda,\omega]$ is \vach\ with
  respect to the modified problem $\problem^{1}$.
\end{conjecture}

The main difference between $\problem$ and $\problem^{1}$ are the
zero-capacity links $f_{s}$ for $s\in\sessions$. At first sight it
might be tempting to think that zero-capacity links cannot change the
capacity region, and as a result that the conjecture is trivially
true. However proving the conjecture is not straightforward. A
zero-capacity link does not mean that absolutely nothing can be
transmitted on the link. In fact according to the definitions, a
finite amount of information (that does not scale with $n$) could be
transmitted along the link. Thus the links $f_{s} $ can in fact be
used (as long as their capacities vanish asymptotically) in any
sequence of network codes achieving $T^{1}[\lambda,\omega]$. In fact
there exists known examples where zero-capacity links can indeed modify
the capacity of certain multi-terminal problems, in particular when
there are correlated sources, or non-ergodic sources.

If Conjecture~\ref{conj:two} does not hold, it is equivalent to saying that a
link with ``vanishing capacity'' can indeed change the set of
achievable tuples, even when all sources are independent and ergodic.

\begin{thm}
  Suppose Conjecture ~\ref{conj:two} holds. Then the outer bound in
  Corollary \ref{cor:ourbd} is tight even when sources are not
  colocated.
\end{thm}
\begin{proof}
  Let $h \in \bar\Gamma^{*}(\problem)$. Define a rank function 
  \[
  g \in\set{H}[\sessions\cup\edges^{\dagger}]
  \]
  such that for any $\beta \subseteq \sessions\cup\edges^{\dagger}$
\begin{align}
g(\beta) = h\left(\alpha_{1} \cup \alpha_{2} \right) 
\end{align}
where 
\begin{align}
\alpha_{1} & = \beta \setminus \{ f_{s}, s\in\sessions \}   \\
\alpha_{2} & =\{ s\in\sessions:  f_{s} \in \beta  \}.
\end{align}
It is straightforward to prove that if 
\[
h\in \bar \Gamma^{*}(\problem)  \cap \set{C}_{\indep}(\problem) \cap \set{C}_{\net}(\problem)  \cap \set{C}_{\de}(\problem),
\] 
then
\[
g\in \bar \Gamma^{*}(\problem^{2})  \cap \set{C}_{\indep}(\problem^{2}) \cap \set{C}_{\net} (\problem^{2}) \cap \set{C}_{\de}(\problem^{2}).
\]
By Theorem \ref{thm:arbitrarymain}, $\proj_{\problem^{2}}[g]$ is
\vach\ with respect to network coding problem $\problem^{2}$.  As
\[
T^{2}[\proj_{\problem}[h ]] = \proj_{\problem^{2}}[g ]
\]
and is achievable with respect to $\problem^{2}$,
$T^{1}[\proj_{\problem}[h ]]$ is achievable with respect to
$\problem^{1}$. By Conjecture~\ref{conj:two}, $\proj_{\problem}[h ]$ is
achievable with respect to $\problem$. Consequently, the outer bound
in Corollary \ref{cor:ourbd} is tight.
\end{proof}

\def\hyperplane{{\cal H}}
  
\def\vA{{\mathbb A}}
\def\vB{{\mathbb B}}
\def\vC{{\mathbb C}}
\def\vU{{\mathbb U}}
\def\vV{{\mathbb V}}
\def\vW{{\mathbb W}}

\section{Linear Network Codes}\label{sec:part2}
In the previous section, the network codes were not subject to any
constraints, other than those required by Definition~\ref{df:nc} and
that decoding error probabilities must be zero or vanishing, according
to Definition~\ref{df:admissible} or
Definition~\ref{df:achievable}. 

For the remainder of the paper, we will consider various subclasses
of network codes which result from imposing different kinds of
additional constraints. 

To begin with, in this section we will study \emph{linear network
  codes}~\cite{Li.Yeung.ea03linear}, which have relatively low
encoding and decoding complexities, making them more attractive for
practical implementation.

\begin{df}[Linear network codes]
  Let
  \begin{align}\label{eq:df:5alnc}
    \{\edgeRV_f: f \: \in \edges \cup \sessions\}
  \end{align}
  be a network code (according to Definition~\ref{df:nc}) for a
  problem $\problem$ on a network $\graph=(\set{V},\set{E})$, with
  local encoding functions
  \[
  \networkCoding \triangleq \{\networkcoding_e :\: e \in \edges \}.
  \]
  The code is called \emph{$q$-linear} (or simply linear) if it
  satisfies the following conditions:
  \begin{enumerate}
  \item For $s\in\sessions$, $Y_{s}$ is a random row vector such that
    each of its entries is selected independently and uniformly over
    $\field(q)$.
\item All the local encoding functions 
are linear.
\end{enumerate}
A network coding problem is said to be subject to a \emph{$q$-linearity
  constraint} if only $q$-linear network codes are allowed.
\end{df} 

Let the row vector $Y_{s}$ have $\lambda_{s}$ elements.  Clearly, for
a linear network code \eqref{eq:df:5alnc}, one can construct a
$\sum_{i\in\sessions} \lambda_{i} \times \lambda_{s}$ matrix $G_{s}$
for any $s\in \sessions$ such that
\begin{align}\label{eq:70}
  Y_{s} = Y \times  G_{s}.
\end{align}
where $Y\defined[Y_1 Y_2 \dots Y_{|\set{S}|}]$ is the length
$\sum_{i\in\sessions} \lambda_{i} $ row vector obtained from the
concatenation of the $Y_i$ and $\times$ is the usual vector-matrix
multiplication.

Similarly, as all local encoding functions are linear, for each
$e\in\edges$, there exists a $\sum_{i\in\sessions} \lambda_{i} \times
\omega_{e}$ matrix $G_{e}$ such that
\begin{align} 
Y_{e} = Y \times  G_{e}.
\end{align}
Hence, the symbol transmitted on link $Y_{e}$ is a length $\omega_{e}$ vector over $\field(q)$.

Following the nomenclature in \cite{Yeung02first}, the matrices
$G_{f}$, $f\in\sessions\cup\edges$ will be called the \emph{global
  encoding kernels}.  They define the linear relation between $Y_{e}$
(the message sent along edge $e$) and $\{Y_{s},s\in\sessions\}$ (the
symbols generated at the sources).  It is easy to prove that
\begin{align}
  |\support(Y_{s})| & = q^{ \lambda_{s} } \label{eq:74} \\
  |\support(Y_{e})| & \le q^{\omega_{e}}. \label{eq:75}
\end{align}
with equality holding in \eqref{eq:75}  when $G_{e}$ has full column rank.  

A sink node $u\in\destinationLocation(s)$ can uniquely decode a source
$Y_{s}$ if and only if it can solve for $Y_{s}$ (but not necessarily
all other $Y_{i}$, $i \in \sessions \setminus s$) from the following
linear system with unknowns $Y$:
\begin{align}\label{eq:systemofeq}
  Y_{e} & =  Y \times  G_e, \quad \forall  e\in\incoming(u).
\end{align} 
It is clear that if $Y_{s}$ cannot be uniquely determined from
\eqref{eq:systemofeq}, then there must be at least $q$ solutions to
\eqref{eq:systemofeq} such that the values of $Y_{s}$ in each solution
are all different.  We can easily show that with maximum likelihood
decoding, the decoding error probability must be at least $1-1/q$.
Consequently, for a network coding problem subject to a $q$-linearity
constraint, a rate-capacity tuple $(\lambda,\omega)$ is \zach\ if and
only if it is \vach. Therefore, we will always assume that linear
codes are zero-error codes.

As in Theorem \ref{thm:arbitrarymain} for general (possibly
non-linear) network codes, a characterisation for the set of \zach\
rate-capacity tuples subject to a $q$-linearity constraint can be
obtained by using entropy functions. However as we shall see, the
entropy functions for linear network codes will be constrained to be
\emph{representable}~\cite{oxley92} in the sense of matroid theory.
\begin{df}\label{df:repfn}
  A rank function $h \in \set{H}[\sessions\cup\edges]$ is called
  \emph{$q$-representable} if there exists vector subspaces
  \[
  \vU_{i}, i\in \sessions\cup\edges
  \]
  over $\field(q)$ such that for all $\alpha\subseteq \sessions\cup\edges$,
  \begin{align}\label{eq:77}
    h(\alpha) = \dim \la \vU_{i} , i\in\alpha \ra ,
  \end{align}
  where $\dim\la \vU_i, i\in\alpha \ra$ denotes the dimension of the smallest vector
  space containing all of the $\vU_i, i\in\alpha$.
\end{df}

In an abstract sense, representable rank functions are similar to
entropy functions, where the entropies of random variables are
replaced by dimensions of vector spaces.  
It is well-known that representable functions are indeed entropy
functions (and hence also
polymatroidal)~\cite{Yeung97framework,Chan2011Recent}. Hence we will
sometimes use the following conventions. For vector spaces $\vU$ and
$\vV$, we will use
\[
H\left(\vV\right), \:H\left(\vU,\vV\right), \text{ and } H\left(\vU \mid  \vV\right)
\]
to respectively denote 
\[
\dim \vV, \:  \dim \la \vU,\vV \ra,   \text{ and } \dim \la \vU,\vV \ra - \dim \vU,
\]
as if $\vU$ and $\vV$ were random variables.  

Similar to Definition \ref{df:classification} we can define weakly-
and almost-representable functions.
\begin{df}[Weakly/almost representable]
  A function $h$ is called
  \begin{itemize}
  \item \emph{weakly $q$-representable} if $c \cdot h$ is
    $q$-representable for some $c > 0$
  \item \emph{almost
      $q$-representable} if it is the limit of a sequence of weakly
    $q$-representable functions.
  \end{itemize}
  We use $\Upsilon^{*}_{q}(\sessions\cup\edges)$ and
  $\bar\Upsilon^{*}_{q}(\sessions\cup\edges)$ to respectively denote
  the sets of $q$-representable and almost $q$-representable rank
  functions in $\set{H}[\sessions\cup\edges]$.
\end{df}

 
\thm[Achievability by linear codes]\label{thm:mainlinear} For any
network coding problem $\problem=(\graph, \multicastRequirement)$
subject to a $q$-linearity constraint, a rate-capacity tuple $(\lambda, \omega)$ is \zach\ if
and only if
\begin{align} 
  (\lambda, \omega) \in \major \left(  \proj_{\problem}
    \left[\bar\Upsilon^{*}_{q}(\problem) \cap
      \set{C}_{\indep}(\problem) \cap \set{C}_{\net}(\problem)  \cap
      \set{C}_{\de}(\problem)  \right]\right) . 
\end{align} 
\endthm 
Theorem \ref{thm:mainlinear} (for linear network codes) is a counterpart to
Theorem \ref{thm:arbitrarymain} (for general network codes codes). However, unlike in
Theorem \ref{thm:arbitrarymain}, Theorem \ref{thm:mainlinear} holds
\emph{even when the sources are not colocated}.

Before we proceed to prove Theorem \ref{thm:mainlinear} (which
involves proving both an if-part and an only-if part), we will
illustrate the main idea by proving the following special case of the
if-part of Theorem \ref{thm:mainlinear}.
\begin{prop}\label{prop:5.1}
Consider a $q$-representable function $h\in \Upsilon^{*}_{q}(\problem)$ such that
\begin{align} \label{eq:78}
  h\in \set{C}_{\indep}(\problem) \cap \set{C}_{\net}(\problem)  \cap \set{C}_{\de}(\problem)
\end{align}
and let $(\lambda,\omega) =\proj_{\problem}[h] $.
Then $(\lambda,\omega)$ is 0-achievable by $q$-linear network codes.
\end{prop}
\begin{proof}
By Definition \ref{df:repfn}, there exists a collection of subspaces 
\[
\{ \vV_{i} , i\in\sessions\cup\edges  \}
\]  
over $\field(q)$ such that 
\[
h(\alpha) = \dim \la \mathbb V_{i} , i\in\alpha \ra .
\]
As $h \in \set{C}_{\indep}(\problem)   $, 
\[
h(\sessions) = \sum_{s\in\sessions} h(s).
\]
Therefore,  
\begin{align}\label{eq:81}
\dim \la \vV_{s}, s\in\sessions \ra = \sum_{s\in\sessions} \dim \vV_{s}.  
\end{align}
Similarly, as $h\in   \set{C}_{\net}(\problem) \cap \bar\Gamma^{*}(\problem)$, we can prove that 
\[
h(e, \sessions) =  h(\sessions)
\]
for all $e\in\edges$. Consequently, 
 \begin{align}\label{eq:82}
\vV_{e} \subseteq \la \vV_{s}, s\in\sessions \ra , \quad \forall e\in\edges.
\end{align} 
Let $k=\dim \la \vV_{s}, s\in\sessions \ra$. By \eqref{eq:81} and
\eqref{eq:82}, we may assume without loss of generality that all
elements in $\vV_{f}$ are length $k$ column vectors for
$f\in\sessions\cup\edges$.  For each $s\in\sessions$, let $M_{s}$ be a
$k\times h(s)$ full column rank matrix such that the space spanned by
its columns is equal to $\vV_{s}$. Similarly, for each $e\in\edges$,
let $M_{e}$ be a $k\times h(e)$ full column rank matrix such that the
space spanned by its columns is equal to $\vV_{e}$.

Let $Z$ be a length $k$ row vector such that each of its entries is
independently and uniformly selected from $\field(q)$. Then $Z$,
together with the matrices $\{M_{f}, f\in\sessions\cup\edges\}$
induces a set of random variables
\[
\{Y_{f}, f\in\sessions\cup\edges\}\] such that $Y_{f} \triangleq Z
\times M_{f}$ for all
$f\in\sessions\cup\edges$. 
 Setting $Y=[Y_{s},s\in\sessions]$, and similarly $M = [M_{s},
 s\in\sessions]$ we have
\[
Y = Z \times  M .
\]
By \eqref{eq:81}, it is easy to see that $M$ must be a $k \times k$ invertible matrix.
Therefore
\begin{align}
  Z = Y  \times M^{-1}
\end{align}
and consequently, $Y_{f} = Y \times G_{f}$ where $G_{f} = M^{-1}
\times M_{f}$.

In the following, we will prove that $\{Y_{f},
f\in\sessions\cup\edges\}$ is indeed a zero-error linear network
code. To see this, first notice that $Y_{s}$ is a random row vector of
length $h(s)$ and all its entries are independently and uniformly
distributed over $\field(q)$.  Now, consider any $e\in\edges$.  As
$h\in\set{C}_{\net}$,
\[
h(e, \incoming (e)) = h(\incoming (e)).
\]
Therefore,  
\[
\dim \la \vV_{e}, \vV_{f}, f\in{\incoming(e)}  \ra = \dim \la  \vV_{f}, f\in\incoming(e)  \ra
\]
or equivalently, 
\[
 \vV_{e}\subseteq \la  \vV_{f}, f\in\incoming(e)  \ra.
\]

Consequently,  we can construct a matrix $\psi_{e}$ such that 
\[
M_{e}  =  [ M_{f}, f\in\incoming(e)  ]\times\psi_{e}
\]
Therefore, 
\begin{align}
Y_{e} &= Z \times M_{e} \nonumber\\
&= Z \times[ M_{f}, f\in\incoming(e)  ] \times \psi_{e} \nonumber \\
&=  [ Z\times M_{f}, f\in\incoming(e)  ] \times\psi_{e}  \nonumber\\
&=  [ Y_{f}, f\in\incoming(e)  ] \times\psi_{e} \label{eq:85}.
\end{align}
The equation \eqref{eq:85} clearly indicates that the local encoding functions are indeed linear. 

Similarly, as $h\in \set{C}_{\de}(\problem) $, 
\[
h(s, {\incoming(u)}) = h(\incoming(u))
\]
for any $u\in\destinationLocation(s)$.
We can once again construct a ``decoding function'' $\psi_{s,u}$ such that
\[
Y_{s} = [ Y_{f}, f\in\incoming(u)  ] \times\psi_{s,u}.
\] 
Hence, the receiver node $u$ can uniquely decode $Y_{s}$ from
$\{Y_{f}, f\in\incoming(u)\}$ and the probability of decoding failure
is zero.  As $\{Y_{f}, f\in\sessions\cup\edges\}$ is a zero-error
linear network code, together with \eqref{eq:74}-\eqref{eq:75},
$(\lambda,\omega)$ is \zach\ by linear network codes. 
\end{proof}

Our next step is to prove that Proposition \ref{prop:5.1} holds, even
when the function $h$ in \eqref{eq:77} is almost
$q$-representable. Before we prove this extension, we will need a few
basic results from linear algebra.
\begin{lemma}\label{lemma:5.1}
  Let $\vA, \vB$ be vector subspaces. Then there exists a subspace $\vC \subseteq \vB$  such that
  \begin{align*}
    \la \vA,\vC \ra & = \la \vA, \vB\ra \\
    \vA \cap \vC & = \{ \bf 0 \} .
  \end{align*}
  Consequently, $H\left(\vC\right) = H\left(\vB \mid  \vA \right) = H\left(\vA,\vB \right)  - H\left(\vA\right)$.
\end{lemma}
  
Using Lemma \ref{lemma:5.1}, for any subspace $\vA$ of  $\vB$, one can easily construct a vector subspace
 ${\vA}^{*}$  of ${\vB}$ such that
 \begin{align}
 \la{\vA},{\vA}^{*}\ra & ={\vB}, \\
 {\vA} \cap {\vA}^{*} & = \{{\bf 0}\}.
 \end{align}
Any vector $  u\in{\vB}$ can be uniquely expressed as the  sum of vectors ${u}_{1}\in {\vA}^{*}$ and $ {u}_{2}\in {\vA}$.  For notational simplicity, we use $T_{\vA}({u})$ to denote ${ u}_{1}$.  Similarly, 
\[
T_{\vA}(\vV) \defined \{ T_{\vA}({u}): \: u\in \vV  \}.
\]
 Clearly, if $\vV$ is a subspace, then so is  $T_{\vA}(\vV) $.

 In general, $T_{\vA} ({u})$ depends on the specific choice of
 $\vA^{*}$. However, all the results mentioned in this paper involving
 $T_{\vA}$ remain valid for any legitimate choice of ${\vA}^{*}$. The
 following lemma may be directly verified.
\begin{lemma}[Properties]\label{lemma:5.properties}
For any subspaces $\vA, \vB_{1},\vB_{2}$,
\begin{align}
  T_{\vA}(\vB_{1}) \cap \vA & =  \{ {\bf 0} \} \\
  \la T_{\vA}(\vB_{1}), \vA \ra & = \la \vA,\vB_{1}\ra \\
  H(T_{\vA}(\vB_{1})) & = H\left(\vB_{1} \mid \vA \right) \\
  \la T_{\vA}(\vB_{1}) , T_{\vA}(\vB_{2}) \ra & = T_{\vA}(\la \vB_{1}, \vB_{2} \ra)
\end{align}
Furthermore, if $\vB_{1}\subseteq \vB_{2}$, then 
\[
T_{\vA}(\vB_{1}) \subseteq T_{\vA}(\vB_{2}).
\]
\end{lemma}

Using  $T_{\vA}$, we can ``transform'' a set of subspaces 
\[
\{\seq {\vB}1n \}
\]
into another set of subspaces
\[
\left\{T_{\vA}(\vB_{1}) , \ldots , T_{\vA}(\vB_{n}) \right \}
\]
satisfying Lemmas~\ref{lemma:5.3} and ~\ref{lemma:5.4} below, which are direct consequences of  Lemma \ref{lemma:5.properties}.
\begin{lemma}[Conditioning] \label{lemma:5.3}
\begin{align}
  H\left(T_{{\vA}}({\vB_{i}}), i\in\alpha \right) = H\left(T_{\vA} \la \vB_{i},i\in\alpha\ra\right)
  =  H\left(\vB_{i}, i\in\alpha\mid \vA\right).
\end{align}
\end{lemma}

%
\begin{lemma}[Preserving functional dependencies] \label{lemma:5.4}
  If 
  \[
  H\left(\vB_{k}\mid \vB_{i},i\in\alpha\right)=0,
  \]
  or equivalently, $\vB_{k} \subseteq \la \vB_{i},i\in\alpha \ra$,  then 
  \[
  H\left(T_{{\vA}}({\vB_{k}}) \mid T_{{\vA}}({\vB_{i}}), i\in\alpha\right) = 0.
  \]
\end{lemma}

\begin{prop}\label{prop:5.5}
  Consider any almost $q$-representable function $ \bfh \in
  \set{H}[\sessions\cup\edges]$.  Let
  \begin{equation*}
    \Delta \define \{ (k,\alpha) : k\in \sessions\cup\edges , \alpha \subseteq \sessions\cup\edges \text{ such that } {\bfh}( {k} \mid {\alpha})=0\}.
\end{equation*}
  Then there exists a sequence of weakly $q$-representable rank
  functions $ \bfh^{\ell}\in \set{H}[\sessions\cup\edges]$ such that
  \begin{align}
    \lim_{\ell\to\infty}{h^{\ell}}  & = {h} 
  \end{align}
  and
  \begin{align}
    {\bfh^{\ell}}( {k}  \mid {\alpha}  ) & =0,  
  \end{align}
  for all positive integers $\ell$,   and $(k,\alpha) \in \Delta$.
\end{prop}
\begin{proof}
  By definition,  for any almost $q$-representable function $\bfh$,  
  there exists a sequence $\{ \vU^{\ell}_{f},
  f\in\sessions\cup\edges\}$, $\ell=1,2,\dots$ of collections of subspaces 
  over $\field(q)$ and $c_{\ell}>0$ such that for any $\alpha \subseteq \sessions\cup\edges$,
  \begin{align} \label{eq:97}
    \lim_{\ell \to\infty} c_\ell H\left( \vU^{\ell}_{\alpha}\right) = h\left(\alpha\right).
  \end{align}
For every pair of $(k,\alpha) \in \Delta$,  by using Lemma \ref{lemma:5.1},
we can  construct a subspace $\vW^{\ell}_{k,\alpha}$ such that 
\begin{align}
H\left(\vW_{k,\alpha}^{\ell}\right) & = H (\vU_{k}^{\ell} , \vU_{\alpha}^{\ell}) - H ( \vU_{\alpha}^{\ell})  \label{eq:98}\\
  \vU_{k}^{\ell}  & \subseteq \la \vU_{\alpha}^{\ell}, \vW_{k,\alpha}^{\ell} \ra.
\end{align}
Let $\vW_{\Delta}^{\ell} = \la \vW_{k,\alpha}^{{\ell}}, (k,\alpha) \in
\Delta \ra$.  By \eqref{eq:97}-\eqref{eq:98} and the fact that
\[
h\left(k,\alpha\right) = h\left(\alpha\right), \quad \forall (k,\alpha) \in \Delta,
\] 
 we have 
\begin{align}
\lim_{\ell\to\infty} c_{\ell}H\left(\vW_{\Delta}^{\ell}\right) & \le 
	\lim_{\ell\to\infty} c_{\ell} \sum_{(k,\alpha) \in\Delta}H\left(\vW_{k.\alpha}^{\ell}\right) \nonumber \\
& = \sum_{(k,\alpha) \in\Delta} \lim_{\ell\to\infty} c_{\ell} H\left(\vW_{k.\alpha}^{\ell}\right) \nonumber\\
& = \sum_{(k,\alpha) \in\Delta} \left(h\left(k,\alpha\right)  - h\left(\alpha\right) \right) \nonumber\\
& = 0.	\label{eq:100}
\end{align} 
 
Define a new collection of subspaces
\begin{align}
\vV^{\ell}_{f} \defined \la U^{\ell}_{f}, W_{\Delta}^{\ell} \ra, \quad \forall f\in\sessions\cup\edges
\end{align}
and let $g^{\ell}$ be the representable function induced by 
\[
\{ \vV^{\ell}_{f} , f\in \sessions\cup\edges\}.
\]
Obviously, for all $(k,\alpha) \in \Delta$, $\vV_{k}^{\ell}$ is a
subspace of $\vV_{\alpha}^{\ell}$, or equivalently,
$g^{\ell}(k,\alpha) = g^{\ell}(\alpha)$.  By \eqref{eq:100}, using a
similar argument as given in \eqref{eq:29}-\eqref{eq:34} in the proof
of Proposition \ref{prop:3.support}, we can also prove that
\[
\lim_{\ell\to\infty} c_{\ell} g^{\ell } = h.
\] 
The proposition then follows by letting $h^{\ell} = c_{\ell} g^{\ell}$.
\end{proof}

In the proof for Theorem \ref{thm:arbitrarymain}, an extra step (via
the introduction of a regular partition set) is taken to construct a
network code from a set of quasi-uniform random variables. This extra
step requires that all sources are colocated. As we shall see, when
we construct linear codes from a set of subspaces, the colocated
assumption is no longer needed.
\begin{lemma}\label{lemma:5.5}
  For any subspaces $\vA ,\vB$,
  \begin{align}
    H\left(\vA\mid  \vB \right) = H\left(\vA\mid  \vA \cap \vB\right).
  \end{align}  
\end{lemma}
\begin{proof} Direct verification. \end{proof}

\begin{lemma}\label{lemma:5.6}
  Let $\{ \vV_{f}, f\in\sessions\cup\edges \}$ be a collection of subspaces. Then there exists a subspace $\vA$ such that
  \[
  H(T_{\vA}( \vV_{s}), s\in\sessions  )  =\sum_{s\in\sessions} H(T_{\vA}( \vV_{s}))
  \]
  and 
  \[
  H\left(\vA\right) = H\left( \vV_{\sessions} \right) -
  \sum_{s\in\sessions} H\left(\vV_{s}\mid \vV_{\sessions \setminus s}\right).
  \]
\end{lemma}
\begin{proof}
  Let 
  \begin{align}
    \vW_{s} \defined \vV_{s} \cap \la \vV_{{j}},
    j \in \sessions\setminus s \ra , \: \forall s\in\sessions \label{eq:63}
  \end{align} 
  and  $\vA \defined \la \vW_{s}, s \in\sessions \ra $.  Then
  %
  %
  \begin{align*}
    H\left(\vV_{\sessions} \mid  \vA\right) & = H\left(\vV_{\sessions}  \mid  \vW_{i}, i\in\sessions\right)  \\
    & \ge \sum_{ s\in \sessions} H\left(\vV_{s} \mid \vW_{i}, i\in\sessions,  \vV_{{k}}, k\neq s  \right) \\  
    & \nequal{(a)} \sum_{ s\in \sessions} H\left(\vV_{s} \mid  \vW_{s},   \vV_{{k}}, k\neq s  \right) \\  
    & \nequal{(b)} \sum_{ s\in \sessions} H\left(\vV_{s} \mid  \vV_{{k}}, k \neq s  \right) \\  
    &\nequal{(c)} \sum_{ s\in \sessions} H\left(\vV_{s} \mid  \vW_{s}\right) \\
    & \ge  \sum_{ s\in \sessions} H\left(\vV_{s} \mid  \vW_{i } , i\in\sessions\right) \\
    & =\sum_{ s\in \sessions} H\left(\vV_{s}\mid  \vA\right) \\
    & \ge H\left(\vV_{\sessions} \mid  \vA\right)  
  \end{align*}
  where $(a)$ follows from that 
  $\vW_{{k}} \subseteq \vV_{{k}}$,  $ (b) $   from that
  $\vW_{{k}} \subseteq \la \vV_{{i}} , i\neq k \ra$, and $(c)$ from 
  Lemma \ref{lemma:5.5}.
  As  all the above inequalities are in fact equalities,  we have
  \[
  H\left(\vV_{\sessions} \mid  \vA\right)  = \sum_{s\in\sessions} H\left(\vV_{s}\mid  \vA\right) = \sum_{ s\in \sessions} H\left(\vV_{s} \mid  \vV_{{k}}, k \neq s  \right).
  \]
  Finally, as  $\vA \subseteq \la  \vV_{s}, s\in\sessions  \ra$, 
  \begin{align}
    H\left(\vA\right) & = H\left(\vV_{s},s\in\sessions\right) - H\left(\vV_{\sessions} \mid  \vA\right) \\
    & = H\left(\vV_{s},s\in\sessions\right) - \sum_{ s\in \sessions} H\left(\vV_{s} \mid  \vV_{{k}}, k \neq s  \right).
  \end{align}
\end{proof}

We now have all the elements required to prove Theorem \ref{thm:mainlinear}.
\begin{proof}[Proof of Theorem \ref{thm:mainlinear} ]

We begin with the if-part. Suppose 
\[
h \in \bar\Upsilon^{*}_{q} \cap \set{C}_{\indep} \cap \set{C}_{\net}  \cap \set{C}_{\de}.
\]
Using Proposition \ref{prop:5.5}, we can construct a sequence of
$q$-representable functions $ f^{\ell} $ and $c_{\ell} > 0$ such that
\[
\lim_{\ell\to\infty} c_{\ell} f^{\ell} = h
\] 
and each $f^{\ell}$ satisfies all of the same functional dependencies
as $h$, i.e. 
$$h\left(k\mid\alpha\right)=0\implies f^{\ell}(k\mid\alpha) = 0.$$
In particular, 
\[
f^{\ell} \in \set{C}_{\net}  \cap \set{C}_{\de}.
\]
For each $\ell$, by definition, there exists subspaces 
\[
\{ \vU_{i}^{\ell}, i\in\sessions\cup\edges \}
\]
over $\field(q)$ such that $f^{\ell}(\alpha) = H\left(\vU_{i}^{\ell},i\in\alpha\right)$. 

Then by Lemma \ref{lemma:5.6}, there exists a subspace $\vA^{\ell}$ such that
\begin{align}\label{eq:g.104}
H(T_{\vA^{\ell}} ( \vU^{\ell}_{s}) , s\in\sessions ) = \sum_{s\in\sessions} H(T_{\vA^{\ell}} ( \vU^{\ell}_{s}) )
\end{align}
and
\begin{align}\label{heq:109}
H\left(\vA^{\ell}\right) = H\left(\vU^{\ell}_{s}, s\in\sessions\right) - \sum_{s\in\sessions} H\left(\vU^{\ell}_{s}\mid  \vU^{\ell}_{\sessions\setminus s}\right).
\end{align}
Let $g^{\ell}$ be the representable function induced by the subspaces 
\[
\{T_{\vA^{\ell}} (\vU^{\ell}_i ), i\in\sessions\cup\edges    \}.
\]
Then by \eqref{eq:g.104}, $g^{\ell } \in \set{C}_{\indep}$. 
As each $f^{\ell} \in \set{C}_{\net} \cap \set{C}_{\de}$, by Lemma \ref{lemma:5.4}, 
$g^{\ell } \in \set{C}_{\net}\cap\set{C}_{\de}$. By Proposition \ref{prop:5.1}, 
$\proj_{\problem}[g^{\ell }]$ is \zach.

Due to \eqref{heq:109}, 
\[
\lim_{\ell \to\infty} c_{\ell} H\left(\vA^{\ell}\right) = 0,
\]
and hence,
\[
\lim_{\ell\to\infty} c_{\ell} g^{\ell} = \lim_{\ell\to\infty} c_{\ell} f^{\ell} = h.
\]
%
Thus, $\proj_{\problem}[h]$ is also \zach\ and the if-part of Theorem
\ref{thm:mainlinear} is proved.


Now, we will prove the only-if part.  Suppose $(\lambda, \omega) \in
{\mathsf T}(\problem) $ is \zach\ subject to a $q$-linearity
constraint. Then there exists a sequence of zero-error linear network
codes $\{Y_f^{n}, f\in {\sessions\cup\edges}\}$ and positive constants
$c_n$ such that
\begin{align*}
  \lim_{n\to\infty} {c_n} {H (Y_{e}^{n}) } & \le  \lim_{n\to\infty} {c_n}{\log | \support(Y_{e}^{n})|}  \le \omega({e}), \quad \forall e\in\edges, \\
  \lim_{n\to\infty} {c_n} {H (Y_{s}^{n}) }  & = \lim_{n\to\infty} {c_n}{\log | \support( Y_{s}^{n})|}  \ge \lambda({e}), \quad \forall s \in\sessions.
\end{align*}
Again, each set of random variables $\{ Y^{n}_{f}, f\in \sessions \cup
\edges \}$ induces a $q$-representable function $h^{n}$ such that
$H\left(Y_{f}^{n}, f\in\alpha\right) = h^{n}(\alpha)$ for all $\alpha \subseteq
\sessions \cup\edges$.  Since $\{ Y^{n}_{f}, f\in \sessions \cup
\edges \}$ is a zero-error linear code, we have $h^{n} \in
\bar\Upsilon^{*}_{q} \cap \set{C}_{\indep} \cap \set{C}_{\net} \cap
\set{C}_{\de}$.  Therefore,
\[
(\lambda,\omega) \in \major(\bar\Upsilon^{*}_{q} \cap \set{C}_{\indep} \cap \set{C}_{\net}  \cap \set{C}_{\de})
\]
and the theorem is proved. 
\end{proof}

\section{Routing}\label{sec:part3}
Another class of network coding constraints that is of great practical
importance is \emph{routing}, which requires that network nodes
perform only store-and-forward operations.  We will consider two main
cases. In Section~\ref{sec:routingonly} we consider networks where \emph{all }
nodes must perform routing. In Section~\ref{sec:routingcoding} we
consider heterogeneous networks consisting of both routing and network
coding nodes.

\subsection{Routing-only schemes}\label{sec:routingonly}
We first consider networks where the nodes are only able to perform
routing. We will formalise what we mean by ``routing'' later, and
proposed a generalisation. In such routing-based schemes, information
is transmitted from the sources to the destinations via a collection
of ``routing subnetworks''.
\begin{df}[Routing subnetworks]\label{df:routingsubnetwork}
  For any given network coding problem $\problem =
  (\graph,\multicastRequirement)$, a routing subnetwork is a subset
  $\set{T}$ of $\sessions \cup \edges$ such that
  \begin{enumerate}
  \item $|\set{T}\cap \sessions| = 1$. Thus, $\set{T}$ is associated
    with a source and we denote that unique source in $\set{T}\cap
    \sessions$ by $\nu(\set{T})$.
  \item For any link $e\in\set{T}\cap\edges$, $\incoming(e) \cap
    \set{T} \neq \emptyset$.  In other words, either there exists
    another link $f\in\set{T} $ such that
    \[
    f \in \incoming(e),
    \]
    or $\nu(\set{T})\in \incoming({e})$, i.e., the originating node of
    link $e$ is a source node of $\nu(\set{T})$. Hence, the subnetwork
    formed by the set of links in $\set{T}$ is in fact ``connected''
    and is ``rooted'' at $\nu(\set{T}) $.
  \end{enumerate}
\end{df}
A routing subnetwork is in fact a simple generalisation of the usual
multicast trees used in networks with point-to-point links (i.e. the
underlying network is a directed graph) for constructing a routing
solution (where messages are being forwarded and relayed at
intermediate nodes without coding). While it is sufficient to consider
multicast trees in such networks, the concept of multicast trees does
not extend naturally to wireless networks (where the underlying
network is a directed \emph{hypergraph}). In particular, in our
hypergraph model, links $\edges$ are broadcast, i.e., the message sent
over a link $e$ can be received by more than one node. Therefore, it
is not reasonable (and also not necessary) to insist that there is a
unique path connecting a source to a sink. According to our definition
of a multicast constraint, sources may also be available at more than
one node. Therefore, the condition $s \in \incoming(e)$ means that
there exists a node $u \in \sourceLocation(s)$ such that $u =
\tail{e}$.


\begin{df}[Achievability]\label{df:routingadmissible}
  A rate-capacity tuple 
  \[
  (\lambda,\omega)  \in \chi(\problem)
  \] is called \zach\ subject to a \emph{routing constraint} if there
  exists a collection of routing subnetworks $\set{T}_{i}$ and
  \emph{subnetwork capacities} $c_{i} \ge 0$ such that
  \begin{description}
  \item [(R1)]
    For any edge $e\in\edges$, 
    \begin{align}
      \omega({e}) \ge \sum_{i: e \in \set{T}_{i}} c_{i}.
    \end{align}
  \item [(R2)] For any $i$ and $u\in
    \destinationLocation(\nu(\set{T}_{i}))$, there exists
    $e\in\set{T}_{i}$ such that $u \in \head{e}$. In other words, the
    node $u$ is on the routing subnetwork.
  \item [(R3)]
    For any source $s\in\sessions$,  
    \begin{align} 
      \lambda({s}) = \sum_{i: \nu(\set{T}_{i}) = s } c_{i}.  
    \end{align}
%
%
  \end{description}
\end{df}

Clearly, these three conditions are not chosen arbitrarily but have a
meaning in practice. Suppose $(\lambda,\omega)$ is \zach\ subject
to a routing constraint. This tuple corresponds to a zero-error
routing solution defined as follows:
Assume without loss of generality that $\lambda(s), \omega(e)$ and
$c_{i}$ are all positive integers.  For each $s\in\sessions$, let the
source message $Y_{s}$ be a $q$-ary row vector of length
$\lambda(s)$. For each $i$, one can use the routing subnetwork
$\set{T}_{i}$ to transmit $c_{i}$ $q$-ary symbols of $Y_{s}$ from
the source nodes (which have access to the source $Y_{s}$) to all sink
nodes $u\in\destinationLocation(s)$.  By (R2), it is guaranteed that
all sink nodes receive all $\lambda(s)$ $q$-ary symbols of
$Y_{s}$ and hence can decode $Y_{s}$. Furthermore, a link $e\in\edges$
is used in the routing subnetwork $\set{T}_{i}$ if $e\in
\set{T}_{i}$. Therefore,
\[
\sum_{i: e\in\set{T}_{i}} c_{i}
\]
is the total number of $q$-ary symbols that have been transmitted on
link $e$. Clearly, the rate-capacity tuple $(\lambda,\omega)$ is fit for
this routing based solution.

In this routing solution, a source node does not perform any coding,
except for partitioning a source message into several independent
segments, and forwarding each segment via a routing subnetwork to the
corresponding sink nodes. This corresponds to the usual concept of
routing in networks consisting of point-to-point links.  For
successful decoding, a sink node must receive every segment of the
source message from the required sources.

In the following, we consider a slight generalisation of the concept
of routing, where source nodes can encode source messages into
\emph{correlated} segments (corresponding to intra-session coding). By
doing so, we can weaken the conditions (R2) and (R3).
\begin{df}[Generalised routing constraint]\label{df:generalisedrouting}
  A rate-capacity tuple $(\lambda,\omega)$ is called \zach\ subject to
  a \emph{generalised routing constraint} if there exists a collection
  of routing subnetworks $\set{T}_{i}$ and \emph{subnetwork
    capacities} $c_{i} \ge 0$ satisfying (R1) and the following
  condition:
  \begin{description}
  \item [(R2$\mbox{}^\prime$)]
    for any source $s\in\sessions$ and any sink node $u\in\destinationLocation(s)$,  
    \begin{align} 
      \lambda({s}) \le \sum_{i: \incoming(u) \cap \set{T}_{i} \neq \emptyset \text{ and }\nu(\set{T}_{i}) = s } c_{i}.  
    \end{align}
  \end{description}
\end{df}
Again, each \zach\ tuple $(\lambda,\omega)$ subject to a generalised
routing constraint is fit for a zero-error routing scheme as follows:
Let $Y_{s}$ be a $q$-ary row vector of length $\lambda(s)$. Instead of
partitioning a source message $Y_{s}$ into independent pieces, one can
encode (e.g. using simple codes for erasure channels) $Y_{s}$ into
$\sum_{i: \nu(\set{T}_{i}) =s} c_{i}$ $q$-ary symbols such that any
$\lambda(s)$ encoded symbols can reconstruct $Y_{s}$ with no
error. These  $\sum_{i: \nu(\set{T}_{i}) =s} c_{i}$
symbols will be forwarded via the routing subnetworks $\set{T}_{i}$
(where $\nu(\set{T}_{i})=s$) to sink nodes in
$\destinationLocation(s)$. As before, all intermediate network nodes
 perform only store-and-forward operations. The condition
(R2$\mbox{}^\prime$) then guarantees that each sink node
$u\in\destinationLocation(s)$ receives at least $\lambda(s)$ coded
symbols of $Y_{s}$. Hence, the node can decode $Y_{s}$ without 
error.

In the following, we will characterise the set of 0-achievable  tuples
subject to a (generalised) routing constraint using a similar
framework as developed in Sections \ref{sec:part1} and
\ref{sec:part2}. Developing a characterisation within this same
framework provides a convenient way to evaluate how a routing
constraint may reduce the set of \zach\ tuples.  We should point
out that we are not the first to characterise  \zach\
rate-capacity tuples subject to (generalised) routing constraints. In
fact, if
\[
|\head{e}| = 1, \quad \forall  e\in\edges,
\]
then the characterisation of \zach\ rate-capacity tuples subject to
(generalised) routing constraint can be obtained by solving variations
of the fractional Steiner tree packing
problem~\cite{Wu.Chou.ea04comparison}.  Our characterisation is
however unified with the entropy function formulation used for network
coding and highlights the differences (and similarities) between
different characterisations with or without (generalised) routing
constraints.

So far we have seen that entropy functions and representable entropy
functions were the key ingredients in characterising the capacity
regions for general network codes and for linear codes.  For networks
with routing constraints, we introduce \emph{almost atomic functions}
which in Theorem~\ref{thm:routing1} below will provide the corresponding
characterisation of the set of \zach\ tuples.
\begin{df}[Atomic rank function] 
A rank function $h\in\set{H}[\sessions\cup\edges]$ is called \emph{atomic}  if 
there exists  $\set{T} \subseteq \sessions\cup\edges $ such that  
\begin{align}\label{eq:atomic}
h  (\beta) = 
\begin{cases}
1 & \text{ if } \beta\cap \set{T}   \neq \emptyset \\
0 & \text{ otherwise. }
\end{cases}
\end{align}
It is called  \emph{almost atomic} if it can be written
\[
h = \sum_{i} c_{i} h^{i}
\]
where for all $i$,  $c_{i} \ge 0$ and $h^{i}$ is atomic.
\end{df}

\def\aa{{\sf AA}} 

Let $\Gamma_{\aa}(\problem)$, or simply $\Gamma_{\aa}$, be the set of
all almost atomic rank functions in $\set{H}[\sessions\cup\edges]$. It
can be easily proved that $\Gamma_{\aa}$ is a closed and convex cone
contained in $\Gamma^{*}$. Thus, almost atomic rank functions are
entropic.
 
\begin{prop}\label{prop:routingequivalence}
  Let $h$ be an atomic function such that there exists nonempty subset $\set{T}
  \subseteq \sessions \cup \edges$ and \eqref{eq:atomic} holds.  Then
  $\set{T}$ is a routing subnetwork (for network coding problem
  $\problem$) if and only if
\[
h \in  \set{C}_{\net}(\problem)\cap\set{C}_{\indep}(\problem).
\]
\end{prop}
\begin{proof}
We first prove the \emph{only-if} part.  
Let $\set{T}$ be a routing subnetwork. 
By definition, $|\set{T} \cap \sessions | =1 $. 
It can be verified directly from definition that
\[
h(\sessions) = \sum_{s\in\sessions}h(s).
\]
Hence, $h\in\set{C}_{\indep}(\problem)$.
It remains to prove that $h\in \set{C}_{\net}(\problem)$. 

For any $e\in\edges$, if $e\not\in\set{T}$, then 
$h(\incoming(e), e) = h(\incoming(e))$ by \eqref{eq:atomic}. 
On the other hand, suppose $e\in\set{T}$. 
 Again as $\set{T}$ is a routing subnetwork,   $\incoming(e) \cap \set{T}$ is nonempty.
Thus,  
\[h(\incoming(e), e) = 1 = h(\incoming(e)).\] Hence, $h(\incoming(e),
e) = h(\incoming(e))$ for all $e\in\edges$. Consequently, $h \in
\set{C}_{\net}$. The only-if part follows.

Now, we will prove the \emph{if}-part. 
 Suppose  $e \in \set{T}\cap\edges $. Then $h(e) =1$.  
As $h\in\set{C}_{\net}$,  
\begin{align}
 h( e, \incoming(e) ) = h(\incoming(e) )
\end{align}
and consequently $h(\incoming(e)) = 1$. By \eqref{eq:atomic},   
\begin{align}
\set{T} \cap \incoming(e) \neq \emptyset
\end{align}
and condition 2) of Definition~\ref{df:routingsubnetwork} is satisfied.

As $h\in\set{C}_{\indep}$, it can verified directly that $| \set{T}
\cap \sessions | \le 1 $.  Since the network $\graph$ is acyclic and
there are only finite number of links, there must exist at least one
$s\in\sessions$ such that $s\in\set{T}$. Consequently, $| \set{T} \cap
\sessions | =1 $. Hence condition 1) of
Definition~\ref{df:routingsubnetwork} is satisfied and the proposition
follows.
\end{proof}

\begin{thm}[Routing capacity]\label{thm:routing1}
  A rate-capacity tuple $(\lambda,\omega)$ is \zach\ subject to a
  routing constraint if and only if
\[
(\lambda,\omega) \in \major(\proj_{\problem} [  \Gamma_{\aa}  \cap \set{C}_{\net} \cap \set{C}_{\de} \cap\set{C}_{\indep} ]).
\]
\end{thm}
\begin{proof}
  We will first prove the \emph{only-if} part.  Suppose
  $(\lambda,\omega)$ is \zach\ subject to a routing constraint.  By
  Definition~\ref{df:routingadmissible}, there exists a collection of
  routing subnetworks $\set{T}_{i}$ and nonnegative real numbers
  $c_{i}$ such that conditions (R1) -- (R3) hold.
 
  By Proposition \ref{prop:routingequivalence},   each $\set{T}_{i}$ is associated with  an atomic rank function  
  $h^{i} \in  \set{C}_{\net} \cap \set{C}_{\indep}\cap  \Gamma_{\aa} $
  such that 
  \begin{align}
    h^{i}(\beta) = 
    \begin{cases}
      1 & \text{ if } \beta\cap \set{T}_{i} \neq \emptyset \\
      0 & \text{ otherwise. }
    \end{cases}\label{eq:5.74}
  \end{align}
  For all sink nodes $u\in\destinationLocation(\nu(\set{T}_{i})
  )$, (R2) implies that $h^{i}(\incoming(u))=1$. Hence, $1 =
  h^{i}(\incoming(u)) = h^{i} (\incoming(u), \nu(\set{T}_{i}) )$. 
On the other hand, if $s \neq \nu(\set{T}_{i})$, then $s\not\in \set{T}_{i}$ and 
$  h^{i}(\incoming(u)) = h^{i} (\incoming(u), \nu(\set{T}_{i}) ) $.  Consequently, $h^{i} \in \set{C}_{\de}$.
  
  Let $h = \sum_{i} c_{i}h^{i}$.
  Since 
  \[
  h^{i}\in\Gamma_{\aa}  \cap \set{C}_{\net} \cap  \set{C}_{\de}\cap\set{C}_{\indep}
  \]
  for all $i$, $h$ is also in $\Gamma_{\aa} \cap \set{C}_{\net} \cap
  \set{C}_{\de}\cap\set{C}_{\indep}$.  Finally, (R2) and (R3) imply
  that for any $s\in\sessions$
  \begin{align}
    \lambda({s}) & = \sum_{i: \nu(\set{T}_{i})=s} c_{i} \\
    & = \sum_{i} c_{i} h^{i}(s) \\
    &= h(s).
  \end{align}
  Similarly,  (R1) implies   $\omega({e}) \ge h(e)$ for all $e\in\edges$. 
  Thus, $(\lambda,\omega) \in \major(\proj_{\problem}[h])$ and the only-if part follows.

  Now, we will prove the \emph{if-part}.  It is easy to prove that if
  $(\lambda,\omega)$ is \zach\ subject to a routing constraint, then
  all tuples in $\major(\lambda,\omega)$ are also \zach. 
  Therefore, it is sufficient to prove that $\proj_{\problem}[h]$ is
  \zach\ subject to a routing constraint for all
  \[
  h \in \Gamma_{\aa}  \cap \set{C}_{\net} \cap \set{C}_{\de} \cap \set{C}_{\indep}.
  \]

  Since $h$ is almost atomic, there exist atomic functions 
  \begin{align}
    h^{i} (\beta) = 
    \begin{cases}
      1 & \text{ if } \beta\cap \set{T}_{i}  \neq \emptyset \\
      0 & \text{ otherwise. }
    \end{cases}\label{eq:5.78}
  \end{align}
  such that   $h = \sum_{i} c_{i}h^{i}$.
  
  As each $h^{i}$ is entropic (and hence polymatroidal) and $c_{i}$ is nonnegative for all $i$,  
  $$\sum_{i} c_{i}h^{i}\in\set{C}_{\net}\cap \set{C}_{\de}\cap \set{C}_{\indep}$$ implies that   
  \[
  h^{i}\in \set{C}_{\net}\cap \set{C}_{\de}\cap \set{C}_{\indep}, \quad \forall i .
  \]
  By  Proposition \ref{prop:routingequivalence}, each $\set{T}_{i}$ is in fact a routing subnetwork.  Also, $h^{i} \in \set{C}_{\de}$ implies that for any 
  $u \in \destinationLocation(\nu(\set{T}_{i}))$, 
  \[
  h(\incoming(u)) = h(\incoming(u) , \nu(\set{T}_{i})) = 1 .
  \]
  This implies $\incoming(u) \cap \set{T}_{i} \neq \emptyset$ and hence  (R2) is satisfied.
  
For any $s\in\sessions$, and $u\in\destinationLocation(s)$,
\begin{align}
  h(s) & = \sum_{i} c_{i}h^{i} (s) \\
  & \nequal{(i)} \sum_{i:   \nu(\set{T}_{i}) = s} c_{i}
\end{align}
where ($i$) follows from the fact that $h^{i}(s)=0$ if $\nu(\set{T})
\neq s$.  Hence (R3) is satisfied. Condition (R1) can also be proved
directly. The if-part is then proved.
\end{proof} 

Using a similar approach as in Theorem \ref{thm:routing1}, we can also
characterise the set of \zach\ rate-capacity tuples subject to the
generalised routing constraint.
\begin{thm}[Generalised routing capacity]\label{thm:grouting}
  Consider a network coding problem $\problem$.  A rate-capacity tuple
  $(\lambda,\omega)$ is 0-achievable subject to the generalised routing
  constraint of Definition~\ref{df:generalisedrouting} if and only if
  \[
  (\lambda,\omega) \in \major(\proj^{*}_{\problem} [  \Gamma_{\aa} \cap \set{C}_{\net} \cap \set{C}_{\indep} ] ).
  \]
  where
  \begin{align}
    \proj^{*}_{\problem}[h](s) &\defined  \min_{u\in\destinationLocation(s)} {h}( s \wedge \incoming(u)) \\
    \proj^{*}_{\problem}[h](e) &\defined \; h(e) .
  \end{align}
\end{thm}
\begin{proof}
  Starting with the \emph{only-if} part, suppose $(\lambda,\omega)$ is
  \zach\ subject to the generalised routing constraint. By
  Definition~\ref{df:generalisedrouting}, there exists a collection of
  routing subnetworks $\set{T}_{i}$, and nonnegative constants $c_{i}$
  such that conditions (R1) and (R2$^{\prime}$) hold.  Each
  $\set{T}_{i}$ is associated with a rank function $h^{i} \in
  \Gamma_{\aa} \cap \set{C}_{\net}\cap\set{C}_{\indep}$ defined as in
  \eqref{eq:5.74}.  Let
  \[
  h = \sum_{i} c_{i} h^{i}. 
  \]
  Again, (R1) implies that
  \[
  \omega(e) \ge h(e), \quad \forall  e\in\edges.
  \]
  By (R2$^{\prime}$), for any $s\in\sessions$ and $u\in\destinationLocation(s)$,
  \begin{align}
    \lambda({s}) & \le \sum_{i: \incoming(u) \cap \set{T}_{i} \neq \emptyset \text{ and } s\in \set{T}_{i} } c_{i} \\
    & \nequal{(a)} \sum_{i } c_{i} {h^{i}}(s\wedge \incoming(u)) \\
    & =  {h}(s\wedge \incoming(u)) \label{eq:5.86}
  \end{align}
  where $(a)$   follows from that 
  \[
  {h^{i}}(s\wedge \incoming(u)) = 
  \begin{cases}
    1 & \text{ if } \incoming(u) \cap \set{T}_{i} \neq \emptyset \text{ and } s \in \set{T}_{i}\\
    0 & \text{ otherwise.}
  \end{cases}
  \]
  
  As \eqref{eq:5.86} holds for all $u\in\destinationLocation(s)$, 
  we have
  \begin{align}
    \lambda({s}) & \le  \proj^{*}_{\problem}[h](s). 
  \end{align}
  Thus, $(\lambda,\omega) \in \major(\proj^{*}_{\problem}[h])$ and the only-if part follows.
  
  Now, we will prove the \emph{if}-part. Let $h \in \Gamma_{\aa} \cap
  \set{C}_{\net} \cap \set{C}_{\indep}$ and $(\lambda,\omega) \in
  \major(\proj^{*}_{\problem}[h])$. As before, we can construct a
  collection of functions $h^{i}$, routing subnetworks $\set{T}_{i}$
  and positive constants $c_{i}$ such that $h = \sum_{i} c_{i}h^{i}$
  and \eqref{eq:5.78} holds.  By definition,
  \[
  \omega({e}) \ge h(e), \: \forall  e\in\edges.
  \]
  and for any $s\in\sessions$ and $u\in \destinationLocation(s)$,
  \begin{align}
    \lambda({s}) & \le  \proj^{*}_{\problem}[h](s) \\
    & \le  {h}(s\wedge \incoming(u)) \\
    & = \sum_{i: \incoming(u) \cap \set{T}_{i} \neq \emptyset \text{ and } s\in\set{T}_{i}  } c_{i}. 
  \end{align}
  Then both (R1) and (R2$^{\prime}$) are satisfied and the result follows.
\end{proof}

\subsection{Heterogeneous networks: Partial routing constraints}\label{sec:routingcoding}
In the previous section, we considered two varieties of routing
schemes defined by routing subnetworks. In those schemes, each
subnetwork is dedicated to sending a segment of data from a source to
its respective sinks. Intermediate network nodes can only perform
store-and-forward operations to forward the same data segment across a
subnetwork.  As only store-and-forward operations are performed, the
computational requirements for intermediate nodes are relatively low.
Despite this advantage, such routing-based schemes may suffer loss in
throughput, as evidenced by the now famous example of the butterfly
network~\cite{Ahlswede2000Network}. In some cases, this loss can be
significant.

In this section, we will consider more advanced schemes where some subsets of
the intermediate nodes have sufficient computational resources to
permit more sophisticated data processing in order to increase the
throughput. Thus the network now consists of two types of nodes:
routing nodes and network coding nodes. As demonstrated by the
butterfly network example, there are known instances where the maximum
possible throughput can in fact be achieved with only one network
coding node, and all other nodes performing routing.
 
The aim of this section is to extend out methodology to such
heterogeneous networks.  As a first step, we need to clarify the
concept of store-and-forward.  Figure \ref{fig:eg1} is a subnetwork of
$\graph$ such that the node $v$ is a ``routing node'' (where only
store-and-forward operation is allowed). The node has two incoming
links and one outgoing link.  Suppose that $v$ receives
$(b_{0},b_{1})$ from the incoming link $e_{1}$ and $(b_{2},b_{3})$
from link $e_{2}$. A natural question is: If $v$ can only perform
store-and-forward operations, which types of outgoing message it can
send?

\newfigure{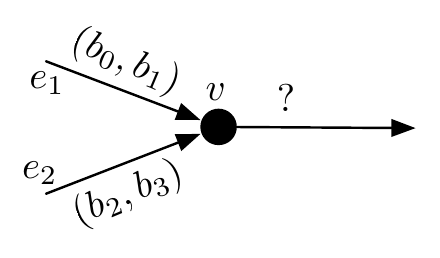}{Partial routing contraint.}{fig:eg1}{1.2}

Naturally, we should allow the routing node $v$ to send $b_{i}$ for
any $i=0,1,2,3$, but not $b_{0} \oplus b_{3}$.  The question however
is if $v$ can send $b_{0} \oplus b_{1}$ which is a function of the
incoming message from $e_{1}$?

In this paper, we will assume that $v$, as a routing node, is in fact
permitted to perform intra-edge coding and send $b_{0} \oplus
b_{1}$. We do not allow inter-edge coding across different incoming
links.  Using this slightly generalised definition of routing, we can
once again use the tools developed earlier to characterise
\zach/\vach\ rate-capacity tuples for heterogeneous networks with
``partial routing'' constraints.


\begin{df}[Routing nodes]
  With respect to a network code $\{ Y_{f}, f\in\sessions\cup\edges
  \}$, an intermediate node $v$ is said to be a \emph{routing node} if
  for all outgoing links $e$ of $v$ (i.e., $\tail{e} = v$), there
  exist auxiliary random variables
  \[
  \{Y_{f,e},  f\in \incoming(e) \}
  \]
  such that 
  \begin{align}
    H\left(Y_{e} \mid Y_{f,e}, f\in \incoming(e)\right) & = H\left(Y_{f,e}, f\in \incoming(e)\mid Y_{e} \right) = 0 \label{eq:df:route1}\\
    H\left(Y_{f,e} \mid Y_{f}\right) & = 0  , \: \forall f\in \incoming(e).\label{eq:df:route2}
  \end{align}
  In other words, the outgoing message $Y_{e}$ is formed by a set $\{
  Y_{f,e}, f\in \incoming(e) \}$ such that each element $Y_{f,e}$ is a
  function of the incoming message $Y_{f}$ from the link $f$.
  \emph{Routing links} are defined as outgoing links from a routing
  node.  
\end{df}

Let $\routing \subseteq \edges $ be the set of all routing links. In
other words, $e\in \routing $ if and only if $\tail{e}$ is a routing
node. We refer to $\routing$ as a \emph{partial routing constraint}.
\begin{df}[Network code with partial routing constraints]
  A network code satisfying the partial routing constraint $\routing$ is a set of random variables 
  \[
  \{ Y_i,  i\in\sessions\cup\edges \}  \cup \{Y_{j,e}, e\in \routing , j\in \incoming(e) \}.
  \]
  such that $\{ Y_i, i\in\sessions\cup\edges \} $ is an ordinary
  network code according to Definition \ref{df:nc} and in addition,
  \eqref{eq:df:route1} and \eqref{eq:df:route2} hold for all
  $e\in\routing$. We refer to such a code as a $\routing$-network code
\end{df}

To go along with our definition of a network code with partial routing
constraints, we need to update our definition of fitness.
\begin{df}[Fitness of a network code with partial routing]
  A rate-capacity tuple $(\lambda,\omega)$ is \emph{fit} for a $\routing$-network
  code
  $\{ Y_i,  i\in\sessions\cup\edges \}  \cup \{Y_{j,e}, e\in \routing , j\in \incoming(e) \}$
  if
  \begin{align}
    \lambda(s) & \le \log |\support(Y_{s})|, \forall s\in\sessions \\
    \omega(e) & \ge \sum_{f\in \incoming(e)} \log |\support(Y_{f,e})|, \forall e\in 
    \routing 
    \label{eq:df:admiss:routing} \\
    \omega(e) & \ge \log |\support(Y_{e})|, \forall e \not\in\routing. 
  \end{align}
\end{df}

Note that we use \eqref{eq:df:admiss:routing} rather than \eqref{eq:6}
to highlight that the outgoing message $Y_{e}$ will not be jointly
compressed by the routing node.  The set of \zach\ and \vach\
rate-capacity tuples subject to a partial routing constraint
$\routing$ can be defined similar to Definitions \ref{df:admissible} and
\ref{df:achievable}.
 
Our approach for characterisation of the set of \zach\ or \vach\
rate-capacity tuples for $\routing$-codes is to transform the problem
with partial routing constraints into an equivalent unconstrained
problem $(\graph^{\dagger},\multicastRequirement)$.

Given a network coding problem $(\graph,\multicastRequirement)$ with
partial routing constraint $\routing$, define
$\graph^{\dagger}\defined (\nodes^{'}, \edges^{'})$ as follows
\begin{enumerate}
\item Add new nodes:   
  \[
  \nodes^{'} = \nodes \cup \{V_{[j,e]}, e\in\routing , j \in \incoming(e) \}.
  \]
\item Add new links:  
\[
\edges^{'} = \edges \cup \{ [j,e], e\in\routing , j \in \incoming(e)  \}
\]
such that  
\begin{align}
\head{[j,e]} & = \tail{e} \\
\tail{[j,e]} & = V_{[j,e]}.
\end{align}
\item Modifying existing link connections: 
For all $f\in\edges$, the set $\head{f}$  is modified as 
\begin{align*}
\left(\head{f} \setminus \{ \tail{e}: \:  e\in\routing, f\in \incoming(e) \} \right)  \cup \{ V_{[f,e]} :\:  e\in\routing, f\in \incoming(e) \}.
\end{align*}
In other words, if a link $f$ was directed to a routing node
$\tail{e}$ for some $e\in\routing$, it will be redirected to the newly
created node $V_{[f,e]}$.
\end{enumerate}

Figure \ref{fig:b} is an example illustrating how to modify a network
to remove the partial routing constraint. In this example, Figure
\ref{fig:b}(a) is one part of the network where $e\in\routing$ is a
routing link. Figure \ref{fig:b}(b) shows how that part of the network
is transformed. In the new network, we no longer impose any routing
constraint.
\begin{figure}[htbp]
\begin{center}
  \subfigure[Original network.]
 {\hspace{5mm}\includegraphics{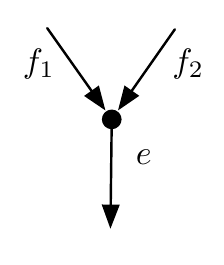}\hspace{5mm}}
 \subfigure[Modified network.]
{\includegraphics{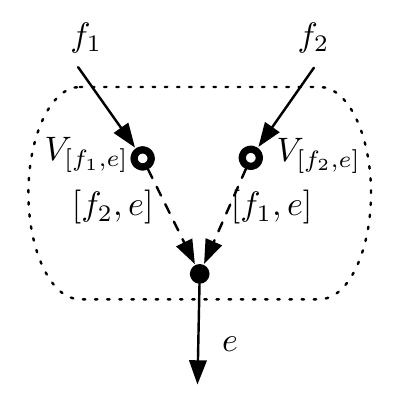}}
  \caption{Removing a routing contraint.}
  \label{fig:b}
\end{center}
\end{figure}
 
\begin{thm}[Network transformation]\label{thm:nettransform}
  A rate-capacity tuple $(\lambda,\omega)$ is \zach/\vach\ with
  respect to a network coding problem $(\graph,
  \multicastRequirement)$ subject to a partial routing constraint
  $\routing$ if and only if there exists a \zach/\vach\ rate-capacity
  tuple $( \lambda^{'} , \omega^{'})$ for the network coding problem
  $(\graph^{\dagger}, \multicastRequirement)$ where
  \begin{align}
    \lambda_{s} & = \lambda^{'}_{s}, \quad \forall s\in\sessions \\
    \omega_{e} & = \omega^{'}_{e}, \quad \forall e\in \edges \setminus \routing \\
    \omega_{e} & = \sum_{f\in \incoming(e)} \omega^{'}_{[f,e]}, \quad \forall e\in  \routing.\label{eq:omegaprime}
  \end{align}
\end{thm}
\begin{proof}
  By direct verification.
\end{proof}
The construction of $\graph^{\dagger}$ together with relationships
between $\lambda/\omega$ and $\lambda'/\omega'$ remove the partial
routing constraint by making the capacity of the links $e\in\routing$
in $\graph^{\dagger}$ sufficiently large such that network coding is
never required at $\head{e}$, which has sufficient capacity to simply
forward all of the incoming messages. Also note that the choice of
$\omega'_{[f,e]}$ are free, apart from the
constraints~\eqref{eq:omegaprime}.

As a corollary of Theorem \ref{thm:nettransform}, all of the results
obtained in the earlier sections can also be applied to network coding
problems with partial routing constraints.

\def\wpattern{{\mathsf W}}
\section{Secure Network Codes}\label{sec:part4}
So far we have considered two classes of constraints on network
codes. In Section~\ref{sec:part2} we considered linear network codes,
which may be attractive for practical implementation. In
Section~\ref{sec:part3} we considered networks where some, or all of
the nodes are constrained to perform only store-and-forward types of
operations. 
Another important class of constraints to consider for network coding
are motivated by security considerations. The objective is to
determine the achievable network coding rates when we require secret
transmission that is impervious to  specified eavesdropping attacks. 

Assume that there are $|\set{R}|$ adversaries in the
network. Adversary $r\in\set{R}$ observes messages transmitted along
links in the set $\set{B}_r \subseteq \edges$ and aims to reconstruct
the set of sources indexed by $\set{A}_{r}\subseteq \sessions$.  We
refer to $\wpattern \defined \{(\set{A}_r, \set{B}_r), r\in\set{R}\}$
as the \emph{wiretapping pattern} of the network.

We will use the notation $\problem =
(\graph,\multicastRequirement,\wpattern)$ to denote the network coding
problem subject to a secrecy constraint, also referred to as a
\emph{secure network coding problem}. Here, $\graph$ and
$\multicastRequirement$ are as usual the network topology and the
connection constraint.
The secure communications objective is to transmit information over a
network satisfying the multicast requirements while simultaneously
ensuring that the eavesdroppers gain no information about their
desired sources.

Before we characterise the set of \zach/\vach\ rate-capacity tuples
for secure network coding, we need to point out a significant
difference between ordinary and secure network codes. Without a
secrecy constraint, the transmitted message on any network link $e$
can be assumed without loss of generality to be a function of the
source inputs and received messages available at the node $\tail{e}$.
However, when secrecy constraints are enforced, it is usually
necessary to encode messages \emph{stochastically} to prevent an
eavesdropper from learning any useful information about its desired
sources.
 
\begin{df}[Stochastic network code] \label{df:stochasticnc}
  A \emph{stochastic network code} is defined by a set of random
  variables
  \begin{align}\label{df:snc}
    \{ Y_{f}, f\in\sessions\cup\edges\cup\nodes \}
  \end{align}
  with entropy function $h$ such that $Y_{s}$ is uniformly distributed
  over its alphabet set for all $ s\in\sessions$ and
  \[ 
  h \in \set{C}_{\indep}(\problem) \cap \set{C}_{\net}(\problem) 
  \]
  where 

  \begin{align}
    \set{C}_{\indep} (\problem) 
   & \define \left\{ \begin{array}{l} \hspace{-0.2cm} h\in\set{H}[\sessions\cup\edges\cup\nodes] : \\ 
        \quad {h}({\set{S}} , \nodes )  = \sum_{s\in\sessions} {h}( s)
        + \sum_{u\in\nodes}
        h(u)  \end{array}\right\} 	  \label{eq:sconstraint1} \\
    \set{C}_{\net} (\problem) &\define 
    \left\{ 
      \begin{array}{l}
        \hspace{-0.2cm} h \in\set{H}[\sessions\cup\edges\cup\nodes] : \\ 
        \quad h\left( s\mid   \incoming(e) , \tail{e} \right)  = 0,   \forall  e\in\edges 
      \end{array}
    \right\}. \label{eq:sconstraint2} 
  \end{align}
\end{df}

In the definition, $\{Y_{s}, s\in\sessions\}$ and $\{Y_{e},
e\in\edges\}$ are again the set of random sources (indexed by
$s\in\sessions$) and the set of messages (transmitted on hyperedges
$e\in\edges$). The random variables $\{ Y_{u}, u\in\nodes \}$ can be
thought of as the \emph{random keys} available at nodes $u\in\nodes$
for stochastic encoding. Specifically, each link $e\in\edges$ is
associated with a local encoding function such that
\begin{align}
  Y_{e} = \phi_{e} (Y_{i}, i\in\incoming(e),Y_{\tail{e}}).
\end{align}
Clearly, we have 
\[
H(Y_{e} \mid Y_{i}, i\in\incoming(e), Y_{\tail{e}})
\]
and hence \eqref{eq:sconstraint2}. Furthermore, we want to point out
that the random keys $\{ Y_{u}, u\in\nodes \}$ are \emph{not} like the
usual secret keys that are privately shared between nodes in
Shannon-style secure communications. Instead, they are locally (and
hence independently) generated at each node.  In other words, there
are \emph{no} correlated or common keys shared privately between nodes
in advance.  Therefore, we will assume that $\{Y_{f} , f\in\sessions
\cup \nodes \}$ are mutually independent, and require
\eqref{eq:sconstraint1} to hold.
 
%

\subsection{Weak Secrecy}

\begin{df}[Weak secrecy]
  For a secure network coding problem
  $\problem=(\graph,\multicastRequirement,\wpattern)$, a rate-capacity
  tuple $(\lambda,\omega)$ is called \zach\ subject to a \emph{weak
    secrecy} constraint if there exists a sequence of stochastic
  network codes
  \[
  \networkCoding^{n} =
  \{\edgeRV_f^{n}: f \: \in \edges \cup \sessions\cup\nodes\}
  \]
  and positive
  normalising constants $c_n$ such that
  \begin{enumerate}
  \item[(S1)]  for all $e\in \edges$
    and $s\in\sessions$,
    \begin{align}
      \lim_{n\to\infty}  c_{n} \log |\support(\edgeRV_e^{n})|   & \le   \edgeRate(e), \\
      \lim_{n\to\infty} c_{n} \log |\support(\edgeRV_s^{n})|   & \ge
      \inputRate(s).
    \end{align}
  \item[(S2)] for $s\in\sessions$ and $u \in \destinationLocation(s)$,  
    \[
    H( Y_s^{n}\mid   \edgeRV_f^{n},  f \in \incoming(u) ) = 0.
    \]
 \item[(S3)] For all $r\in \set{R}$, 
\begin{align}
 \lim_{n\to\infty}  c_{n } {I(Y^{n}_{\set{A}_{r}} ; Y^{n}_{\set{B}_{r}})}   =0. 
\end{align}    
\end{enumerate}

Similarly, a rate capacity tuple $(\inputRate , \edgeRate ) $ is
called \vach\ subject to a weak secrecy constraint if there exists a
sequence of network codes
\[
\networkCoding^{n} =
\{\edgeRV_f^{n}: f \: \in \edges \cup \sessions\cup\nodes\}
\]
satisfying (S1) and (S3) and the following condition (S2$^{\prime}$):
  \begin{enumerate}
  \item[(S2$^{\prime}$)] for $s\in\sessions$ and $u \in
    \destinationLocation(s)$, there exists decoding functions
    $g_{s,u}^{n}$ such that
    \begin{align*}
      \lim_{n\to\infty} \Pr(Y_{s}^{n} \neq g_{s,u}^{n}(\edgeRV_f^{n}:
      f \in \incoming(u) )) = 0 .
    \end{align*}
  \end{enumerate}
\end{df}

\def\secure{{\sf S}}

The following theorem can be proved by using the same technique as in
Corollary \ref{cor:ourbd}. For brevity, we state the theorem without
proof.
\begin{thm}[Outer bound]\label{thm:mainsecure}
  Consider any secure network coding problem
  $\problem=(\graph,\multicastRequirement,\wpattern)$ subject to a
  weak secrecy constraint.  Let
\begin{align}
\set{C}_{\de} (\problem)& \define \left\{ 
	\begin{array}{l}
\hspace{-0.2cm} h \in\set{H}[\sessions\cup\edges\cup\nodes]: h\left( s\mid  \incoming(u) \right)  = 0, \\  
  	\hspace{2.5cm} \forall  s\in\sessions, u\in\destinationLocation(s) 
 	\end{array}
\right\}, \label{eq:sconstraint3} \\
\set{C}_{\sec} (\problem)& \define \left\{ h\in\set{H}[\sessions\cup\edges\cup\nodes] : h\left(\set{A}_{r} \wedge \set{B}_{r}\right) =0, \forall r\in\set{R} \right\}.\label{eq:sconstraint4}
\end{align}
If a rate-capacity tuple  $(\lambda, \omega) \in {\chi} (\problem)  $ is $\epsilon$-achievable,  then there exists 
\[
h\in  \bar \Gamma^{*}(\sessions\cup\edges\cup\nodes)  \cap \set{C}_{\indep}(\problem) \cap \set{C}_{\net}(\problem)  \cap \set{C}_{\de}(\problem)   \cap \set{C}_{\secure}(\problem) 
\]
such that
\begin{align}
\lambda(s) & \le h(s) , \: \forall s\in\sessions, \\
\omega(e) & \ge h(e), \: \forall e\in\edges.
\end{align}
Or equivalently, 
\[
(\lambda,\omega) \in \major(\proj_{\problem}[\bar \Gamma^{*}(\sessions\cup\edges\cup\nodes)  \cap \set{C}_{\indep}(\problem) \cap \set{C}_{\net}(\problem)  \cap \set{C}_{\de}(\problem)   \cap \set{C}_{\secure}(\problem)]).
\]
\end{thm}

The condition \eqref{eq:sconstraint3} corresponds to the decoding
constraint, requiring that any node $u\in\destinationLocation(s)$ can
decode the source $s$ with vanishingly small error. The condition
\eqref{eq:sconstraint4} is the secrecy constraint, ensuring that an
adversary can learn no information about the sources it is interested
in.

Unlike in Theorem \ref{thm:arbitrarymain}, we do not claim tightness
of the outer bound even for colocated sources. This is because secure
network nodes may locally generate random keys for the purpose of
stochastic encoding. These keys, to a certain extent, behave like
sources (with no corresponding sink nodes), and hence the colocated
source condition fails, even if the actual sources are colocated.

\subsection{Strong Secrecy}
Weak secrecy requires that the amount of information leakage vanishes
asymptotically after normalisation. In other words, the amount of
information leakage is negligible (when compared with the size of the
source messages). We can also consider a \emph{strong secrecy
  constraint}, where we require the information leakage to be exactly zero.
\begin{df}[Strong secrecy]
  A rate-capacity tuple $(\lambda,\omega)$ is \zach\ subject to a
  \emph{strong secrecy} constraint if there exists a sequence of network
  codes
  \[
  \networkCoding^{n} =
    \{\edgeRV_f^{n}: f \: \in \edges \cup \sessions\cup\nodes\}
    \]
     and positive
    normalising constants $c_n$ 
    satisfying  (S1), (S2) and the following condition 
  \begin{enumerate}
  \item[(S3$^{\prime}$)]  $I(Y^{n}_{\set{A}_{r}} ; Y^{n}_{\set{B}_{r}})=0$ for all $n$ and $r\in\set{R}$.   
 \end{enumerate}
 Similarly, it is \vach\ subject to a strong secrecy constraint if the
 sequence of codes satisfies (S1), (S2$^{\prime}$) and
 (S3$^{\prime}$).
\end{df}

In general, it is very hard to characterise the set of achievable
rate-capacity tuples subject to a strong secrecy constraint, even
implicitly via entropy functions. However, under the additional
constraint of linearity, the set of \zach\ rate-capacity tuples can in fact be
characterised implicitly via the use of representable functions.

\begin{df}[Strongly secure linear network codes]
  Let
  \begin{align} 
    \{\edgeRV_f: f \: \in \edges \cup \sessions \cup \nodes\}
  \end{align}
  be a stochastic network code (according to Definition~\ref{df:stochasticnc}) for a
  secure network coding problem $\problem$ on a network $\graph=(\set{V},\set{E})$.
  The code is called \emph{$q$-linear} (or simply linear) if it
  satisfies the following conditions:
  \begin{enumerate}
  \item For $f\in\sessions\cup\nodes$, $Y_{f}$ is a random row vector such that
    each of its entries is selected independently and uniformly over
    $\field(q)$.
\item For any $e\in\edges$, there exists a linear function $\networkcoding_e$ such that 
\[
Y_{e} = \networkcoding_e (Y_{\incoming(e)}, Y_{\tail{e}}).
\]
 \end{enumerate}
A network coding problem is said to be subject to a \emph{$q$-linearity
  constraint} if only $q$-linear network codes are allowed.
\end{df}

\begin{thm}[Strongly secure linear network codes]
  Consider a secure network coding problem $\problem$ where
  $|\sourceLocation(s)|=1$ for all
  $s\in\sessions$. 
  A rate-capacity tuple $(\lambda,\omega)$ is 0-achievable subject to
  $q$-linearity and strong secrecy if and only if
\[
(\lambda, \omega)  \in \major(  \proj_{\problem} [ \bar\Upsilon^{*}_{q} \cap \set{C}_{\indep} \cap \set{C}_{\net}  \cap \set{C}_{\de}   \cap \set{C}_{\secure} ]).
\]
\end{thm}
\begin{proof}
  We first prove the \emph{only-if} part.  Suppose $(\lambda, \omega)$
  is 0-achievable subject to linearity and strong secrecy constraints.
  By definition, there exists a sequence of linear codes
\[
 \{ Y_{f}^{n}, f\in\sessions \cup\edges\cup\nodes  \}
 \]
 with entropy function $h^{n}$ and normalising constants $c_n>0$ such that
 (S1), (S2) and (S3$^{\prime}$) hold. 
By  (S2) and (S3$^{\prime}$), 
 \[
 h^{n} \in \Upsilon^{*}_{q} \cap \set{C}_{\indep} \cap \set{C}_{\net}  \cap \set{C}_{\de}   \cap \set{C}_{\secure}.
 \] 
And hence, $c_{n} h^{n}  \in \bar\Upsilon^{*}_{q} \cap \set{C}_{\indep} \cap \set{C}_{\net}  \cap \set{C}_{\de}   \cap \set{C}_{\secure}$.
Consequently, 
\[
(\lambda,\omega) \in \major(\proj_{\problem}[\bar\Upsilon^{*}_{q} \cap \set{C}_{\indep} \cap \set{C}_{\net}  \cap \set{C}_{\de}   \cap \set{C}_{\secure}]).
\]  
and the only-if part follows. 

Now let 
\[
h \in  \bar\Upsilon_{q}^{*} \cap \set{C}_{\indep} \cap \set{C}_{\net}  \cap \set{C}_{\de}   \cap \set{C}_{\secure} .
\]
To prove the \emph{if}-part, it suffices to prove that $\proj_{\problem}[h]$ is 0-achievable.
As in the proof of Theorem \ref{thm:mainlinear}, 
there exists a sequence of  $q$-representable functions
\begin{align} h^{n} & \in  \Upsilon_{q}^{*} \cap \set{C}_{\net}  \cap \set{C}_{\de} \cap \set{C}_{\indep}
\end{align}
 and positive scalars $c_{n}$ such that 
\begin{align}
h& = \lim_{i\to\infty} c_{n} h^{n} . 
\end{align} 
For each $n$, $h^{n}$ induces a zero-error linear network code 
\begin{align}\label{eq:149}
\{ Y^{n}_{i} , i\in\sessions\cup\edges\cup\nodes \}
\end{align}
such that 
\begin{align}
H\left(Y^{n}_{s}\right)   & = h^{n}(s),\: \forall s\in\sessions \\
H\left(Y^{n}_{e}\right)  & = h^{n}(e),\: \forall e\in\edges.
\end{align}
However, the linear network code \eqref{eq:149} need not be strongly secure (i.e., $h^{n} \in \set{C}_{\secure}$).
In the following, we will create from $h^{n}$ another representable function 
$g^{n}$ such that 
\begin{align}
\lim_{n\to\infty} c_{n} g^{n} & = \lim_{n\to\infty} c_{n} h^{n} = h \\
g^{n} & \in  \Upsilon_{q}^{*}\cap \set{C}_{\indep} \cap \set{C}_{\net}  \cap \set{C}_{\de}   \cap \set{C}_{\secure}.
\end{align}
 
Since $h^{n}$ is $q$-representable, there exists subspaces 
\[
\{ \vV^{n}_{i}, i\in\sessions\cup\edges\cup\nodes \}
\]
 such that for all $\alpha\subseteq \sessions\cup\edges\cup\nodes$, 
\[
h^{n}(\alpha) =\dim \la  \vV^{n}_{j}, j\in \alpha   \ra.  
\]
For each $r\in\set{R}$ and $s\in\set{A}_{r}$, we define
\begin{align}
\vW_{r,s}^{n} \defined \vV^{n}_{s} \cap \la \vV^{n}_{f}, f\in \set{B}_{r} \ra.
\end{align}
Then by direct verification, 
\[
H\left(\vW_{r,s}^{n} \right) = {h^{n}}(s \wedge \set{B}_{r}) 
\]
 and 
hence $\lim_{n\to\infty} c_{n} H\left(\vW_{r,s}^{n} \right) = 0$.
%

Let 
\[
\vW^{n}_{s} \defined \la \vW^{n}_{r,s}, r\in\set{R} \text{ and } s \in\set{A}_{r}  \ra.
\]
By Lemma \ref{lemma:5.1}, for every $s\in\sessions$, there exists a subspace  $\vU^{n}_{s}$  of $\vV_{s}^{n}$ such that 
\begin{align}
\dim \vV^{n}_{s} &= \dim \vU^{n}_{s} + \dim \vW^{n}_{s}, \\
 \{ {\bf 0} \} & =  \vU^{n}_{s} \cap  \vW^{n}_{s} .
\end{align}

For any $s\in\sessions$, let $\sourceLocation(s)$ be the unique source node where the $s^{th}$ source is available. Let
\begin{align}
\vU_{u}^{n} & = \la \vV^{n}_{u}, \vW^{n}_{s}, \sourceLocation(s)=u \ra , \quad \forall u\in\nodes \\
\vU^{n}_{e} & = \vV^{n}_{e}, \quad \forall e\in\edges.
\end{align}

On the other hand, 
\begin{align}
\vU_{s}^{n} \cap \la \vV^{n}_{f}, f\in \set{B}_{r} \ra & \subseteq \vV_{s}^{n} \cap \la \vV^{n}_{f}, f\in \set{B}_{r} \ra \\
& = \vW^{n}_{r,s}.
\end{align}
As 
\[
\vU_{s}^{n} \cap \vW^{n}_{r,s} = \{{\bf 0}\}, 
\]
we have
\[
\vU_{s}^{n} \cap \la \vV^{n}_{f}, f\in \set{B}_{r} \ra = \{{\bf 0}\}.
\]

Since $h^{n} \in \set{C}_{\indep}$, 
\begin{align}
\dim \la  \vV_{s}, s\in\sessions \ra = \sum_{s\in\sessions} \dim \vV_{s}.
\end{align}
As $\vU_{s} \subseteq \vV_{s}$ for $s\in\sessions$,
\begin{align}
\dim \la  \vU_{s}^{n}, s\in\sessions \ra = \sum_{s\in\sessions} \dim \vU_{s}^{n}.
\end{align}

Let $g^{n}$ be the representable function induced by 
\[
\{ \vU_{f}^{n} , f\in\sessions\cup\edges\cup\nodes \}.
\]
Then, it can be directly verified that  
%
%
\begin{enumerate}
\item
$g^{n} \in \Upsilon^{*}_{q}  \cap \set{C}_{\indep} \cap \set{C}_{\net}  \cap \set{C}_{\de} \cap  \set{C}_{\sec}$  where 
$g^{n}$ is the rank function induced by $\{Y^{n}_{f}, f\in \sessions\cup\edges \cup \nodes  \}$, and 

\item   $\lim_{n\to\infty } c_{n} g^{n} = h$.

\end{enumerate}
Consequently,  $\proj_{\problem}[g^{n}]$, and also $\proj_{\problem}[c_{n} g^{n}]$  and $\proj_{\problem}[h]$, are 0-achievable subject to the two constraints. 
\end{proof}

\subsection{Secret Sharing}
In secret sharing~\cite{Shamir1979How-to-share}, a secret is shared
among a set of users $\N$ where each user holds a component of the
secret.  The main objective is to ensure that only specified
legitimate subgroups of users (indexed by a subset $\A$ of $\N$) can
successfully decode the secret.  All other illegitimate subgroups of
users should receive no information about the secret. The collection
of all legitimate subsets $\Omega$ is called the \emph{access
  structure} of the secret sharing problem.
  
We can reformulate a secret sharing problem as a secure network coding
problem $\problem = (\graph, \multicastRequirement, \wpattern)$.  In
this secure network coding problem, there is only one source (the
secret) which is only available at the source node $u^{*}$.  There are
$|\N|$ intermediate nodes, each of which represents a user.  The
transmitted message an intermediate node (or a user) received from the
source corresponds to the component of the secret that it holds. There
are $|\Omega|$ sink nodes indexed by $\{ v_{\alpha}, \alpha \in \Omega
\}$.  The sink node $v_{\alpha} $ is connected to nodes (or users) $i
\in \alpha$ and aims to reconstruct the secret. We also assume that
each $\beta \not\in \Omega$, is associated with an eavesdropper who
can wiretap the set of edges $\{ e_{i} , i\in \beta \} $. The secrecy
constraint implies that all illegitimate subgroups of users have no
information about the secret.

Mathematically, the secure network coding problem is defined  as follows.
\begin{enumerate}
\item $\graph=(\nodes,\edges)$ where $\nodes=\{u^{*}\} \cup\N \cup \{ v_{\alpha}, \alpha \in \Omega \}$ and $\edges = \{e_{i}, f_{i}, i\in \N \}$;
\item for any $i\in\N$, $\tail{e_{i}} = u^{*}$,  $\head{e_{i}} = \{ i \}$, $\tail{f_{i}} = i$ and   $\head{f_{i}} = \{ v_{\alpha} : i \in \alpha \}$; 
\item $\multicastRequirement = (\sessions, \sourceLocation, \destinationLocation)$ where 
($i$) $\sessions = \{1\}$, ($ii$) $\sourceLocation(1) = \{u^{*}\}$ and ($iii$) $\destinationLocation(1) = \{v_{\alpha} , \alpha \in \Omega\}$;
\item $\wpattern = \{ ( 1 , e_{i} , i\in \beta )  : \: \beta \subseteq \N \text{ and } \beta \not\in \Omega \}$.
\end{enumerate} 
By translating a secret sharing problem to a secure network coding
problem, the results obtained in this paper can be applied to secret
sharing.

\section{Challenges in characterising achievable tuples}\label{sec:challenges}

Characterising the set of achievable rate-capacity tuples for a
network coding problem is generally very hard.  So far, there are only
a limited number of scenarios where the sets of \zach/\vach\
rate-capacity tuples have been explicitly determined. One scenario is
when there is only one source, $|\sessions|=1$ and no partial routing
constraint or secrecy constraints. In this case, the set of achievable
rate-capacity tuples is explicitly characterised by the cut-set
bound~\cite{CovTho91}.  If a secrecy constraint is additionally
imposed, the set of achievable tuples can still be determined if \emph{(i)}
all links have unit capacity and \emph{(ii)} the eavesdropper is
\emph{$t$-uniform} in the sense that an eavesdropper can wiretap any $t$
links in the network~\cite{Cai2011Secure,cai2002snc,CaiYeu07}.  In
both cases, linear codes are optimal.

It is natural to wonder whether there is any hope that wide classes of
network coding problems could have simple, explicit characterisations.
In the following two subsections, we will show that even in some very
simple scenarios, finding the set of achievable rate-capacity tuples
can be extremely hard. The first scenario will be an incremental
multicast. The second scenario is a secure multicast.

\subsection{Incremental Multicast}\label{sec:incr-mult}
In a \emph{incremental multicast} problem, sources are totally ordered
such that a receiver who wants to reconstruct source $s$ is also
required to reconstruct all other sources $i$ for $i < s$. {Here, the
  symbol $<$ is defined with respect to the total ordering of the
  sources}.  Incremental multicast is common in multimedia
transmission, where data such as video or audio may be encoded into
multiple layers. A layer can only be used for reconstruction at a
receiver if all its previous layers are also available. This leads
directly to an incremental multicast problem.

We will construct the simplest case of a incremental multicast problem
involving  two layers, colocated at the same source node. Hence, there
are two types of receivers: those which request source $1^{\prime}$, and
those which request both sources ($1^{\prime}$ and $2^{\prime}$).  Even for such a simple setup, we
will show that determining the set of achievable rate-capacity tuples
can be as hard as solving any network coding problem in general.

To prove our claim, we consider any ordinary network coding problem
$\problem = (\graph,\multicastRequirement)$, where sources may or may not be
colocated.  We will transform $\problem$ into a two-layer incremental
multicast problem $\problem^{\dagger} =
(\graph^{\dagger},\multicastRequirement^{\dagger})$ and prove in
Theorem \ref{thm:challenge1} that determining the set of achievable
tuples in the incremental multicast problem $\problem^{\dagger}$ is at
least as hard as determining the outer bound of Corollary
\ref{cor:ourbd} for the problem $\problem$.

The network $\graph^{\dagger}= (\nodes^{\dagger} , \edges^{\dagger})$ is
obtained from its subgraph $\graph$ by adding nodes and hyperedges
\begin{align*} 
  \nodes^{\dagger} &\defined \nodes \cup \{ \gamma_{s},
  \tau_{s,u}, \: s\in\sessions , u\in\destinationLocation(s) \} \cup
  \{\psi, \phi , \eta, \eta^{*}\}\\ \edges^{\dagger} &\defined \edges
  \cup \{a_{s}, b_{s}, c_{s}, s\in\sessions \} \cup \{ d_{s,u} ,
  s\in\sessions , u\in \destinationLocation(s) \}
\end{align*} with connections
\begin{align} \tail{a_{s}} &= \phi \\ \head{a_{s}} &=
  \sourceLocation(s) \cup \{\gamma_{s},\eta^{*}\}. \label{eq:im:as}\\
  \tail{b_{s}} &= \phi \\ \head{b_{s}} &= \{\gamma_{s}, \eta,\psi\} \cup
  \bigcup_{j\neq s} \{\tau_{j,v} , v\in \destinationLocation(j) \} \\
  \tail{c_{s}} &= \gamma_{s} \\ \head{c_{s}} &= \{\eta ,\eta^{*} \} \cup
  \{\tau_{s,v} , v\in \destinationLocation(s) \}\\ \tail{d_{s,u}} &=
  u \\
  \head{d_{s,u}} &= \tau_{s,u}, 
\end{align}
for all $s\in\sessions,
  u\in\destinationLocation(s)$.

In addition to augmenting the network, we also need to define the
connection requirement $\multicastRequirement^{\dagger} =
(\sessions^{\dagger}, \sourceLocation^{\dagger},
\destinationLocation^{\dagger})$.  In our two-layer incremental
multicast problem there are two sources indexed by
\[
\sessions^{\dagger}  \defined  \{1^{\prime}, 2^{\prime} \} .
\]
All the sources are colocated  at the node $\phi$, i.e.,
\begin{align*}
  \sourceLocation^{\dagger}(1^{'})  = & \: \sourceLocation^{\dagger}(2^{'}) = \phi.
\end{align*}
Finally, the destination location mapping $\destinationLocation^{\dagger}$ is defined as 
\begin{align*}
  \destinationLocation^{\dagger}(1^{'}) = & \:\{ \tau_{s,u}, s\in\sessions, u\in\destinationLocation(s) \} \cup \{ \eta,\eta^{*},\psi\} ,\\
  \destinationLocation^{\dagger}(2^{'}) = & \: \{ \eta, \eta^{*}\}.
\end{align*}
Figure \ref{fig:challenge1} exemplifies how to convert an ordinary
network coding problem $\problem$ (which has two sources) into a two
layers incremental multicast problem. Here, a source will be denoted
by a double circle, and a sink by an open square. The label beside a
source or a sink denotes the index of the sources which are available
or are required at the node. Note that in the figure, the sink nodes
$u$ and $u'$ in the original problem $\problem$ are no longer sink
nodes in the incremental multicast problem $\problem^{\dagger}$.
\def\scalefactor{0.7}
\begin{figure}[htbp]
  \begin{center}
   \subfigure[Network $\graph$]
    {\includegraphics*[scale=\scalefactor]{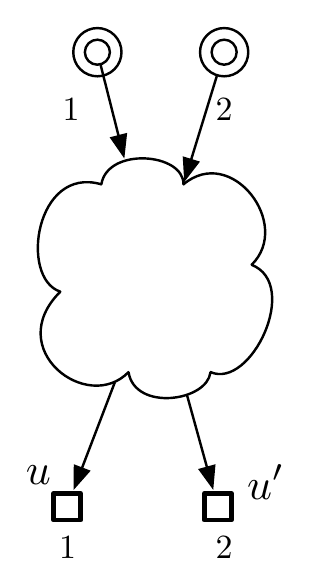}}
    \hspace{2mm}
    \subfigure[Network $\graph^{\dagger}$]
    {\includegraphics*[scale=\scalefactor]{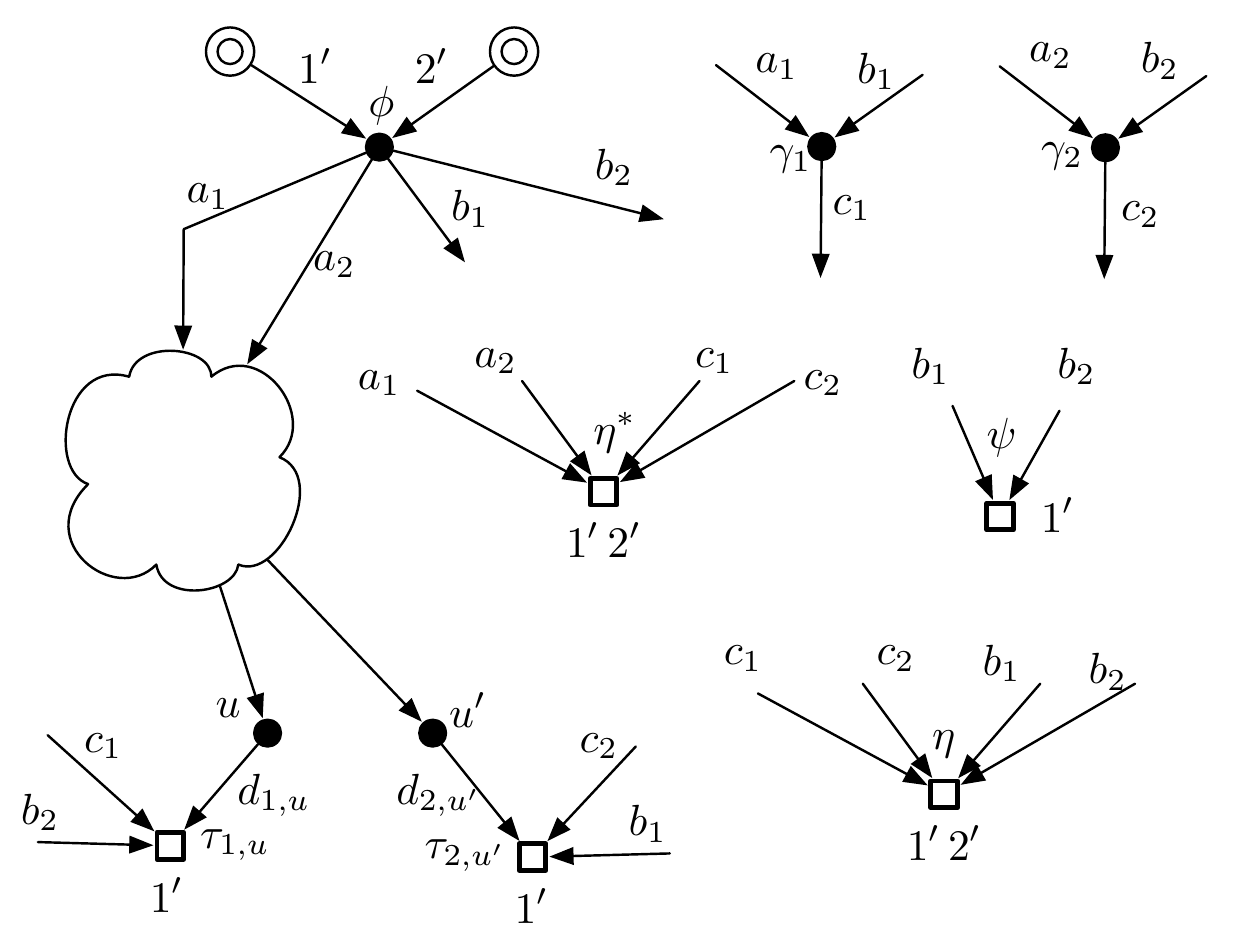}}
 \end{center}
  \caption{Transformation from $\graph$ to $\graph^\dagger$.}
  \label{fig:challenge1}
\end{figure}
Any rate-capacity tuple $(\lambda,\omega) \in {\chi} (\problem)$ for
$\problem$ induces another rate-capacity tuple
\[
 (\lambda^{\dagger},\omega^{\dagger}) \defined T^{\dagger}(\lambda,\omega) 
 \]
 in $ {\chi} (\problem^{\dagger})$ for $\problem^{\dagger}$
such that  
\begin{align}
\omega^{\dagger}(e) &= \omega(e), \label{eq:8A.166} \\
\omega^{\dagger}(a_{s}) = \omega^{\dagger}(b_{s}) =\omega^{\dagger}(c_{s}) & =\lambda(s), \label{eq:8A.167}\\
\omega^{\dagger}(d_{s,u}) &= \lambda(s), \label{eq:8A.168}\\
\lambda^{\dagger}(1^{'})= \lambda^{\dagger}(2^{'}) &= \sum_{s\in\sessions} \lambda(s).\label{eq:8A.169}
\end{align}
where $e\in\edges, s\in\sessions$ and $u\in\destinationLocation(s) $.


\begin{thm}\label{thm:challenge1}
Let $(\lambda,\omega)$ be a rate-capacity tuple in $\chi(\problem)$. Then the following two claims are valid. 
\begin{enumerate}
\item  If $ T^{\dagger}(\lambda,\omega) \in \major(  \proj_{\problem^{\dagger}} [ \bar \Gamma^{*}(\problem^{\dagger})  \cap \set{C}_{\indep}(\problem^{\dagger}) \cap \set{C}_{\net}(\problem^{\dagger})  \cap \set{C}_{\de}(\problem^{\dagger})  ]),
$ then
\[
(\lambda,\omega) \in  \major(  \proj_{\problem} [ \bar \Gamma^{*}(\problem)  \cap \set{C}_{\indep}(\problem) \cap \set{C}_{\net}(\problem)  \cap \set{C}_{\de}(\problem)  ]).
\]

\item  If a rate-capacity tuple $(\lambda,\omega)$ for $\problem$ is 0-achievable, then $ T^{\dagger}(\lambda,\omega)$  is 0-achievable with respect to $\problem^{\dagger}$.

\end{enumerate}
\end{thm}
\begin{proof}
  See Appendix~\ref{app:challenge1}
\end{proof}


Using a similar arguments as in Theorem  \ref{thm:challenge1}, we can also prove the following theorem, whose proofs we omit for brevity.   
\begin{thm}\label{thm:challenge1linear}
Let $(\lambda,\omega)$ be a rate-capacity tuple in $\chi(\problem)$. Then the following two claims are valid.
\begin{enumerate}

\item   If $ T^{\dagger}(\lambda,\omega) \in \major(  \proj_{\problem^{\dagger}} [ \bar \Upsilon_{q}^{*}(\problem^{\dagger})  \cap \set{C}_{\indep}(\problem^{\dagger}) \cap \set{C}_{\net}(\problem^{\dagger})  \cap \set{C}_{\de}(\problem^{\dagger})  ]),
$ then
\[
(\lambda,\omega) \in  \major(  \proj_{\problem} [ \bar\Upsilon_{q}^{*}(\problem)  \cap \set{C}_{\indep}(\problem) \cap \set{C}_{\net}(\problem)  \cap \set{C}_{\de}(\problem)  ]).
\]

\item  If   $(\lambda,\omega)$ is 0-achievable with respect to  $\problem$  subject to the $q$-linearity constraint, then $ T^{\dagger}(\lambda,\omega)$  is also 0-achievable with respect to $\problem^{\dagger}$, subject to the $q$-linearity constraint.
\end{enumerate}
\end{thm}

\begin{cor}[Colocated sources] \label{cor:challenge1}
Suppose all the sources are colocated in  $\problem$. Then 
\begin{enumerate}
\item
$(\lambda,\omega)$ is 0-achievable with respect to $\problem$ if and only if $ T^{\dagger}(\lambda,\omega)$ is also 0-achievable with respect to $\problem^{\dagger}$. 

\item
$(\lambda,\omega)$ is 0-achievable with respect to $\problem$ subject to the $q$-linearity constraint if and only if $ T^{\dagger}(\lambda,\omega)$ is also 0-achievable with respect to $\problem^{\dagger}$ subject to the same linearity constraint. 
\end{enumerate}
\end{cor}
\begin{proof}
A direct consequence of Theorems  \ref{thm:arbitrarymain}, \ref{thm:mainlinear}, \ref{thm:challenge1} and \ref{thm:challenge1linear}. 
\end{proof}

In \cite{Chan2008Dualities}, a specific network coding problem $\problem$ 
was proposed, such that all sources are colocated and that determining the set of $\epsilon$-achievable rate-capacity tuple is at
least as hard as determining the set of all information
inequalities. Furthermore, it was also proved that linear codes are  not  optimal\footnote{Linear codes are not optimal in the sense that there exists a 0-achievable rate-capacity tuple $(\lambda,\omega)$ for $\problem$ which is not achievable when subject to the additional linearity constraint.}. Therefore, by Corollary \ref{cor:challenge1}, 
we can directly prove the following proposition.

\begin{prop}\label{prop:challenge1}
There exists a two-layer incremental multicast network coding problem $\problem^{\dagger}$ such that 
\begin{enumerate}
\item
Characterising the set of achievable rate-capacity tuples for a two-layer incremental network is in general no simpler than determining the set of
all information inequalities. 

\item
Linear codes are not optimal. 
%
\end{enumerate}
\end{prop}  

\subsection{Secure Multicast}\label{sec:secure-multicast}
In this subsection, we consider another scenario, very simple secure
network coding problem with only one source. We will again show that
solving the resulting secure network coding problem can be as hard as
solving a general multi-source unconstrained network coding problem.
Our approach is essentially the same as that used in the previous
subsection for the incremental multicast. We will construct a simple
single-source secure network coding problem
$\problem^{\ddagger}=(\graph^{\ddagger} ,
\multicastRequirement^{\ddagger}, \wpattern)$ from an ordinary network
coding problem $\problem=(\graph,\multicastRequirement)$.  We will
then show that solving the so-constructed secure network coding
problem is as hard as solving the original network coding problem.

Construct the network $\graph^{\ddagger}$ in $\problem^{\ddagger}$
by adding nodes and hyperedges
\begin{align*}
  \nodes^{\ddagger} &\defined \nodes \cup \{ \psi_{s}, \gamma_{s},
  \theta_{s,u}, \tau_{s,u}, \: s\in\sessions ,
  u\in\destinationLocation(s) \} \cup \{ \phi , \eta\} \\
    \edges^{\ddagger} &\defined \edges \cup \{a_{s}, b_{s}, c_{s},
    e_{s}, s\in\sessions \}  \cup \{ d_{s,u} , w_{s,u},
    s\in\sessions , u\in \destinationLocation(s) \}
\end{align*}
with link connections:
\begin{align*}
   \tail{a_{s}}  &=  \phi  \\ 
   \head{a_{s}}  &=  \sourceLocation(s) \cup \{\eta \} \cup \{\gamma_{s}\}\\
   \tail{b_{s}}  &=  \phi  \\
    \head{b_{s}}  &=  \{\gamma_{s}, \eta\} \cup \bigcup_{j\neq s}
    \left\{\theta_{j,i} , i\in \destinationLocation(j) \right\}\\
   \tail{c_{s}}  &=  \gamma_{s}  \\ 
\head{c_{s}}  &=  \{\psi_{s} \} \\
    \tail{d_{s,u}}  &=  u  \\ 
\head{d_{s,u}}  &=  \tau_{s,u} \\
   \tail{e_{s}}  &=  \phi  \\
    \head{e_{s}}  &=  \{\psi_{s} \} \cup \{\tau_{s,i} , i\in \destinationLocation(s) \}.
\end{align*}
  for all $s\in\sessions, u\in\destinationLocation(s)$.

The connection requirement  $\multicastRequirement^{\ddagger} = (\sessions^{\ddagger}, \sourceLocation^{\ddagger}, \destinationLocation^{\ddagger})$ is defined as follows:  
\begin{align*}
\sourceLocation(1^{'}) & \defined \phi \\
\destinationLocation(1^{'}) &\defined \{ \tau_{s,u}, s\in\sessions, u\in\destinationLocation(s) \} \cup \{\psi_{s},s\in\sessions\} \cup \{ \eta\}	\\
\wpattern & \defined \{ (\set{A}_{1}, \set{B}_{1}) , (\set{A}_{2}, \set{B}_{2})  \}
\end{align*}
where 
\begin{align*}
\set{A}_{1} &= \set{A}_{2} = 1^{\prime} \\
\set{B}_{1} & = \{ a_{s}, s\in\sessions \} \\
\set{B}_{2} & = \{ b_{s}, s\in\sessions \}. 
\end{align*}

Figure \ref{fig:thenetwork} exemplifies how to convert an
unconstrained network coding problem $\problem$ into a single-source
secure network coding problem.
\begin{figure}[htbp]
  \begin{center}
   \subfigure[Network $\graph$]
    {\includegraphics*[scale=\scalefactor]{incrementaloriginal}}
    \hspace{2mm}
    \subfigure[Network $\graph^{\ddagger}$]
    {\includegraphics*[scale=\scalefactor]{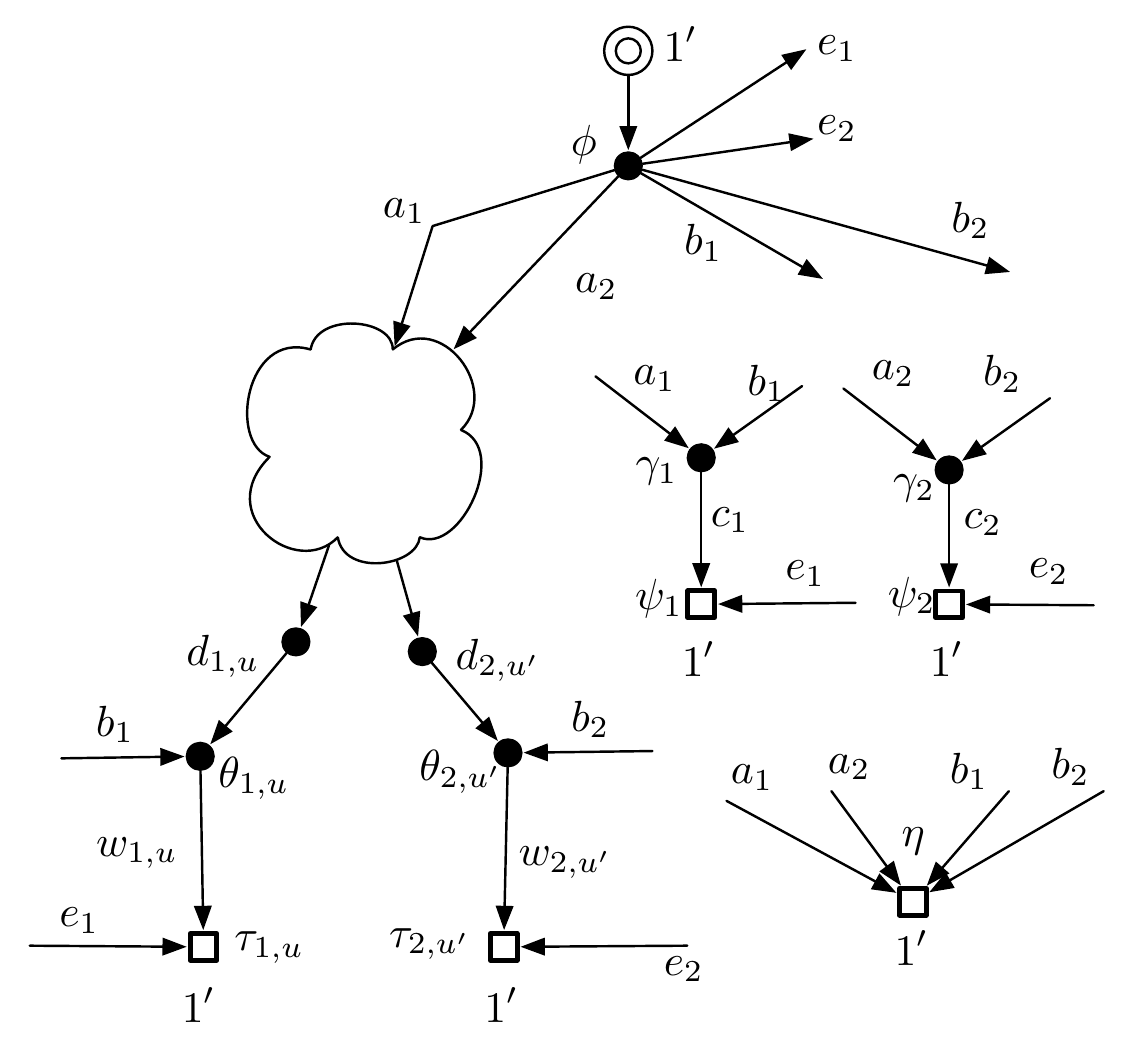}}

 \end{center}
  \caption{Transformation from $\graph$ to  $\graph^\ddagger$.}
  \label{fig:thenetwork}
\end{figure}

As before, for any rate-capacity tuple $(\lambda,\omega) \in {\chi}
(\problem)$, we define a tuple $(\lambda^{\ddagger},\omega^{\ddagger})
\defined T^{\ddagger}(\lambda,\omega) \in {\chi}
(\problem^{\ddagger})$ as follows:
\begin{align}
\omega^{\ddagger}(e) &= \omega(e), \label{195}\\
\omega^{\ddagger}(a_{s}) = \omega^{\ddagger}(b_{s}) =\omega^{\ddagger}(c_{s}) & =\lambda(s),   \\
\omega^{\ddagger}(d_{s,u}) &= \lambda(s),     \\
\omega^{\ddagger}(e_{s}) & = \sum_{i\in\sessions \setminus \{s\}} \lambda(i) \\
\lambda^{\ddagger}(1^{'})  &= \sum_{s\in\sessions} \lambda(s)
\end{align}
for all $e\in\edges, s\in\sessions$ and $u\in \destinationLocation(s)$.

\begin{thm}\label{thm:challenge2}
Let $(\lambda,\omega)$ be a rate-capacity tuple in $\chi(\problem)$. Then the following two claims are valid. 
\begin{enumerate}
\item  If 
$T^{\ddagger}(\lambda,\omega) \in \major(  \proj  (h^{\ddagger}) ) $ for some 
\[
h^{\ddagger } \in \bar \Gamma^{*}(\problem^{\ddagger})  \cap \set{C}_{\indep}(\problem^{\ddagger}) \cap \set{C}_{\net}(\problem^{\ddagger})  \cap \set{C}_{\de}(\problem^{\ddagger}) \cap \set{C}_{\secure}(\problem^{\ddagger}) .
\]
then
\[
(\lambda,\omega) \in   \major(  \proj ( \bar \Gamma^{*}(\problem)  \cap \set{C}_{\indep}(\problem) \cap \set{C}_{\net}(\problem)  \cap \set{C}_{\de}(\problem)  ))
.
\]

\item   If a rate-capacity tuple $(\lambda,\omega)$ for $\problem$ is 0-achievable, then $ T^{\ddagger}(\lambda,\omega)$  is 0-achievable with respect to $\problem^{\ddagger}$ subject to the strong secrecy constraint. 
\end{enumerate}
\end{thm}
\begin{proof}
  See Appendix~\ref{app:challange2}
\end{proof}

The following theorem is the counterpart of Theorem \ref{thm:challenge1linear}. Again, its proof will be omitted.
\begin{thm}\label{thm:challenge2linear}
Let $(\lambda,\omega)$ be a rate-capacity tuple in $\chi(\problem)$. Then the following two claims are valid.
\begin{enumerate}
\item   If $ T^{\ddagger}(\lambda,\omega) \in \major(  \proj_{\problem^{\ddagger}} [\bar \Gamma^{*}(\problem^{\ddagger})  \cap \set{C}_{\indep}(\problem^{\ddagger}) \cap \set{C}_{\net}(\problem^{\ddagger})  \cap \set{C}_{\de}(\problem^{\ddagger}) \cap \set{C}_{\secure}(\problem^{\ddagger}) ]),
$ then
\[
(\lambda,\omega) \in  \major(  \proj_{\problem} [ \bar\Upsilon_{q}^{*}(\problem)  \cap \set{C}_{\indep}(\problem) \cap \set{C}_{\net}(\problem)  \cap \set{C}_{\de}(\problem)  ]).
\]

\item  If   $(\lambda,\omega)$ is 0-achievable with respect to $\problem$ subject to the $q$-linearity constraint, then $ T^{\ddagger}(\lambda,\omega)$  is also 0-achievable with respect to $\problem^{\ddagger}$,   subject to the strong secrecy  and $q$-linearity constraint.
\end{enumerate}
\end{thm}

\begin{cor}[Counterpart of Corollary \ref{cor:challenge1}]
Suppose all the sources are colocated in  $\problem$. Then 
\begin{enumerate}
\item
$(\lambda,\omega)$ is 0-achievable with respect to $\problem$ if and only if $ T^{\ddagger}(\lambda,\omega)$ is also 0-achievable  with respect to $\problem^{\ddagger}$ subject to the strong secrecy constraint.

\item
$(\lambda,\omega)$ is 0-achievable with respect to $\problem$ subject to the $q$-linearity constraint if and only if $ T^{\ddagger}(\lambda,\omega)$ is also 0-achievable with respect to $\problem^{\ddagger}$ subject to the strong secrecy and  $q$-linearity constraint. 
\end{enumerate}
\end{cor}

%
%

\begin{prop}[Counterpart of Proposition \ref{prop:challenge1}]\label{prop:challenge2}
There exists a single source secure multicast network coding problem $\problem^{\ddagger}$ such that 
\begin{enumerate}
\item
Characterising the set of achievable rate-capacity tuples for $\problem^{\ddagger}$  is in general no simpler than determining the set of all information inequalities. 

\item
Linear codes may not be optimal. 
%
\end{enumerate}
\end{prop}


{\bf Remark: }
In \cite{cai2002snc}, it was proved that in the single-source case, if
all links have equal capacity and the eavesdroppers' capability is
limited by the total number of links it can wiretap, then linear
network codes are optimal. Therefore, Proposition \ref{prop:challenge2} is indeed a
surprising result proving that linear network codes are not optimal in
general.
 
\section{Conclusion}
Characterisation of the set of zero-error or vanishing-error
achievable rate-capacity tuples for network coding is a fundamental
problem in multiterminal information theory. In
\cite{Chan2008Dualities}, it was proved that this characterisation
problem is extremely difficult in general and is as hard as
determining the set of all information inequalities. This goes some
way toward explaining why the problem has so far been solved only for
a few special cases.

The authors in \cite{Yan2007The-Capacity} and \cite{Yeung02first} used
entropy functions to implicitly characterise the set of achievable
rate-capacity tuples for general networks. Although this
characterisation is implicit, it offers insights about the structure
of the set of achievable tuples. For example, knowing that the set of
almost entropic functions $\bar\Gamma^{*}$ is not polyhedral,
\cite{Chan2008Dualities} proved that the set of achievable tuples also
is not polyhedral in general.

This paper extended \cite{Yan2007The-Capacity} and \cite{Yeung02first}
in several aspects. First, we proved that when sources are colocated,
the outer bound given in \cite[Section 15.5]{Yeung02first} is indeed
tight. In particular, we showed that the set of rate-capacity tuples
achievable with vanishing error, and the set achievable with zero
error are indeed the same. We also gave evidence to support our
conjecture that the outer bound in \cite[Section 15.5]{Yeung02first}
remains tight even when sources are not colocated.  

Secondly, we considered network coding problems subject to several
practically-motivated constraints, such as linear coding, the
restriction of some or all nodes to perform only routing, and security
requirements. For these cases we characterised the set of zero-error
and vanishing-error achievable rate-capacity tuples.  Finally in
Section \ref{sec:challenges}, we proved that even for very simple
network coding problems, such as the incremental multicast problem and
the single source secure network coding problem with arbitrary
wiretapping patterns, characterisation of achievable tuples is as hard
as the characterisation problem for general unconstrained network
coding. We also proved that linear codes are suboptimal for both the
general incremental multicast problem and for the single source secure
network coding problem.

\begin{appendices}

\section{Proof of Proposition \ref{prop:regular}}\label{appendix:A}
Consider the following combinatorial problem. Suppose that there are
$k$ boxes, $t$ of which are nonempty. If we randomly select $m$
distinct boxes, then the probability that all selected boxes are empty
is upper bounded by
\begin{align}\label{eq:248}
  \Pr(\text{all $m$ boxes are empty}) &\le  \left( 1- \frac{t}{k} \right)^{m}
\end{align}

Let $\kappa(c) = (1-c)^{1/c}$. Since $\lim_{c\to 0^{+}} \kappa(c) = \exp(-1)$,  there exists $0<\delta<1$ such that $\kappa(c) < \delta$ for all $0<c\le 1$. Hence, \eqref{eq:248} can be relaxed as 
\begin{align}\label{eq:249} 
\Pr(\text{all $m$ boxes are empty}) & \le  \delta^{tm/k} .
\end{align}

Let $(U,V)$ be a pair of quasi-uniform random variables.  As $V$ is
uniform over its support, $|\support(V)|=2^{H(V)}$.  Let
\begin{align*}
m =   H\left(UV\right)^{2} \frac{2^{H\left(V\right)}}{2^{H\left(V\mid U\right)}} .
\end{align*}
Partition the set $\support(V)$ randomly into
\[
2^{H\left(V\right)} / m  =  \frac{2^{H\left(V\mid U\right)}}{H\left(UV\right)^{2}}
\]
subsets, each of which is of size $m$.  These disjoint subsets will be
denoted by $\Xi(b)$ where $$b \in \set{A}_{V} \define \{ 1, \ldots,
2^{H\left(V\mid U\right)} / H\left(UV\right)^{2} \}.$$

For any $ u \in\support(U)$ and $b \in \set{A}_{V}$, let  
$\mathbb E(u,b)$ be the event that  
\[
\{ (u,i):  i\in \Xi(b) \} \cap \support(U,V) \neq \emptyset.
\]
In other words, the event is equivalent to the existence of 
an element  $i \in \Xi(b)$ such that $(u,i) \in \support(U,V)$.

In the following, we will prove that the probability of 
$\mathbb E(u,b)$ is ``arbitrarily close to one asymptotically'' 
for all $u \in \support(U)$ and $b\in \set{A}_{V}$.

For any $u\in\support(U)$, it is easy to see that 
\begin{align}
|\{v : (u,v)  \in \support(U,V) \} | = 2^{H\left(V\mid U\right)}.
\end{align}
By setting $k=2^{H\left(V\right)}$ and $t=2^{H\left(V\mid U\right)}$,  \eqref{eq:249} implies that   
\begin{align}
\Pr(\mathbb E(u,b)) & \ge 1- \delta^{H\left(UV\right)^{2}}
\end{align}
and hence via the union bound, the probability that the event 
$\mathbb E(u,b)$ occurs for all  $u\in \support(U)$ and $b \in \set{A}_{V} $ is at least 
\begin{align*}
1 - 2^{ H\left(U\right)+H\left(V\right) }  \delta^{H\left(UV\right)^{2}}. 
\end{align*}

This probability approaches to 0 as $H\left(UV\right)$ goes to infinity. 
Consequently, if the entropy $H\left(U\right)$ (and hence also $H\left(UV\right)$) is  sufficiently large,  
there exists a way to partition $\support(V)$  
such that for any $u \in \support(U)$ and $b \in \set{A}_{V} $, 
there exists at least one $v \in \Xi(b)$ such that $(u,v)\in \support(U,V)$. 

Assume without loss of generality that $\sessions=\{1,\ldots, |\sessions|\}$.
Repeating the same argument, we can recursively prove that  
for any set of quasi-uniform random variables 
$\{U_{i},i\in\sessions\}$ and  $H\left(U_{1}\right)$ sufficiently large, 
there exists at least a way to partition $\support(U_{s})$ into  
\[
{2^{H\left(U_{i}\mid \seq U1{i-1}\right)} }/{ H\left(\seq U1i\right)^{2}} 
\]
subsets $\Xi_{s}(b_{s}) $ where 
\[
b_{s} \in  \set{A}_{s} \define \{1, \ldots,  2^{H\left(U_{s}\mid  \seq U1{s-1}\right)} / H\left(\seq U1s\right)^{2}\}\]
such that for any  $(\seq u1{s-1}) \in \support(\seq U1{s-1})$ and $b_{s} \in \set{A}_{s}$, there exists  at least one $u_{s} \in \Xi_{s}(b_{s})$ such that 
$(\seq u1s)\in \support(\seq U1s)$. 
Hence, the proposition is proved.

\section{Proof of Theorem~\ref{thm:challenge1}}\label{app:challenge1}

We first prove the first claim. Suppose 
\[
(\lambda^{\dagger} , \omega^{\dagger}) \defined T^{\dagger}(\lambda,\omega) \in \major(\proj_{\problem^{\dagger}}[h^{\dagger}])
\]
for some   
\[
h^{\dagger} \in \bar \Gamma^{*}(\problem^{\dagger})  \cap   \set{C}_{\indep}(\problem^{\dagger}) \cap \set{C}_{\net}(\problem^{\dagger})  \cap \set{C}_{\de}(\problem^{\dagger}) .
\] 
Then by definition,  
\begin{align}
\lambda^{\dagger}(i) & \le  h^{\dagger}(i), \quad i=1,2  \label{eq:8A.170}\\
\omega^{\dagger}(a_{s}) & \ge  h^{\dagger}(a_{s}), \quad \forall s\in\sessions \label{eq:8A.171}\\
\omega^{\dagger}(b_{s}) & \ge  h^{\dagger}(b_{s}), \quad \forall s\in\sessions\label{eq:8A.172} \\
\omega^{\dagger}(c_{s}) & \ge  h^{\dagger}(c_{s}), \quad \forall s\in\sessions\label{eq:8A.173} \\
\omega^{\dagger}(e) & \ge  h^{\dagger}(e), \quad \forall e\in\edges.\label{eq:8A.173b}
\end{align}

Consequently, by \eqref{eq:8A.166}--\eqref{eq:8A.169} and   \eqref{eq:8A.171}--\eqref{eq:8A.173}, 
\begin{align}
2 \sum_{s\in\sessions} \lambda({s})   = \sum_{s\in\sessions}( \omega^{\dagger}(a_{s})  + \omega^{\dagger}(b_{s}) ) 
  \ge \sum_{s\in\sessions} (h^{\dagger}(a_{s})  + h^{\dagger}(b_{s}))  \nge{(i)} h^{\dagger}(a_{s}, b_{s}, s\in\sessions) \label{eq:7.108}  
\end{align}
where $(i)$ follows from the fact that $h^{\dagger} \in \bar
\Gamma^{*}(\problem^{\dagger}) $ (and hence is polymatroidal).
Similarly, we can also prove that
\begin{align}
2 \sum_{s\in\sessions} \lambda({s})   = \sum_{s\in\sessions}( \omega^{\dagger}(a_{s})  + \omega^{\dagger}(c_{s}) ) 
  \ge \sum_{s\in\sessions} (h^{\dagger}(a_{s})  + h^{\dagger}(c_{s}))  \ge h^{\dagger}(a_{s}, c_{s}, s\in\sessions).  \label{eq:7.109}
\end{align}

Let $g$ be the ``projection'' of $h^{\dagger}$ on $\set{H}[\sessions\cup\edges]$ such that  for any $\alpha\subseteq \edges $ and $\beta \subseteq \sessions$, 
\begin{align}\label{eq:8.174}
g(\alpha,\beta) \defined h^{\dagger}(\alpha , a_{i}, i\in\beta).
\end{align}
In the following, we will prove that 
\[
(\lambda,\omega) \in \major(\proj_{\problem}[g ])
\]
and 
\[
g\in \bar \Gamma^{*}(\problem)  \cap \set{C}_{\indep}(\problem) \cap \set{C}_{\net}(\problem)  \cap \set{C}_{\de}(\problem). 
\]

First, $h^{\dagger} \in \bar \Gamma^{*}(\problem^{\dagger})$. Hence,
its projection $g$ is also in $ \bar \Gamma^{*}(\problem)$.  Second,
the network $\graph^{\dagger}$ contains $\graph$ as a subnetwork and
$\sourceLocation(s) \subseteq \head{a_{s}}$. In other words, if a node
$u$ has access to the source $s$ in the network coding problem
$\problem$, then $u$ also has access to what is being transmitted
along the link $a_{s}$ in $\problem^{\dagger}$. The link $a_{s}$ in
$\graph^{\dagger}$ is thus like an imaginary source link in $\graph$.
It can then be verified directly that $g \in
\set{C}_{\net}(\problem)$.

Now, we will prove that $g\in \set{C}_{\indep}(\problem) \cap \set{C}_{\de}(\problem)$. As 
\[
h^{\dagger} \in  \bar \Gamma^{*}(\problem^{\dagger}) \cap \set{C}_{\de}(\problem^{\dagger}) \cap \set{C}_{\net}(\problem^{\dagger}),
\]
 and that the set of links
$
\{ a_{s}, b_{s}, s\in\sessions \}
$
 separates the source node $\phi$ from the sink node $\eta$ in $\graph^{\dagger}$, 
\begin{align}
h^{\dagger}(1^{\prime},2^{\prime} | a_{s},b_{s}, s\in\sessions ) = h^{\dagger}(a_{s},b_{s}, s\in\sessions | 1^{\prime},2^{\prime} ) = 0. \label{175}
\end{align}
Consequently,
\begin{align}
h^{\dagger}(a_{s}, b_{s}, s\in\sessions) & \ge  h^{\dagger}(1^{\prime},2^{\prime})  \\
& \nequal{(i)}  h^{\dagger}(1^{\prime} )  +h^{\dagger}( 2^{\prime})  \\
& \nge{(ii)} \lambda^{\dagger}(1)  +  \lambda^{\dagger}(2) \\ 
& = 2 \sum_{s\in\sessions } \lambda(s)
\end{align}
where $(i)$ follows from the fact that $h^{\dagger} \in \bar
\Gamma^{*}(\problem^{\dagger})\cap
\set{C}_{\indep}(\problem^{\dagger})$ and $(ii)$ follows from
\eqref{eq:8A.170}.

Similarly, the set of links 
$
\{ a_{s}, c_{s}, s\in\sessions\}
$
separates  the source node $\phi$ from the sink node $\eta^{*}$ in $\graph^{\dagger}$. Hence, 
\begin{align}
h^{\dagger}(a_{s}, c_{s}, s\in\sessions) \ge  2 \sum_{s\in\sessions } \lambda(s).
\end{align}

Therefore, all the inequalities in \eqref{eq:7.108} and \eqref{eq:7.109} are  in fact equalities. In particular,  
\begin{align} 
h^{\dagger}(a_{s}) = h^{\dagger}(b_{s})  = h^{\dagger}(c_{s}) = \lambda(s), \quad \forall  s\in \sessions \label{eq:8.185}
\end{align} 
and 
\begin{align}
 h^{\dagger}(a_{s},b_{s} , s\in\sessions)   &   =  \sum_{s\in\sessions} (h^{\dagger}(a_{s})  + h^{\dagger}(b_{s})) ,  \label{eq:8.186} \\
 h^{\dagger}(a_{s} , c_{s}, s\in\sessions)     & =  \sum_{s\in\sessions} (h^{\dagger}(a_{s})  + h^{\dagger}(c_{s})). \label{eq:8.187} 
\end{align}

By \eqref{eq:8.186},   $h^{\dagger}(a_{s},s\in\sessions)  = \sum_{s\in\sessions}h^{\dagger}(a_{s})$.
Hence, 
\[
g( \sessions) = \sum_{s\in\sessions}g({s})
\]
and  $g\in\set{C}_{\indep}(\problem)$. Furthermore, as 
\begin{align}
g(s) &= h^{\dagger}(a_{s}) = \lambda(s), \quad \forall s\in\sessions \\
g(e) &= h^{\dagger}(e) \le \omega(e), \quad \forall e\in\edges,
\end{align}
we  prove that  
\[
(\lambda,\omega) \in \major(\proj_{\problem}[g ]).
\]

\def\hd1{{h^{\dagger}}}
Now, it remains to show that  $g \in\set{C}_{\de}(\problem)$.  
First, consider any $s\in\sessions$ and $u\in\destinationLocation(s)$. By \eqref{eq:8.186}--\eqref{eq:8.187}, 
\[
\hd1\left( b_{\sessions \setminus s} \wedge a_{\sessions} , b_{s}\right) =0.
\]
As $\hd1\in \bar \Gamma^{*}(\problem^{\dagger})  \cap \set{C}_{\net}(\problem^{\dagger})$,  $\hd1(c_{s} \mid a_{s},b_{s}) = \hd1(\incoming(u) \mid a_{\sessions}) = 0$. Hence, 
\begin{align}\label{f.190}
{\hd1}( b_{\sessions \setminus s} \wedge a_{\sessions} , b_{s} , c_{s}, \incoming(u) ) =0.
\end{align}
Together with  the decoding constraint (for the receiver $\tau_{s,u}$) 
\[
\hd1(a_{s} \mid b_{\sessions \setminus s} , c_{s}, \incoming(u) ) =0, 
\]
we can prove that  
\[
\hd1(a_{s} \mid  c_{s},\incoming(u) )=0.
\]
Finally, using \eqref{eq:8.187} and  that $\hd1(\incoming(u) \mid a_{\sessions}) = 0$, 
we have
$
\hd1(a_{s} \mid  \incoming(u) )=0.
$
Thus, 
\[
g(a_{s} \mid \incoming(u)) = \hd1(a_{s} \mid \incoming(u) )=0
\]
and $g\in \set{C}_{\de}(\problem)$. The  first claim is proved.

To prove the second claim, suppose $(\lambda,\omega)$ is 0-achievable with respect to  
$\problem$. By definition, there exists a sequence of zero-error network codes 
\[
\{ Y^{n}_{f}, f\in\sessions\cup\edges \}
\]
for $\problem$, and a sequence of positive constants $c_{n}$ such that 
\begin{align}
\lim_{n\to\infty} c_{n} H(Y_{e}^{n}) \le \lim_{n\to\infty} c_{n} H|\support(Y_{e}^{n})| \le \omega(e) \\
\lim_{n\to\infty} c_{n} H(Y_{s}^{n}) = \lim_{n\to\infty} c_{n} H|\support(Y_{s}^{n})| \ge \lambda(s). 
\end{align}
%
%
%

Assume without loss of generality that
\[
\support(Y_{s}^{n})=\{ 0, \ldots, |\support(Y_{s}^{n})|-1 \}.
\]
For each $n$, define a new set of random variables
\[
\{U_{f}^{n}, f\in\sessions^{\dagger} \cup \edges^{\dagger}  \}
\]
such that  for any $s\in\sessions$ and $u\in\destinationLocation(s)$, 
\begin{enumerate}
\item $U^{n}_{a_{s}} \defined Y^{n}_{s}$; 

\item $U^{n}_{e} \defined Y^{n}_{e}$;

\item $\{U^{n}_{b_{s}}, s \in\sessions \}$ is a set of mutually independent random variables such that each of which is uniformly distributed over $\{ 0, \ldots, |\support(Y_{s}^{n})|-1 \}$ and 
%
\[
H\left(U^{n}_{a_{\sessions}}, U^{n}_{\edges},  U^{n}_{b_{\sessions}}\right) = \sum_{f\in\sessions }U^{n}_{b_{f}} + H\left(U^{n}_{a_{\sessions}}, U^{n}_{\edges}\right);
\]
  
\item $U^{n}_{c_{s}} \defined U^{n}_{a_{s}} + U^{n}_{b_{s}} \mod  |Y^{n}_{{s}}|$; 
\item $U^{n}_{d_{s,u}} \defined U^{n}_{a_{s}}$;

\item $U^{n}_{1^{'}} \defined (U^{n}_{b_{f}}, f\in\sessions)$;

\item $U^{n}_{2^{'}}\defined (U^{n}_{a_{f}}, f\in\sessions)$.

\end{enumerate}
It can then be proved directly that 
$
\{U_{f}^{n}, f\in\sessions^{\dagger} \cup \edges^{\dagger}  \}
$
is a sequence of zero-error network codes for the network coding problem $\problem^{\dagger}$. Consequently,  $T(\lambda,\omega)$ is 0-achievable with respect to $\problem^{\dagger}$. The theorem is proved.

%

\section{Proof of Theorem~\ref{thm:challenge2}}
\label{app:challange2}

\def\hd{h^{\ddagger}}

We first prove the first claim. Let 
\begin{align}\label{c.212}
(\lambda^{\ddagger} , \omega^{\ddagger}) \defined T^{\ddagger}(\lambda, \omega) \in \major(  \proj_{\problem^{\ddagger}}  (h^{\ddagger}) ) 
\end{align}
for some 
\begin{align}\label{c.213}
h^{\ddagger} \in  \bar\Gamma^{*}(\problem^{\ddagger}) \cap   \set{C}_{\indep}(\problem^{\ddagger}) \cap \set{C}_{\net}(\problem^{\ddagger}) \cap \set{C}_{\de}(\problem^{\ddagger}) \cap \set{C}_{\secure}(\problem^{\ddagger}) .
\end{align}

By \eqref{195}-\eqref{c.212},  for all $e\in\edges, s\in\sessions$ and $u\in\destinationLocation(s)$, 
\begin{align}
\sum_{i\in\sessions} \lambda({i}) & = \lambda^{\ddagger}(1^{'}) \le h^{\ddagger}(1^{'}) \label{c.218}\\
\lambda(s) & = \omega^{\ddagger} (a_{s}) \ge h^{\ddagger} (a_{s})  \label{c.219}\\
\lambda(s) & = \omega^{\ddagger} (b_{s}) \ge h^{\ddagger} (b_{s}) \label{c.220} \\
\lambda(s) & = \omega^{\ddagger} (c_{s}) \ge h^{\ddagger} (c_{s})  \label{c.221}\\
\lambda(s) & = \omega^{\ddagger} (d_{s,u}) \ge h^{\ddagger} (d_{s,u})  \label{c.222}\\
\lambda(s) & = \omega^{\ddagger} (w_{s,u}) \ge h^{\ddagger} (w_{s,u})  \label{c.223}\\
\sum_{i\in\sessions, i\neq s} \lambda(i) & = \omega^{\ddagger} (e_{s}) \ge h^{\ddagger} (e_{s})  \label{c.224}\\
\omega(e) & = \omega^{\ddagger}(e) \ge \hd(e).\label{c.224a}
\end{align}
Using \eqref{c.218}--\eqref{c.224a}, we have 
\begin{align}
h^{\ddagger}(1^{'}) \ge \sum_{s\in\sessions} \lambda(s) = \sum_{s\in\sessions} \omega^{\ddagger}(a_{s}) 
 \ge \sum_{s\in\sessions} \hd(a_{s}) \nge{(i)}  h^{\ddagger}(a_\sessions). \label{c.225}
\end{align}
where $(i)$ follows from that $\hd \in \bar\Gamma^{*}(\problem^{\ddagger})$ and hence  is a polymatroid.
Similarly, 
\begin{align}
h^{\ddagger}(1^{'}) \ge \sum_{s\in\sessions} \lambda(s) = \sum_{s\in\sessions} \omega^{\ddagger}(b_{s}) 
 \ge \sum_{s\in\sessions} \hd(b_{s}) \ge  h^{\ddagger}(b_\sessions). \label{c.226}
\end{align}

Recall that $\hd $ is a rank function in the space $\set{H}[\sessions^{\ddagger} \cup \edges^{\ddagger} \cup \nodes^{\ddagger}]$. 
 Let $g$ be its  ``projection'' on $\set{H}[\sessions\cup\edges]$ such that 
for any $\alpha\subseteq \edges $ and $\beta \subseteq \sessions$, 
\begin{align}
g(\alpha,\beta) \defined h^{\ddagger}(\alpha , a_{i}, i\in\beta \mid  \nodes^{\ddagger}\setminus \{\phi\}).
\end{align}
In the following, we will prove that  
\[
(\lambda,\omega) \in \major(\proj_{\problem}[g])
\]
and
\[
g\in \bar \Gamma^{*}(\problem)  \cap \set{C}_{\indep}(\problem) \cap \set{C}_{\net}(\problem)  \cap \set{C}_{\de}(\problem) .
\]

First, as $h^{\ddagger} \in \bar \Gamma^{*}(\problem^{\ddagger})$, it is obvious that 
\[
g\in \bar \Gamma^{*}(\problem).
\]
Second, notice that the network $\graph^{\ddagger}$ contains $\graph$ as a subnetwork. Therefore,  
by $h^{\ddagger} \in \bar \Gamma^{*}(\problem^{\ddagger})  \cap \set{C}_{\net}(\problem^{\ddagger})$,
\[
g \in  \set{C}_{\net}(\problem).
\]

Now,  as  $h^{\ddagger} \in   \set{C}_{\indep}(\problem^{\ddagger})  \cap \set{C}_{\secure}(\problem^{\ddagger})$, we have 
\begin{align}
 \hd( 1^{'}, \nodes^{\ddagger} )  &=  \hd( 1^{'})  + \sum_{u\in\nodes^{\ddagger}} \hd(u)  \label{c.215}\\
 {\hd}\left(1^{'} \wedge a_{\sessions}\right) & = 0 \label{c.216}\\
 {\hd}\left(1^{'} \wedge b_{\sessions}\right) & =0  . \label{c.217} 
\end{align}
By   \eqref{c.215}, we can deduce that 
\begin{align}
\hd\left(1^{'}, \phi \wedge \nodes^{\ddagger}\setminus \{\phi\}\right) =0. \label{c.228a}
\end{align}
Due to the topology constraint (for the links $\{a_{s},b_{s}, s\in\sessions \}$),
\begin{align}
\hd(a_{\sessions} , b_{\sessions}\mid  1^{'}, \phi) =0. \label{d.229a}
\end{align} 
Hence, 
\begin{align}
\hd\left(a_{\sessions},b_{\sessions} , 1^{'}\wedge \nodes^{\ddagger}\setminus \{\phi\}\right) =0 \label{d.232}
\end{align}

On the other hand,  the decoding constraint  $h^{\ddagger} \in \set{C}_{\de}(\problem) $ (for the sink node $\eta$), we have 
\begin{align}\label{e.228}
\hd(1^{'} \mid  a_{\sessions},b_{\sessions} )=0.
\end{align}
Together with   \eqref{c.216}--\eqref{c.217}, (and with the fact that $\hd$ is a polymatroid) \eqref{e.228} implies that  
\begin{align}
\hd(a_\sessions) & \ge \hd(1^{'}) \\
\hd(b_\sessions) & \ge \hd(1^{'}) .
\end{align}
By the upper bounds on $\hd(a_{\sessions})$ and $\hd(b_{\sessions})$ in \eqref{c.225}--\eqref{c.226}, we can in fact prove that  
\begin{align}
h^{\ddagger}\left(a_{\sessions}\wedge b_{\sessions} \right) & = 0 \label{e.237}\\
h^{\ddagger}(a_{\sessions},b_{\sessions}) &=  \sum_{s\in\sessions } (h^{\ddagger}(a_{s}) + h^{\ddagger}(b_{s}) ) . \label{eq:7.127} \\
h^{\ddagger}(a_{s},s\in\sessions) & = \sum_{s\in\sessions}h^{\dagger}(a_{s}) \label{eq:7.127a} \\
h^{\ddagger}(a_{s}) & = h^{\ddagger}( b_{s}) = \lambda(s), \quad \forall s\in\sessions.\label{eq:7.127b}
\end{align}

Now, for any $s\in\sessions$, 
\begin{align*}
g(s) &= \hd(a_{s} \mid  \nodes^{\ddagger}\setminus \{\phi\}) \\
&\nequal{(i)} \hd(a_{s} ) \\
&= \lambda(s)
\end{align*} 
where $(i)$ follows from \eqref{d.232}. 
Also, for any $e\in\edges$,
\begin{align*}
g(e) &= \hd(e \mid  \nodes^{\ddagger}\setminus \{\phi\}) \\
&\le  \hd(e) \\
&\le  \omega(e).
\end{align*} 
Therefore, 
$(\lambda,\omega)\in\major(\proj_{\problem}[g])$.

By  definition,
\begin{align}
g( \sessions) = \hd(a_{\sessions} \mid  \nodes^{\ddagger} \setminus \{\phi\}) 
\nequal{(i)} \hd( a_{\sessions} )  = \sum_{{s}\in\sessions} \hd(a_{s}) \ge  \sum_{s\in\sessions} \hd(a_{s}\mid  \nodes^{\ddagger} \setminus \{\phi\}) = \sum_{s\in\sessions}g({s}) \ge g( \sessions)
\end{align}
where $(i)$ is due to \eqref{d.232}. 
Thus,  $g\in\set{C}_{\indep}(\problem)$.

Our last step is to prove that  $g\in \set{C}_{\de}(\problem)$. 
As \[
\hd(c_{s}) + \hd(e_{s}) \le \sum_{s\in\sessions} \lambda_{s} =  \hd(1^{'}) , 
\]
the decoding constraint  $\hd( 1^{'} \mid  c_{s}, e_{s} )  = 0$  (for the receiver $\psi_{s}$) implies that 
\begin{align}
\hd(c_{s}\mid 1^{'}) = 0   \label{d.237}
\end{align}
and
\begin{align}
\hd(c_{s} ) =\lambda(s) \label{eq:7.128}
\end{align}
By \eqref{c.216}--\eqref{c.217} and  \eqref{eq:7.127}
\begin{align} \label{d.239}
{\hd} ( c_{s} \wedge a_{s}) = {\hd} ( c_{s} \wedge b_{s}) = {\hd} ( a_{s} \wedge b_{s}) = 0
\end{align}

On the other hand, by \eqref{c.215}, $\hd\left(\gamma_{s} \wedge 1^{'} ,\phi\right)=0$. By \eqref{d.229a} and \eqref{d.237},  we have  
$\hd\left(\gamma_{s} \wedge  a_{s},b_{s},c_{s}\right)=0$.  Together with one of the topology constraint (for the link $c_{s}$)
\begin{align}
\hd(c_{s} \mid  a_{s},b_{s}, \gamma_{s}) = 0,
\end{align}
we  have
\begin{align}
\hd(c_{s} \mid  a_{s},b_{s}) = 0. \label{d.241}
\end{align}

By \eqref{eq:7.128} and \eqref{eq:7.127b}, we can prove that 
\begin{align}
\hd(c_{s} \mid  a_{s}, b_{s})  = \hd(a_{s} \mid  c_{s}, b_{s})  = \hd(b_{s} \mid  a_{s}, c_{s})  =0 . \label{d.242}
\end{align}

Similarly, focusing on the receiver $\tau_{s,u}$, the decoding constraint 
$h\left(1^{'} \mid  w_{s,u} , e_{s}\right)  = 0$ and that $\hd(w_{s,u}) \le \lambda(s)$ imply that 
\begin{align}
\hd(w_{s,u}\mid 1^{'}) = 0   \label{d.243}
\end{align}
and
\begin{align}
\hd(w_{s,u} ) =\lambda(s).  \label{eq:7.129}
\end{align}

By \eqref{c.217}, ${\hd} ( w_{s,u} \wedge b_{s}) =0$. 
On the other  hand, by \eqref{d.232} and \eqref{c.216}
\begin{align}
{\hd}(1^{'}\wedge  a_{\sessions}, \nodes^{\ddagger} \setminus \{\phi\} ) = {\hd}(1^{'},  a_{\sessions} \wedge \nodes^{\ddagger} \setminus \{\phi\} ) =0 .
\end{align}
Furthermore, by the topology constraint
\begin{align}
\hd(d_{s,u}\mid  a_{\sessions} , \nodes^{\ddagger} \setminus \{\phi\}) = 0
\end{align}
Therefore, together with \eqref{d.243}, we have 
\begin{align}\label{d.250}
\hd\left(1^{'},w_{s,u}\wedge d_{s,u},a_{\sessions}\right) = 0.
\end{align}

Similarly, by \eqref{d.232}  and \eqref{e.237},
\begin{align}
{\hd}(b_{\sessions} \wedge a_{\sessions}, \nodes^{\ddagger} \setminus \{\phi\}) = {\hd}(a_{\sessions},b_{\sessions} \wedge  \nodes^{\ddagger} \setminus \{\phi\})  = 0
\end{align}
and hence 
\begin{align}
\hd\left(b_{\sessions} \wedge d_{s,u}, a_{\sessions}\right) =0. \label{d.253}
\end{align}
Consequently, we have 
 \begin{align}
{\hd} ( w_{s,u} \wedge b_{s}) = {\hd} ( w_{s,u} \wedge d_{s,u}) = {\hd} ( d_{s,u} \wedge b_{s}) = 0. \label{c.242}
\end{align}
On the other hand, by the topology constraint, 
\begin{align}\label{241}
\hd(w_{s,u} \mid  b_{s}, d_{s,u} , \theta_{s,u}) = 0.
\end{align}
Again, by \eqref{d.232}, $\hd\left(1^{'},b_{s}, d_{s,u} \wedge \theta_{s,u}\right) = 0$. Then, by \eqref{241}, we can prove that 
\begin{align}
\hd(w_{s,u} \mid  b_{s}, d_{s,u}) = 0. \label{d.256}
\end{align}

Using  \eqref{eq:7.129}, \eqref{c.242} and \eqref{d.256} and \eqref{c.222}  and \eqref{eq:7.127b}, we can prove that  
 \begin{align}
\hd(w_{s,u} \mid  b_{s}, d_{s,u})  = \hd(d_{s,u} \mid  b_{s}, w_{s,u})  = \hd(b_{s} \mid  w_{s,u} d_{s,u})  =0  \label{d.257}
\end{align}
and
\begin{align}
\hd(d_{s,u}) = \lambda(s).
\end{align}

Finally, notice that 
\begin{align}
\hd\left(b_{s}\wedge a_{s} d_{s,u}\right) & \nequal{(i)} 0 \label{e.259} \\
\hd\left(w_{s,u},c_{s}\wedge a_{s} d_{s,u}\right) & \nequal{(ii)}  0 \label{e.260} \\
\hd\left(b_{s}\wedge w_{s} c_{s}\right) & \nequal{(iii)} 0 \label{e.261}
\end{align}
where $(i)$, $(ii)$  and $(iii)$ follow respectively \eqref{d.253},  \eqref{d.250} and \eqref{c.217}.
By \eqref{d.242} and \eqref{d.257}, we have  
\begin{align}
\hd(b_{s} \mid  w_{s,u} c_{s}, a_{s} d_{s,u}) &= 0 \label{eq:7.136} \\
\hd(a_{s} d_{s,u} \mid  b_{s}, w_{s,u} c_{s}, ) &=  0.
\end{align}
Together with \eqref{e.259}--\eqref{e.261}, we can prove that 
\[
\hd(a_{s} d_{s,u}) = \hd(b_{s}) = \lambda(s) =\hd(d_{s,u}).
\]
 Thus, $\hd( a_{s} \mid  d_{s,u})  = 0 $
and 
\[
g(a_{s} \mid  \incoming(u) ) = \hd(a_{s} \mid  \incoming(u) ,\nodes^{\ddagger}\setminus \{\phi\} ) = 0.
\]
Thus,  $g \in \set{C}_{\de}(\problem)$ and the first claim is proved.

We will now prove the second claim. 
The idea of the proof is similar to that in the incremental multicast scenario. 
Suppose $(\lambda,\omega)$ is 0-achievable with respect to  
$\problem$. By definition, there exists a sequence of zero-error network codes 
\[
\{ Y^{n}_{f}, f\in\sessions\cup\edges \}
\]
for $\problem$, and a sequence of positive constants $c_{n}$ such that 
\begin{align}
\lim_{n\to\infty} c_{n} H(Y_{e}^{n}) \le \lim_{n\to\infty} c_{n} H|\support(Y_{e}^{n})| \le \omega(e) \\
\lim_{n\to\infty} c_{n} H(Y_{s}^{n}) = \lim_{n\to\infty} c_{n} H|\support(Y_{s}^{n})| \ge \lambda(s). 
\end{align}
%
%
%

Assume  without loss of generality that  $\support(Y^{n}_{s})$ (i.e., the support of $Y^{n}_{s}$) 
is equal to $\{0, \ldots,   |\support(Y^{n}_{s})|-1\}$. For each $n$,  construct the following set of random variables 
\[
\{ U_{f}, f\in\sessions^{\ddagger}\cup \edges^{\ddagger}\cup \nodes^{\ddagger}\}
\]
such that for all  $e\in\edges, s\in\sessions$ and $u\in\destinationLocation(s)$, 
\begin{enumerate}
\item $U^{n}_{a_{s}} = U^{n}_{d_{s,u}} = Y^{n}_{s}$;  

\item $U_{e}^{n} = Y^{n}_{e}$;

\item $U^{n}_{b_{s}} $ is uniformly distributed over $\support(Y^{n}_{s})$ for all $s\in\sessions$ and  that 
\[
H\left(U^{n}_{a_{\sessions}}, U^{n}_{\edges},  U^{n}_{b_{\sessions}}\right) = \sum_{f\in\sessions }U^{n}_{b_{f}} + H\left(U^{n}_{a_{\sessions}}, U^{n}_{\edges}\right);
\]
\item 
$U^{n}_{w_{s,u}} = U^{n}_{c_{s}} = Y^{n}_{a_{s}} + Y^{n}_{b_{s}} \mod  |\support(Y^{n}_{s})|$; 
\item $U^{n}_{e_{s}} \defined  (Y^{n}_{c_{i}}, i\in \sessions\setminus s  )$
\item $U^{n}_{1^{'}}= (Y^{n}_{c_{s}}, s\in\sessions)$;
\item 
$U^{n}_{v} = 1$ (i.e., $U_{v} = 1$ is a deterministic random variable) for all $v\in\nodes^{\ddagger}$;
\end{enumerate}

Again, it can be verified directly that 
\[
\{ U^{n}_{f}, f\in\sessions^{\ddagger}\cup \edges^{\ddagger}\cup \nodes^{\ddagger}\}
\]
is a  strongly secure zero-error network codes for $\problem^{\ddagger}$.
Consequently, $T(\lambda,\omega)$ is 0-achievable with respect to $P^{\ddagger}$, subject to strong secrecy constraint. 

\end{appendices}

\bibliographystyle{IEEEtran}

\end{document}